%% file: ms.tex
\documentclass[12pt]{article}           
\pdfoutput=1

\usepackage[utf8]{inputenc}
\usepackage[T1]{fontenc}

\usepackage{lmodern}
\usepackage{amsmath}
\usepackage[dvipsnames]{xcolor}         
\usepackage{natbib}                     
\usepackage{bm}                         
\usepackage{amsfonts}                   
\usepackage{amsthm}                     
\usepackage{thmtools}                   
\usepackage{mathtools}                  
\usepackage{booktabs}                   
\usepackage{caption}
\usepackage{subcaption}                 
\usepackage{adjustbox}                  
\usepackage{array}                      

\usepackage{tikz}                       
\usetikzlibrary{positioning,matrix,calc,3d,arrows.meta,intersections,decorations.pathreplacing,fit}

\usepackage{siunitx}
\usepackage{lscape}

\usepackage{algorithm}
\usepackage{algpseudocode}

\newcommand*\Let[2]{\State #1 $\gets$ #2}
\algnewcommand\algorithmicinitialize{\textbf{Initialize:}}
\algnewcommand\Initialize{\item[\algorithmicinitialize]}
\algnewcommand\algorithmicbreak{\textbf{break}}
\algnewcommand\Break{\State \algorithmicbreak}

\newcommand{\blind}{0}

\DeclareMathOperator{\Stiefel}{St}
\DeclareMathOperator{\Oblique}{Ob}
\DeclareMathOperator*{\argmin}{arg\, min}

\DeclareMathOperator{\Diag}{diag}

\DeclareMathOperator{\Sign}{sgn}
\DeclareMathOperator{\Rank}{rank}
\DeclareMathOperator{\Variance}{Var}

\newcommand*{\R}{\mathbb{R}}
\newcommand*{\vect}[1]{\bm{#1}}
\newcommand*{\mat}[1]{\bm{\mathrm{#1}}}
\newcommand*{\EV}{\mathbb{E}}

\DeclarePairedDelimiter\abs{\lvert}{\rvert}
\DeclarePairedDelimiter\norm{\lVert}{\rVert}

\declaretheorem{lemma, theorem, corollary}
\declaretheorem{definition, assumption}[%
    style=definition%
]
\declaretheorem{remark, example}[%
    style=remark,
    qed={\hbox{\(\circ\)}},%
]

\defcitealias{TCGA2012}{TCGA}

\definecolor{cred}{HTML}{E69F00}
\definecolor{cblue}{HTML}{56B4E9}
\definecolor{uogred}{HTML}{db5000}
\definecolor{uogblue}{HTML}{004b89}
\definecolor{uoggreen}{HTML}{009e47}

\usepackage{etoolbox}
\newrobustcmd\B{\color{uogblue}\DeclareFontSeriesDefault[rm]{bf}{b}\bfseries}
\newrobustcmd\E{\color{uogred}\DeclareFontSeriesDefault[rm]{bf}{b}\bfseries}
\newrobustcmd\BE{\color{uoggreen}\DeclareFontSeriesDefault[rm]{bf}{b}\bfseries}

\usepackage[colorlinks,allcolors=cblue]{hyperref}

\begin{document}

\def\spacingset#1{\renewcommand{\baselinestretch}%
{#1}\small\normalsize} \spacingset{1}


\if0\blind%
{

  \title{\bf Sparse and Orthogonal Low-rank Collective Matrix Factorization (solrCMF): Efficient data integration in flexible layouts}
  \author{Felix Held \\ 
    Department of Mathematical Sciences\\
    Chalmers University of Technology and University of Gothenburg\\
    \and 
    Jacob Lindbäck \\
    EECS, KTH Royal Institute of Technology \\
    \and 
    Rebecka Jörnsten \\
    Department of Mathematical Sciences\\
    Chalmers University of Technology and University of Gothenburg}
  \date{}
  \maketitle
} \fi

\if1\blind%
{
  \bigskip
  \bigskip
  \bigskip
  \begin{center}
    {\LARGE\bf Sparse and Orthogonal Low-rank Collective Matrix Factorization (solrCMF): Efficient data integration in flexible layouts}
  \end{center}
  \medskip
} \fi


\clearpage
\begin{abstract}

Interest in unsupervised methods for joint analysis of heterogeneous data sources has risen in recent years. Low-rank latent factor models have proven to be an effective tool for data integration and have been extended to a large number of data source layouts. Of particular interest is the separation of variation present in data sources into shared and individual subspaces. In addition, interpretability of estimated latent factors is crucial to further understanding.

We present sparse and orthogonal low-rank Collective Matrix Factorization (solrCMF) to estimate low-rank latent factor models for flexible data layouts. These encompass traditional multi-view (one group, multiple data types) and multi-grid (multiple groups, multiple data types) layouts, as well as augmented layouts, which allow the inclusion of side information between data types or groups. In addition, solrCMF allows tensor-like layouts (repeated layers), estimates interpretable factors, and determines variation structure among factors and data sources.

Using a penalized optimization approach, we automatically separate variability into the globally and partially shared as well as individual components and estimate sparse representations of factors. To further increase interpretability of factors, we enforce orthogonality between them. Estimation is performed efficiently in a recent multi-block ADMM framework which we adapted to support embedded manifold constraints.

The performance of solrCMF is demonstrated in simulation studies and compares favorably to existing methods.
An implementation of our algorithm is available as a Python package at \url{https://github.com/cyianor/solrcmf}.

\end{abstract}

\noindent%

\newpage


\section{Introduction}%
\label{sec:intro}


Unsupervised, joint analyses of heterogeneous data sources are increasingly important in many research domains.
These modeling techniques aim to identify relevant shared signals that coexist across subsets of data sources while flexibly accounting for individual variability stemming from data heterogeneity.
For example, in biostatistics, the increasing availability of multi-omics data collected on multiple groups of individuals (patients, cells) \citepalias[e.g.,][2012]{TCGA2012} poses a huge challenge on how to integrate data from different -omics sources, such as integration of genomic, epigenomic, transcriptomic, and proteomic data, while allowing for their inherent differences. The recent leaps made in single-cell sequencing also require integrative methods to combine data collected on multiple groups of cells or to integrate different data types collected on the same cells \citep{Argelaguet2021,Ryu2023}.
Similarly, large amounts of available user feedback, such as ratings and associated side information, are used to train recommender systems which employ data integration to better predict user preferences \citep[e.g.,][]{SinghGordon2008,Ning2012,Chen2018b}.

A common approach to integrating data from different sources is to partition the variance of each data source by identifying shared subspaces of variability between a subset of input data sources.
A useful model choice in this context is to assume that any \(n \times p\) data matrix \(\mat{X}\) can be written as \(\mat{X} = \mat{Z} + \mat{E}\), where \(\mat{Z}\) is a \(n \times p\) matrix of rank \(k \ll \min(n, p)\), a \emph{low-rank signal}, and \(\mat{E}\) is a matrix containing \emph{additive residual noise}.
A standard assumption on \(\mat{E}\) is that entries are random, uncorrelated, have mean 0, and variance \(\sigma^2\).
This type of model has been found to hold approximately in many real datasets and is therefore applicable in many settings (see e.g., \citet{Udell2019} for more details).

Finding directions of shared variance in two data sources collected on the same samples but with different observed variables has been investigated at least since the introduction of Canonical Correlation Analysis \citep[CCA]{Hotelling1936}. The main idea being to find linear combinations within each block of variables that correlate maximally. 
CCA was eventually extended to allow for two or more matrices as generalized CCA \citep[GCCA]{Horst1961,Kettenring1971}.
Within chemometrics, orthogonal projections to latent structures \citep[O-PLS]{OPLS} is another popular technique to perform data integration and find shared variation between two datasets collected on the same samples. It can be shown that CCA and O-PLS are equivalent \citep{Sun2009}. The method was then improved \citep[O2-PLS]{O2PLS} to account appropriately for individual variation and to improve the estimate of shared variation. O2-PLS, in turn, was then extended to work with two or more matrices as well \citep{OnPLS}.

Inter-battery factor analysis \citep[IBFA]{Tucker1958} finds common factors that explain variation in the data within both groups simultaneously. This is different from CCA, since instead of finding maximally correlated but separate directions of variance within each data source, IBFA finds exactly correlated factors present in both data sources simultaneously. In that way, IBFA partitions the variation in the data into a shared subspace described by the shared factors and two individual spaces, which describe variation not described by the common factors for each data source individually.

JIVE \citep{Lock2013} adopts the factor analysis point of view and separates two or more data matrices matching along samples into globally shared (joint) and individual variation by modeling data matrices as
\begin{equation*}
\mat{X}_i =
\mat{C}_i + \mat{D}_i + \mat{E}_i =
\mat{U} \mat{V}_i^\top
+ \mat{U}_i \mat{W}_i^\top
+ \mat{E}_i,
\end{equation*}
where \(\mat{C}_i\) represents the projection of the signal onto the joint space of variation and \(\mat{D}_i\) represents the remaining, individual, part of the signal. \(\mat{U}\) and \(\mat{U}_i\) are joint and individual factor matrices, respectively, whereas \(\mat{V}_i\) and \(\mat{W}_i\) are factor loadings.
In particular, it is assumed that \(\mat{U}^\top \mat{U}_i = \mat{0}\) to ensure that joint and individual factor patterns are unrelated.
This idea is then expanded on in SLIDE \citep{Gaynanova2019} to allow not just for globally shared signal, but also for uncovering variation shared by a subset of data matrices. They adopt the model
\begin{equation*}
\begin{pmatrix}
\mat{X}_1 \\
\vdots \\
\mat{X}_n \\
\end{pmatrix} =
\begin{pmatrix}
\mat{U}_1 \\
\vdots \\
\mat{U}_n \\
\end{pmatrix}
\mat{V}(\mat{S})^\top
+ \begin{pmatrix}
\mat{E}_1 \\
\vdots \\
\mat{E}_n
\end{pmatrix} = \mat{U} \mat{V}(\mat{S})^\top + \mat{E}
\end{equation*}
where factors \(\mat{U}\) are orthogonal and \(\mat{S}\) describes the 0-1 structure pattern of the estimate. If \(\mat{S}^{(i, j)} = 1\) then factor \(j\) is active in signal matrix \(i\) and loadings are non-zero, otherwise they are zero. This way globally shared, partially shared, and individual factors can be found among two or more matrices.


Group factor analysis \citep[GFA]{Klami2015}, like JIVE, is a generalization of IBFA to two or more matrices.
In contrast to all previous models, GFA has been formulated as a Bayesian model.
It uses the model \(\mat{X}_i = \mat{Z} \mat{W}_i + \mat{E}_i\) where \(\mat{Z}\) contains factors and is modeled essentially as regression coefficients using a normal prior, and \(\mat{W}_i\) contain the loadings for each data matrix. 
Active factors within each data matrix are selected with sparsity-inducing priors on the loading matrices, assuming a zero-mean normal prior with factor specific precision parameters. By assuming a prior distribution over the precision parameters an Automatic Relevance Detection (ARD) prior \citep{Mackay1995} is defined.
In case of GFA, precisions are modelled to follow a second low-rank linear model of even lower dimension, ultimately placing priors on the matrix factors in the linear model. GFA therefore builds a hierarchical model where low-rank models appear in two different places.

MOFA(+) \citep{Argelaguet2018,Argelaguet2020} as well as Collective Matrix Factorization \citep[CMF]{Klami2014} are variants of GFA which both drop the linear model for loading precisions in favour of a simpler Gamma prior.
MOFA adds spike-and-slab priors \citep{Mitchell1988} on factors and/or loadings to perform variable selection. Interpretability of factors is an important aspect during explorative data analysis and reducing the number of entries with any substantial impact in factors or loadings helps.
MOFA+ extends the MOFA model defined for multiple observed groups of variables contained in one group of samples to the case of a multi-group multi-block scenario.
CMF extends the applicability of GFA from the multi-block case to augmented data matrix layouts, allowing side information that relates two groups of variables or two groups of samples with each other. In addition, the traditional separation between variables and samples is given up to allow for a flexible model. The model used by CMF is \(\mat{X}_{ij} = \mat{U}_i \mat{U}_j + \mat{E}_{ij}\) where \((i, j)\) indexes row and column entities (which might be samples or variables).
MMPCA \citep{Kallus2019} then builds on the CMF model, now again within an optimization framework. It adds the assumption of orthogonality on factor matrices \(\mat{U}_i\) and introduces element-wise sparsity in factors, similar to MOFA, however, using a penalized optimization approach instead.

We aim to keep the orthogonality assumption on factor matrices and will work within an optimization framework. However, having to formulate optimization procedures over the space of matrices with orthogonal columns, often called the Stiefel manifold, can be challenging \citep{Edelman1998}. In \citet{Kallus2019} an explicit parametrization of orthogonal matrices using angles was attempted. However, this requires solving over a large parameter space which is often slow and prone to get stuck in local minima.

Splitting methods, such as proximal gradient descent \citep{Parikh2014} and ADMM \citep{Boyd2011}, have become standard methods for solving large-scale problems with non-smooth components, such as sparsity inducing penalties. Many recent works have adapted these algorithms to specific classes of non-convex problems \citep[e.g.,]{wang2019global, li2015accelerated}.
For matrix factorization problems, manifold optimization techniques have gained a great deal of attention due to their ability to handle orthogonality constraints and rank constraints efficiently \citep{Absil2008}. However, efficiently using manifold optimization to solve problems has proven more difficult.
For instance, \citep{huang2022riemannian} adapted proximal gradient descent to the manifold setting. Their approach hinges on solving a non-smooth, convex optimization problem in every iteration to find a search direction, which is potentially expensive.


Here, we are focusing on multi-block ADMM with multi-affine constraints \citep{Gao2020}. This class of algorithms allows to move some non-convex parts in the objective function, such as the product \(\mat{U} \mat{V}^\top\) in matrix factorization, into constraints. If constructed correctly, the objective function can be made convex while retaining a provably convergent algorithm. However, manifold constraints, such as the constraint on factor matrices to have orthogonal columns, were not accounted for in \citet{Gao2020}.


\subsection{Our contributions}

In this paper, we present \emph{sparse and orthogonal low-rank Collective Matrix Factorization (solrCMF)}, a flexible group of low-rank models readily amenable to data integration.
\begin{enumerate}
\item We extend the model presented in CMF \citep{Klami2014} and MMPCA \citep{Kallus2019} to a richer class of data collections. In addition, we discuss interpretation of the model and investigate its identifiability. (Section~\ref{sec:data-integration-model})
\item This model is then used to formulate a penalized optimization approach to perform data integration, automatically finding a suitable partition of variation in the data into shared (possibly among subsets) and individual subspaces, and capture elementwise factor sparsity. (Section~\ref{sec:data-integration-model})
\item An extended formulation of the solrCMF problem is formulated to allow efficient parameter estimation and a detailed analysis of its convergence is provided. In particular we extend the approach in \citet{Gao2020} to allow embedded manifold constraints, such as orthogonality constraints. (Section~\ref{sec:m-admm})
\item Finally, we perform simulation studies to evaluate solrCMF in comparison to other methods.
Particular focus is on structure estimation, correctly estimating factor sparsity, and the impact of sparsity and side information on factor identifiability.
(Section~\ref{sec:simulation})
\end{enumerate}

An implementation of our algorithm is available as a Python package at \url{https://github.com/cyianor/solrcmf}.

\section{Data integration model}%
\label{sec:data-integration-model}

We study data collections describing measurements on a set of \emph{views}. In this context, views can be both cohort groups (samples) as well as feature groups (e.g., gene expression or user attributes).
For simplicity of notation, we assume that views are numbered from 1 to \(m\) and that a data matrix \(\mat{X}_{ij}\) contains data relating view \(i\) in its rows and view \(j\) in its columns. In addition, we allow the observation of replicates, which can be denoted as \(\mat{X}_{ij, \gamma}\). For clarity, we will omit index \(\gamma\) unless it is important in the context. When multiple matrices are observed which match in one or both views, then \emph{layouts} of data matrices are formed.
See Figure~\ref{fig:data-layouts} for examples of layouts that the method described in this paper can take as an input. Note that these layouts can be arbitrarily combined and therefore very \emph{flexible layouts} can be created.
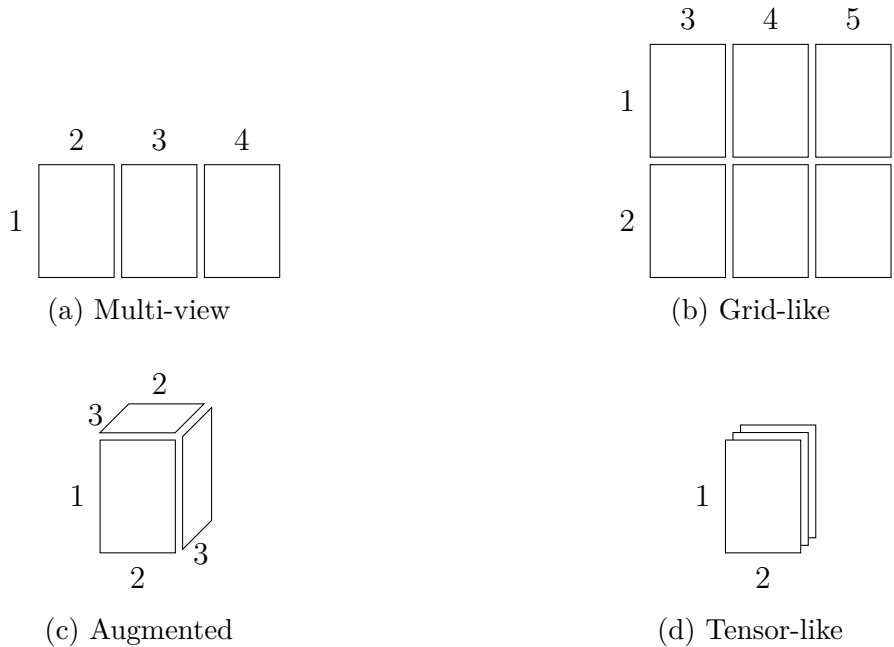
\begin{figure}[hbt]
\begin{subfigure}[t]{.5\textwidth}
\centering
\begin{tikzpicture}
\node[rectangle,draw,minimum width=1cm, minimum height=1.5cm,fill=white] (DAR1) {};
\node[rectangle,draw,minimum width=1cm, minimum height=1.5cm,fill=white] (DAR2) at ($(DAR1)+(1.1cm,0)$) {};
\node[rectangle,draw,minimum width=1cm, minimum height=1.5cm,fill=white] (DAR3) at ($(DAR2)+(1.1cm,0)$) {};

\node[left=0.05cm of DAR1] (DAV1) {1};
\node[above=0.05cm of DAR1] (DAV2) {2};
\node[above=0.05cm of DAR2] (DAV3) {3};
\node[above=0.05cm of DAR3] (DAV4) {4};
\end{tikzpicture}
\caption{Multi-view}
\label{fig:data-layouts:multi-view}
\end{subfigure}%
\begin{subfigure}[t]{.5\textwidth}
\centering
\begin{tikzpicture}
\node[rectangle,draw,minimum width=1cm, minimum height=1.5cm,fill=white] (DBR13) {};
\node[rectangle,draw,minimum width=1cm, minimum height=1.5cm,fill=white] (DBR14) at ($(DBR13)+(1.1cm,0cm)$) {};
\node[rectangle,draw,minimum width=1cm, minimum height=1.5cm,fill=white] (DBR15) at ($(DBR14)+(1.1cm,0)$) {};

\node[rectangle,draw,minimum width=1cm, minimum height=1.5cm,fill=white] (DBR23) at ($(DBR13)+(0cm,-1.6cm)$) {};
\node[rectangle,draw,minimum width=1cm, minimum height=1.5cm,fill=white] (DBR24) at ($(DBR23)+(1.1cm,0)$) {};
\node[rectangle,draw,minimum width=1cm, minimum height=1.5cm,fill=white] (DBR25) at ($(DBR24)+(1.1cm,0)$) {};

\node[left=0.05cm of DBR13] (DBV1) {1};
\node[left=0.05cm of DBR23] (DBV2) {2};
\node[above=0.05cm of DBR13] (DBV3) {3};
\node[above=0.05cm of DBR14] (DBV4) {4};
\node[above=0.05cm of DBR15] (DBV5) {5};
\end{tikzpicture}
\caption{Grid-like}
\label{fig:data-layouts:grid-like}
\end{subfigure}
\par\bigskip
\begin{subfigure}[t]{.5\textwidth}
\centering
\begin{tikzpicture}
\node[rectangle, draw, minimum width=1cm, minimum height=1.5cm, fill=white]
    (DCR1) {};
\begin{scope}[shift={($(DCR1) + (0.5cm,0cm)$)}, canvas is yz plane at x=0]
    \node[rectangle, draw, minimum width=1.5cm, minimum height=1cm, fill=white, transform shape] (DCR2) at (-0.05, -0.75) {};
    
\end{scope}
\begin{scope}[shift={($(DCR1) + (0cm,0.75cm)$)}, canvas is xz plane at y=0]
    \node[rectangle, draw, minimum width=1cm, minimum height=1cm, fill=white, transform shape] (DCR3) at (-0.1, -0.75) {};
\end{scope}

\node[left=0.05cm of DCR1] (DCV1) {1};
\node[below=0.05cm of DCR1] (DCV21) {2};
\node (DCV31) at ($(DCR2) + (0.05cm,-1cm)$) {3};
\node (DCV22) at ($(DCR3) + (0.1cm,0.45cm)$) {2};
\node (DCV32) at ($(DCR3) + (-0.75cm,0.05cm)$) {3};
\end{tikzpicture}
\caption{Augmented}
\label{fig:data-layouts:augmented}
\end{subfigure}%
\begin{subfigure}[t]{.5\textwidth}
\centering
\begin{tikzpicture}
\node[rectangle,draw,minimum width=1cm, minimum height=1.5cm,fill=white] (DDR3) {};
\node[rectangle,draw,minimum width=1cm, minimum height=1.5cm,fill=white] (DDR2) at ($(DDR3)+(-0.1cm,-0.1cm)$) {};
\node[rectangle,draw,minimum width=1cm, minimum height=1.5cm,fill=white] (DDR1) at ($(DDR2)+(-0.1cm,-0.1cm)$) {};

\node[left=0.05cm of DDR1] (DDV1) {1};
\node[below=0.05cm of DDR1] (DDV2) {2};
\end{tikzpicture}
\caption{Tensor-like}
\label{fig:data-layouts:tensor-like}
\end{subfigure}
\caption{Illustrations of possible layouts supported by our method. Rectangles correspond to data matrices and numbers represent views.}
\label{fig:data-layouts}
\end{figure}

\citet{Klami2014} and \citet{Kallus2019} describe two data integration methods that support arbitrary augmented layouts and assume that the low-rank signal contained in \(\mat{X}_{ij}\) can be described by an outer product of two \(k\) view-specific factors \(\mat{U}_i\) and \(\mat{U}_j\), respectively, resulting in the model
\begin{equation}\label{eq:cmf-model}
\mat{X}_{ij} = \mat{U}_i \mat{U}_j^\top + \mat{E}_{ij}.
\end{equation}
\citet{Kallus2019} make the additional assumption that \(\mat{U}_i\) and \(\mat{U}_j\) have orthogonal columns.
Despite the similarity of the model above to Factor Analysis \citep[Chapter 12]{Hardle2019}, no separation is made between loadings and factors. This is to facilitate augmented layouts (Figure~\ref{fig:data-layouts:augmented}).
The scale of factor \(l\) in the signal of \(\mat{X}_{ij}\) is equal to \(\norm{\mat{U}_i^{(:, l)}}_2 \norm{\mat{U}_j^{(:, l)}}_2\), where \(\mat{U}_i^{(:, l)}\) is the \(l\)-th column of \(\mat{U}_i\). Factor scale is a proxy for signal strength, with larger factor scales resulting in more prominent factors.
Ideally, despite reusing factors, it would be desirable if factor scale could be chosen separately for each data matrix.
A shortcoming of the model in Eq.~\eqref{eq:cmf-model} is that factor scales can end up in a deadlock situation. In \citet{Klami2014} this can be resolved by duplicating a factor and using it to describe separate aspects of the integrated data sources. However, if factors are required to be orthogonal within a view the consequences are more detrimental. This can limit the range of data sources that can jointly be described by the model.

To illustrate this issue, assume a \(2 \times 2\) grid layout was observed. Let \(\vect{u}_i = s_i \vect{v}_i\) with \(\norm{\vect{v}_i}_2 = 1\) and \(d_i > 0\) for \(i = 1, 2, 3,\) and \(4\).
The data is assumed to follow the model in Eq.~\eqref{eq:cmf-model} such that
\(\mat{X}_{ij} = \vect{w}_i \vect{w}_j^\top + \mat{E}_{ij}\) for
\((i, j) \in \{(1, 3), (1, 4), (2, 3), (2, 4)\}\).
The factor scale of the signal in \(\mat{X}_{ij}\) is therefore \(d_{ij} = s_i s_j\) and the \(2 \times 2\) grid-layout then implies the relationship \(d_{13} = (d_{23} d_{14}) / d_{24}\). One factor scale is therefore completely determined by the three others. Factor scale of different factors relative to each other within a data matrix is of importance. However, the scale of a factor in one data matrix relative to its scale in another data matrix is essentially arbitrary since it can easily change by scaling one matrix. It is therefore crucial to make factor scale a data matrix-specific property instead of a view-specific property.

In this paper, we improve on the model in Eq.~\eqref{eq:cmf-model} by separating factor scales from factors and assuming that
\begin{equation}\label{eq:data-model}
\mat{X}_{ij} = \mat{V}_i \mat{D}_{ij} \mat{V}_j^\top + \mat{E}_{ij}
\end{equation}
where \(\mat{V}_i\) and \(\mat{V}_j\) contain \(k\) orthogonal factors each, \(\mat{D}_{ij}\) contains the factor scales, and, as above, \(\mat{E}_{ij}\) contains uncorrelated additive noise with mean 0 and variance \(\sigma_{ij}^2\).
We will typically assume that \(\mat{X}_{ij}\) is bicentered, i.e. centered in rows as well as columns, due to the exchangable meaning of row and column views.
In addition, normalization such that \(\norm{\mat{X}_{ij}}_F = 1\) is recommended. In particular we assume that, when normalized, matrices \(\mat{X}_{ij}\) have a comparable residual variance. This is an assumption that is often made implicitly \citep[e.g.,][]{Lock2013,Gaynanova2019,Kallus2019}.

The signal \(\mat{Z}_{ij} = \mat{V}_i \mat{D}_{ij} \mat{V}_j^\top\) has many similarities with Singular Value Decomposition (SVD) \citep[Chapter 2.4]{Golub2013} and therefore entries on the diagonal of \(\mat{D}_{ij}\) will be called \emph{singular values}. In contrast to SVD for a single matrix, the diagonal entries of \(\mat{D}_{ij}\) are not ordered in a non-increasing fashion and neither is non-negativity enforced. In particular, any entry on the diagonal of \(\mat{D}_{ij}\) can be zero, indicating that there is no variation along this factor in \(\mat{Z}_{ij}\). A discussion of these properties is given in Section~\ref{ssec:interpreting-singular-values}. 

Both models in Eq.~\eqref{eq:cmf-model} and Eq.~\eqref{eq:data-model} are capable of representing multi-view, grid-like, and augmented layouts (Figures~\ref{fig:data-layouts}a-c). However, when using the model with view-specific factor scaling in Eq.~\eqref{eq:cmf-model} for replicates of data matrices \(\mat{X}_{ij, \gamma}\), only trivial replicates can be described with fluctuation in the additive noise but constant signal across layers.
Our model in Eq.~\eqref{eq:data-model} can describe more flexible, tensor-like, signals (Figure~\ref{fig:data-layouts:tensor-like}).
Factors in \(\mat{V}_i\) and \(\mat{V}_j\) are kept constant across tensor layers. Singular values \(\mat{D}_{ij, \gamma}\), however, are specific to each tensor layer. Factors can therefore have varying scales across layers and factors describing the variation in the signal can vary since some singular values can be zero.

When performing data integration it is desirable to be able to split the signal \(\mat{Z}_{ij}\) into \emph{shared} and \emph{individual} signal. Conceptually, row and column spaces of signals are partitioned appropriately.
\begin{enumerate}
\item \textbf{Shared signal:} Two or more signal matrices contain a shared subspace of dimension 1 or larger. Depending on whether signals match along rows or columns, these are subspaces within the corresponding row or column spaces.
\item \textbf{Individual signal:} Subspaces of the row and column spaces of a signal matrix that are not shared.
\end{enumerate}
Signal can therefore be present globally in all integrated signal matrices, partially shared between a subset of matrices, or individual to a single matrix.
The model in Eq.~\eqref{eq:data-model} encodes information about shared and individual signal into factors and singular values. Given two signal matrices \(\mat{Z}_{ij}\) and \(\mat{Z}_{ik}\), if \(\mat{D}_{ij}^{(l, l)} \neq 0 \neq \mat{D}_{ik}^{(l, l)}\), then the \(l\)-th factor in view \(i\) contains variation that appears in both signal matrices and describes a 1-dimensional shared signal subspace. Clearly, more than one factor can be shared and different subsets of factors can be shared between different subsets of signal matrices.


\subsection{Interpreting singular values}%
\label{ssec:interpreting-singular-values}




Historically, singular values were introduced as square roots of eigenvalues and therefore non-negative. If SVD is defined via \(\mat{X} = \mat{U} \mat{D} \mat{V}^\top\), there is no immediate reason as to why singular values should be non-negative. However, to ensure uniqueness of the result when singular values are all distinct, non-negativity and a non-increasing order of the singular values is imposed. This convention ensures that singular values keep their interpretation as roots of eigenvalues and from a statistical standpoint makes it more natural to interpret singular values as scaled standard deviations in Principal Component Analysis (PCA) \citep[Chapter 11]{Hardle2019}.
From a technical perspective, if a singular value were estimated negative, then the singular vector in \(\mat{V}\) corresponding to this singular value can be elementwise sign-flipped to compensate.

In comparison to SVD and PCA, the model in Eq.~\eqref{eq:data-model} aims to capture a much broader type of data sets, and therefore enforcing non-negativity of singular values is not possible anymore.
A simple example as to why this is not possible is given in the following.
Imagine a layout (1, 2), (1, 3), and (2, 3) and assume that all three signals share a single factor. If the singular value for (2, 3) is negative and positive for (1, 2) and (1, 3), enforcing a positive sign for all matrices would require flipping the sign of the factor in either view 2 or 3. This in turn would require flipping the sign of the corresponding factor in view 1 to ensure that the singular value of (1, 2) or (1, 3) does not become negative. However, only one of view 2 or 3 had its sign flipped, so this leads to the sign for either (1, 2) or (1, 3) being negative. It is therefore not possible to ensure non-negativity of all singular values.
The inability to always restrict signs of singular values to only non-negative values is therefore a consequence of complex layouts between views, see e.g.\@ Figure~\ref{fig:data-layouts:augmented}. However, simpler layouts like the multi-view layout in Figure~\ref{fig:data-layouts:multi-view} always allow selecting non-negative signs of singular values.

This poses the question of how singular values in our model should be interpreted statistically. When restricting the solution to a single data matrix, singular values can always be chosen non-negative as is the case for standard SVD. It is therefore reasonable to continue to interpret factors as directions of variation in the data and absolute values of singular values as the magnitude of this variation within each data matrix.

\subsection{Interpreting factors}%
\label{ssec:interpreting-factors}

Interpretability of factors is important during exploratory data analysis. As is common in PCA, data can be projected onto two or three factors and visualized.
However, factor entries might be of interest on their own. If one view represents a list of genes, it can be insightful to investigate genes that correspond to values large in absolute value within a factor, since they contribute most to the variation described by the factor. 
Another example is drug effect where one view represents drugs and the other samples. Factors describe variation within the space of drug effect and drugs strongly affecting it might be interesting to investigate.
To improve the interpretability of factors, it can be helpful to force some or most factor entries to zero. A factor containing a lot of zeros is called \emph{sparse}.
Sparsity is well-known to increase interpretability of coefficient vectors in regression problems \citep{Tibshirani2002} and has also been shown to decrease bias in estimated principal components in high dimensions \citep{Johnstone2009}.
It is therefore desirable to allow for ways to induce sparsity in factor estimates.
In the context of data integration, this has been previously attempted using spike-and-slab priors on factors in a Bayesian setting \citep{Bunte2016,Argelaguet2018} or using sparsity-inducing penalties on factor entries during parameter optimization \citep{Kallus2019}, where the latter is reminiscent of approaches to sparse PCA \citep{Jolliffe2003,Witten2009}. 

\subsection{Missing data}

Data is often not fully observed. This may be due to, e.g., technical limitations, patient dropout, sensor failure or high cost associated with collecting data. Low-rank models have been extensively used to perform data imputation, i.e., the prediction of missing values \citep{Candes2010,Mazumder2010,Hastie2015}. In addition, data imputation is an important motivation for data integration. Side information coming from other data sources can help in the recovery of entries \citep{Klami2014}.


\subsection{Notation}

We introduce some necessary notation before continuing. The set of all views is denoted \(\mathcal{V} = \{1, \dots, m\}\) and each view is associated with a dimension \(p_i\); indices of all observed data matrices are collected in \(\mathcal{I} = \{ (i, j, \gamma) : \mat{X}_{ij,\gamma} \text{ is observed}\}\);
vector and matrix indexing is denoted in Matlab/NumPy-like notation (e.g., \(\mat{A}^{(i, j)}\) is the \((i, j)\)-th element in \(\mat{A}\) and \(\mat{U}^{(:, l)}\) denotes the \(l\)-th column of \(\mat{U}\));
the indicator function of a set \(\mathcal{S}\) is defined as \(\iota_{\mathcal{S}}(x) = 0\) for \(x \in \mathcal{S}\), and \(\iota_{\mathcal{S}}(x) = \infty\) otherwise; the operator \(\mathcal{P}_\mathrm{\Omega} : \R^{n \times p} \rightarrow \R^{n \times p}\) keeps all elements in an input matrix \(\mat{A}\) with indices in set \(\mathrm{\Omega} \subset \{(i, j): 1 \leq i \leq n, 1 \leq j \leq p\}\) and sets all remaining elements to 0;
the soft-thresholding operator is defined as \(\mathrm{ST}_\beta(x) = \Sign(x) \max(\abs{x} - \beta, 0)\);
the Stiefel manifold in \(\R^{p \times k}\) with \(k \leq p\) is the set \(\Stiefel_{k,p} = \{\mat{A} \in \R^{p \times k} : \mat{A}^\top \mat{A} = \mat{I}_k\}\);
the oblique manifold in \(\R^{p \times k}\) with \(k \leq p\) is the set \(\Oblique_{k, p} =
\{
\mat{A} \in \R^{p \times k} : \norm{\mat{A}^{(:, j)}}_2 = 1 \text{ for all } j
\}\);
the convex set of diagonal matrices in \(\R^{k \times k}\) is denoted as \(\mathcal{D}_k\);
the usual Euclidean norm of a vector \(\vect{a} \in \R^n\) is defined as \(\norm{\vect{a}}_2^2 = \sum_i {\vect{a}^{(i)}}^2\);
the Frobenius norm of a matrix \(\mat{A} \in \R^{n \times p}\) is analogously defined as \(\norm{\mat{A}}_F^2 = \sum_{i, j} {\mat{A}^{(i, j)}}^2\), and the corresponding \(\ell_1\)-norm is defined as \(\norm{\mat{A}}_1 = \sum_{i, j} \abs{\mat{A}^{(i, j)}}\).

\subsection{Signal estimation}

To estimate model parameters an optimization perspective will be adopted.
It is well-known that the SVD of a matrix truncated at \(k\) components produces the best rank-\(k\) estimate of said matrix in Frobenius norm \citep{Eckart1936,Mirsky1960}.
Given data matrices \(\mat{X}_{ij}\) for \((i, j) \in \mathcal{I}\) and index sets \(\mathrm{\Omega}_{ij}\) of observed entries, it is therefore reasonable to attempt to estimate the signal matrices \(\mat{Z}_{ij} = \mat{V}_i \mat{D}_{ij} \mat{V}_j^\top\) 
by solving the following optimization problem
\begin{equation*}
\argmin_{\mat{V}_i \in \Stiefel_{k, p_i}, \mat{D}_{ij} \in \mathcal{D}_k}
\frac{1}{2} \sum_{(i, j) \in \mathcal{I}}
\norm*{%
\mathcal{P}_{\mathrm{\Omega}_{ij}}\left(\mat{X}_{ij}
- \mat{V}_i \mat{D}_{ij} \mat{V}_j^\top\right)
}_F^2.
\end{equation*}
Note that \(k\) is typically unknown and has to be chosen appropriately. By using the operator \(\mathcal{P}_{\mathrm{\Omega}_{ij}}\) to select only observed entries it is straight-forward to account for missing values.
As mentioned above, the additional index \(\gamma\) for replicates is omitted to simplify notation, but all of our results hold true for \((i, j, \gamma) \in \mathcal{I}\). 


To perform the data integration and to select which factors are active we follow \citet{Kallus2019} and impose \(\ell_1\)-penalties on the diagonal entries of each \(\mat{D}_{ij}\). 
In addition, to find structure in factor matrices \(\mat{V}_i\) for \(i \in \mathcal{V}\) as discussed in Section~\ref{ssec:interpreting-factors} sparsity in factor entries is enforced. Similarly to \citet{Witten2009}, we impose \(\ell_1\)-penalties on the entries of each \(\mat{V}_i\).
The final optimization problem, solving the solrCMF problem, is then
\begin{equation}\label{eq:sparse-cmf:penalized-model}
\begin{aligned}
\argmin_{\mat{V}_l \in \Stiefel_{k, p_l}, \mat{D}_{ij} \in \mathcal{D}_k}
&\frac{1}{2} \sum_{(i, j) \in \mathcal{I}}
\norm*{%
\mathcal{P}_{\mathrm{\Omega}_{ij}}\left(\mat{X}_{ij}
- \mat{V}_i \mat{D}_{ij} \mat{V}_j^\top\right)
}_F^2 \\
+\ &\lambda_1 \sum_{(i, j) \in \mathcal{I}} 
\sum_{l = 1}^k \vect{w}_{ij}^{(l)}
\abs{\mat{D}_{ij}^{(l, l)}}
+ \lambda_2 \sum_{l = 1}^m \vect{w}^{(l)} \norm{\mat{V}_l}_1,
\end{aligned}
\end{equation}
where \(\lambda_1 > 0\) and \(\lambda_2 > 0\) are hyperparameters, \(\bm{w}_{ij}\) and \(\bm{w}\) are weight vectors that can be used to account for varying magnitudes of the estimated signal strengths as well as sizes of the matrices in the penalties.
We choose \(\bm{w}^{(l)} = 1 / \sqrt{p_l}\) to account for smaller entries in matrices of larger dimension.
For \(\bm{w}_{ij}\) it is possible to use ideas from the adaptive Lasso \citep{Zou2006b}. After an initial unpenalized run, resulting in estimates \(\widehat{\mat{D}}_{ij}\) the weights can be set to e.g.\@ \(\bm{w}_{ij} = 1 / \Diag(\widehat{\mat{D}}_{ij})^2\), with exponentiation and division being performed elementwise. The goal is to avoid unnecessary \emph{shrinkage} in non-zero entries, which is achieved by decreasing penalization of strong signals. We do not make use of \(\bm{w}_{ij}\) in our implementation but perform a two-step estimation procedure to avoid shrinkage instead. See Section~\ref{ssec:hyperparameter-selection} for details.


\subsection{Model identification}%
\label{ssec:model-identification}

Estimates of model parameters \(\mat{V}_l\) and \(\bm{D}_{ij}\) depend on initial values, due to the non-convexity of the solrCMF objective function in Eq.~\eqref{eq:sparse-cmf:penalized-model}. It is therefore necessary to discuss equivalent solutions and possible identifiability issues. First, note that the order of factors in a solution is arbitrary. Factors can be re-ordered as long as all corresponding values in all \(\mat{V}_l\) and all \(\mat{D}_{ij}\) are re-ordered. Ordering the singular values in a decreasing fashion like in SVD is not possible since factors are not guaranteed to contribute in the same order of magnitude to each data matrix.

In addition, signs of singular values can be flipped by sign-flipping the entries of corresponding factors, as has been discussed in Section~\ref{ssec:interpreting-singular-values}.
This makes signs of singular values often arbitrary. Some setups, as exemplified earlier, enforce a relationship between signs of different singular values but a chain of signs thus defined is equally valid when flipped. 
This implies for our model that solutions with reversed or partially reversed signs have to be seen as equivalent if the view layout allows the corresponding sign flips.

It is well-known that factor analysis models are invariant under rotation since for any rotation matrix \(\mat{R}\) it holds that \(\mat{U} \mat{V}^\top = \mat{U}\mat{R} \mat{R}^\top \mat{V}^\top\). The SVD-like formulation is not susceptible to this invariance if all singular values are distinct.
However, a well-known identifiability issue in SVD is subspace rotation. If \(n \geq 2\) singular values are equal, the corresponding singular vectors form a subspace of dimension \(n\) and any orthogonal basis of this subspace will serve as singular vectors for these \(n\) singular values. In other terms, the \(n\) dimensional subspace is rotationally invariant. 
solrCMF inherits the problem of unidentifability under subspace rotations.
Due to the presence of noise, subspace rotations can occur even for singular values that are only close but not equal to each other.
However, data integration in addition to the assumption of orthogonality among factors can help with subspace rotations. For example, assume \(\mat{D}_{12} = \Diag(2, 2)\) and \(\mat{D}_{13} = \Diag(0, 3)\). On their own, the factors making up the signal in \(\mat{X}_{12}\) are unidentifiable. By identifying the second factor in \(\mat{X}_{13}\), however, orthogonality then ensures that the first factor \(\mat{X}_{12}\) becomes identifiable as well.
In addition, sparsity in factors can guide the identification of the factors even without side information.



\subsection{Quantification of variation in estimates}

After model estimation, quantification of the amount of variation in the data captured by the estimate is an important task. We investigate the \emph{proportion of variation} typically used for PCA and a modification of \emph{directed \(R^2\)} from \citet{Kallus2019}.
To introduce these measures, consider first a more general type of low-rank model where \(\mat{X}_{ij} = \mat{Z}_{ij} + \mat{E}_{ij}\) with \(\Rank(\mat{Z}_{ij}) \leq k\). To avoid having to use mean vectors in rows and columns, assume that \(\mat{X}_{ij}\) is centered in both rows and columns. This bicentering is helpful since in our model there is no clear distinction between samples and features and views can assume these roles exchangeably. Data integration methods then form a structured estimate \(\widehat{\mat{Z}}_{ij}\) to \(\mat{Z}_{ij}\).

The proportion of variation is used to estimate the percentage of variance captured by a subset of the principal components during PCA. Analogously, we define
\begin{equation}\label{eq:prop-variation}
V_{ij} = \norm{\widehat{\mat{Z}}_{ij}}_F^2 / \norm{\mat{X}_{ij}}_F^2
\end{equation}
to be the proportion of variation in \(\mat{X}_{ij}\) captured by the model.
If it is assumed that \(\widehat{\mat{Z}}_{ij} = \mat{V}_i \mat{D}_{ij} \mat{V}_j^\top\) for matrices \(\mat{V}_i \in \Stiefel_{k, p_i}\) and \(\mat{V}_j \in \Stiefel_{k, p_j}\) then \(\norm{\widehat{\mat{Z}}_{ij}}_F^2 = \norm{\mat{V}_i \mat{D}_{ij} \mat{V}_j^\top}_F^2 = \sum_l \left(\mat{D}_{ij}^{(l, l)}\right)^2\).
It holds that \(\EV[\norm{\mat{X}_{ij}}_F^2] = \norm{\mat{Z}_{ij}}_F^2 + p_i p_j \sigma_{ij}^2 = \norm{\mat{Z}_{ij}}_F^2 (1 + 1 / \mathrm{SNR})\) where SNR is the signal-to-noise ratio of the input data.
This implies that the maximal value of \(V_{ij}\) depends on the SNR. A large SNR allows \(V_{ij}\) close to 1, whereas a lower SNR leads to a smaller upper-bound.
For example, \(\mathrm{SNR} = 0.5\) implies \(V_{ij} = 1/3\) assuming \(\widehat{\mat{Z}}_{ij} = \mat{Z}_{ij}\). If a larger proportion of variation is measured, then noise has been captured by the estimate \(\widehat{\mat{Z}}_{ij}\), and a smaller proportion of variation indicates that signal was missed. In practice however, the true residual variance and therefore the signal-to-noise ratio will often be unknown.
It holds that \(\sigma_{ij}^2 = \Variance[\mat{X}_{ij} - \mat{Z}_{ij}]\) and a natural estimate is therefore \(\widehat{\sigma}_{ij}^2 = \frac{1}{p_i p_j} \norm{\mat{X}_{ij} - \widehat{\mat{Z}}_{ij}}_F^2\), which can be used to compute an estimate for the signal-to-noise ratio.

Proportion of variation is only computed for a single matrix and does not account for relationships between matrices in the data integration model. Let two matrices \(\mat{X}_{ij}\) and \(\mat{X}_{ik}\) be given that share view \(i\) in their rows. We are interested in the amount of variation in the signal part of \(\mat{X}_{ij}\) that is linearly predicted by the signal in \(\mat{X}_{ik}\).
Assuming that both matrices follow the low-rank model specified above, we are interested in the linear relationship between signal matrices \(\mat{Y}_{ij} = \EV[\mat{X}_{ij}] = \EV[\mat{X}_{ik}] \mat{B} = \mat{Y}_{ik} \mat{B}\) for some coefficients \(\mat{B}\).
This corresponds to a multi-response linear regression model which can be solved using the usual normal equations.
Define the pseudo-inverse of a matrix \(\mat{A} = \mat{R}\mat{\Sigma}\mat{T}^\top\), with \(\mat{R}, \mat{T}\) orthogonal and \(\mat{\Sigma}\) diagonal, to be \(\mat{A}^+ := \mat{R} \mat{\Sigma}^+ \mat{T}^\top\), where \(\mat{\Sigma}^+\) is the diagonal matrix where all non-zero entries of \(\mat{\Sigma}\) were inverted. It then holds that
\(\mat{B} = (\mat{Y}_{ik}^\top \mat{Y}_{ik})^+ \mat{Y}_{ik}^\top \mat{Y}_{ij}\). Assuming that Eq.~\eqref{eq:data-model} holds implies \(\mat{B} = \mat{V}_k \mat{D}_{ik}^+ \mat{D}_{ij} \mat{V}_j^\top\)
and therefore \(\mat{Y}_{ik} \mat{B} = \mat{V}_i \mat{J}_{ik} \mat{D}_{ij} \mat{V}_j^\top\), where \(\mat{J}_{ik}\) is the diagonal matrix with ones on the diagonal where \(\mat{D}_{ik}\) is non-zero and zero otherwise. Analogous to the definition of proportion of variation above, we can then define a \emph{directed \(R^2\) value} of \(\mat{X}_{ij}\) predicted by \(\mat{X}_{ik}\) as
\begin{equation}\label{eq:directed-r2}
V_{ik \rightarrow ij} = \norm{\widehat{\mat{X}}_{ik} \mat{B}}_F^2 / \norm{\mat{X}_{ij}}_F^2 = \sum_l \mat{J}_{ik}^{(l, l)} \left(\mat{D}_{ij}^{(l, l)}\right)^2 / \norm{\mat{X}_{ij}}_F^2.
\end{equation}
Note that the definition of proportion of variation in Eq.~\eqref{eq:prop-variation} is a special case of directed \(R^2\) in Eq.~\eqref{eq:directed-r2} if data follows the model in Eq.~\eqref{eq:data-model} since \(V_{ij} = V_{ij \rightarrow ij}\). Since linear regression estimates are projections of the input data, it holds that \(V_{ik \rightarrow ij} \leq V_{ij}\) for any predictor matrix \(\mat{X}_{ik}\).
As a consequence, the range of \(V_{ik \rightarrow ij}\) is dependent on the SNR of \(\mat{X}_{ij}\).



If the shared view is in the columns then we consider the linear approximation \(\widehat{\mat{X}}_{ji}^\top = \widehat{\mat{X}}_{ki}^\top \mat{B}\). For the model in Eq.~\eqref{eq:data-model} this results in \(\mat{B} = \mat{V}_k \mat{D}_{ki}^+ \mat{D}_{ji} \mat{V}_j^\top\) and the same formula for \(V_{j,i \rightarrow k}\) as in Eq.~\eqref{eq:directed-r2} is recovered.
Similarly, it is possible that a view appears in the rows of one matrix and the columns of the other or vice-versa. The dependent and independent parts can then be transposed appropriately so that the shared view is always in the rows.


Other data integration methods use a different model than Eq.~\eqref{eq:data-model}. To be able to use \(V_{ik \rightarrow ij}\) with these methods it may not always be possible to find convenient expressions as in Eq.~\eqref{eq:directed-r2} after the second equality sign. In that case, the regression form is used.
Bayesian methods typically do not return a binary decision on which factors are included or excluded in an estimate. It is then necessary to threshold factor inclusion to obtain an interesting estimate of directed \(R^2\). Otherwise, all factors are always active and \(V_{ik \rightarrow ij} = V_{ij}\) for all predictors \(\mat{X}_{ik}\).

\section{Parameter estimation in the solrCMF problem}%
\label{sec:m-admm}

As established in the aforementioned sections, recovering a low-rank approximation of a data collection useful for data integration amounts to solving the solrCMF problem stated in Eq.~\eqref{eq:sparse-cmf:penalized-model}.
However, this problem has several complicating properties; it is non-convex due to the interactions between \(\mat{V}_i\) and \(\mat{V}_j\) for \(i \neq j\), and \(\mat{V}_l\) are constrained to the non-convex Stiefel manifolds \(\Stiefel_{k, p_l}\) \citep{Absil2008}. Moreover, the sparsity-inducing penalties render the problem non-smooth, which, if not handled carefully, can slow down algorithms significantly. Despite this, several of these issues can be mitigated by the fact that projecting onto the Stiefel manifold has a closed formula projection, and the $\ell_1$--regularization has an efficient proximal operator, namely soft thresholding. Therefore, we propose leveraging a multi-block ADMM framework with multi-affine constraints \citep{Gao2020}, which will enable us to use these efficient subroutines.

To be able to apply this framework, additional variables need to be introduced. As mentioned above, the loss function contains a non-convex term which can instead be replaced by new variables fulfilling multi-affine constraints. By introducing the multi-affine constraint
\(\mat{Z}_{ij} = \mat{V}_i \mat{D}_{ij} \mat{V}_j^\top\) the loss for input \(\mat{X}_{ij}\) becomes \(\norm{\mathcal{P}_{\mathrm{\Omega}_{ij}}(\mat{X}_{ij} - \mat{Z}_{ij})}_F^2\) which is convex in \(\mat{Z}_{ij}\).
The subproblem for \(\mat{V}_l\) requires handling elementwise-sparsity, due to the \(\ell_1\)-penalties, and orthogonality constraints simultaneously. No closed-form solution to this subproblem is known to the authors, and it was therefore decided, to decouple the orthogonality constraint from the element-wise sparsity constraint.
New variables \(\mat{U}_l \in \Oblique_{k, p_l}\) defined on the Oblique manifold \citep{Absil2008} are introduced, such that \(\mat{U}_l = \mat{V}_l\), and the \(\ell_1\)-penalties are applied to \(\mat{U}_l\) while the orthogonality constraint remains on \(\mat{V}_l\).
However, due to technical constraints in the framework of \citet{Gao2020}
it is necessary to introduce slack variables
\(\mat{V}'_l \in \R^{p_l \times k}\) such that
\(\mat{U}_l = \mat{V}_l + \mat{V}'_l\). Terms proportional to
\(\norm{\mat{V}'_l}_F^2\) are added to the objective function to control the
magnitude of the slack variables, with the aim of keeping them small.
The requirement that \(\mat{U}_l \in \Oblique_{k, p_l}\) may seem unnecessarily restrictive. If the columns of \(\mat{U}_l\) were unrestricted, then whole columns could become zero due to the \(\ell_1\)-penalty placed upon \(\mat{U}_l\). Since \(\mat{V}_l\) is an orthogonal matrix, its columns will never be zero vectors and the slack variable \(\mat{V}'_l\) will have to absorb the difference to the zero column in \(\mat{U}_l\). \(\mat{V}'_l\) then essentially becomes equal to \(\mat{V}_l\) in that column. However, we aim for the components of \(\mat{V}'_l\) to be close to zero and to only capture minor, essentially inconsequential, differences between \(\mat{U}_l\) and \(\mat{V}_l\). Therefore, it is beneficial to require that columns of \(\mat{U}_l\) are unit vectors to meaningfully guide \(\mat{V}_l\) towards sparser solutions. The complexity of the optimization subproblem for \(\mat{U}_l\) does increase when solved on the Oblique manifold, but allows for the derivation of a closed-form solution.

In the following, split the set of variables into two groups
\begin{equation*}
\mat{\Theta} = (\mat{V}, \mat{D}, \mat{U})\quad\text{and}\quad
\mat{\Delta} = (\mat{V}', \mat{Z}).
\end{equation*}
In addition denote
\(\mat{\Delta}_1 = \mat{V}'\) and \(\mat{\Delta}_2 = \mat{Z}\).

The resulting optimization problem is then
\begin{equation}\tag{P}\label{eq:m-admm:admm-opt-problem}
\left\{\begin{array}{rl}
\displaystyle\min_{\mat{\Theta}, \mat{\Delta}}
&\displaystyle\frac{1}{2} \sum_{(i, j) \in \mathcal{I}}
\norm*{\mathcal{P}_{\mathrm{\Omega}_{ij}}\left(\mat{X}_{ij,} - \mat{Z}_{ij}\right)}_2^2 \\
&\displaystyle+ \lambda_1 \sum_{(i, j) \in \mathcal{I}} \sum_{l = 1}^k
\vect{w}_{ij}^{(l)} \abs{\mat{D}_{ij}^{(l, l)}}
+ \lambda_2 \sum_{l = 1}^m  \vect{w}_1^{(l)} \norm{\mat{U}_l}_1
+ \frac{\mu}{2} \sum_{l = 1}^m \vect{w}_2^{(l)} \norm{\mat{V}'_l}_F^2 \\
\text{subject to }& \mat{V}_l \in \Stiefel_{k, p_l}, \mat{U}_l \in \Oblique_{k, p_l}, \mat{D}_{ij} \in \mathcal{D}_k, \text{ and} \\
A(\mat{\Theta}) + Q(\mat{\Delta}) &=
\begin{pmatrix}
A_1(\mat{\Theta}) & + & Q_1(\mat{\Delta}_1) \\
A_2(\mat{\Theta}) & + & Q_2(\mat{\Delta}_2) \\
\end{pmatrix}
=
\begin{pmatrix}
\mat{U}_l - \mat{V}_l & + & (-\mat{V}'_l) \\
-\mat{V}_i \mat{D}_{ij} \mat{V}_j^\top & + & \mat{Z}_{ij} \\
\end{pmatrix}
= \mat{0}
\end{array}
\right.
\end{equation}
where it is understood that each constraint is included for all \(l \in \mathcal{V}\)
and \((i, j) \in \mathcal{I}\).
The objective function of the reformulated problem \eqref{eq:m-admm:admm-opt-problem} is convex in all participating variables,
except for the manifold constraints on \(\mat{V}_l\) and \(\mat{U}_l\).

The ADMM approach requires the introduction of Lagrange multipliers for the
equality constraints in \eqref{eq:m-admm:admm-opt-problem} and a step-size
parameter \(\rho\).
Let \(\mat{M}^{(1)}_l\) for \(l \in \mathcal{V}\) and \(\mat{M}^{(2)}_{ij}\) for
\((i, j) \in \mathcal{I}\) be the raw multipliers and define the scaled
multipliers as \(\mat{\Lambda}^{(j)} = \mat{M}^{(j)} / \rho\). Denote the
set of all multipliers by \(\mat{\Lambda}\).
The augmented Lagrangian function corresponding to \eqref{eq:m-admm:admm-opt-problem} is then
\begin{equation}\label{eq:m-admm:admm-augmented-lagrangian}
\begin{aligned}
\mathcal{L}_{\rho}(\mat{\Theta}, \mat{\Delta}, \mat{\Lambda}) &=
\frac{1}{2} \sum_{(i, j) \in \mathcal{I}}
\norm*{\mathcal{P}_{\mathrm{\Omega}_{ij}}\left(\mat{X}_{ij} - \mat{Z}_{ij}\right)}_2^2 \\
&+ \lambda_1 \sum_{(i, j) \in \mathcal{I}} 
\sum_{l = 1}^k \left(
    \iota_{\mathrm{D}_k}(\mat{D}_{ij,})
    + \vect{w}_{ij}^{(l)} \abs{\mat{D}_{ij}^{(l, l)}}
\right) \\
&+ \lambda_2 \sum_{l = 1}^m \vect{w}_1^{(l)} \norm{\mat{U}_l}_1
+ \frac{\mu}{2} \sum_{l = 1}^m 
\vect{w}_2^{(l)} \norm{\mat{V}'_l}_F^2 \\
&+ \frac{\rho}{2} \sum_{l = 1}^m \left[
\norm{\mat{U}_l - \mat{V}_l - \mat{V}'_l + \mat{\Lambda}^{(1)}_{l}}_F^2
- \norm{\mat{\Lambda}^{(1)}_{l}}_F^2\right] \\
&+ \frac{\rho}{2} \sum_{(i, j) \in \mathcal{I}} \left[
\norm{\mat{Z}_{ij} - \mat{V}_i \mat{D}_{ij} \mat{V}_j^\top
+ \mat{\Lambda}^{(2)}_{ij}}_F^2
- \norm{\mat{\Lambda}^{(2)}_{ij}}_F^2\right].
\end{aligned}
\end{equation}
Multi-block ADMM with multi-affine constraints as described in \citet{Gao2020} finds a solution to \eqref{eq:m-admm:admm-opt-problem} by
(i) sequentially optimizing the augmented Lagrangian over the variables in \(\mat{\Theta}\),
(ii) simultaneously optimizing over the variables in \(\mat{\Delta}\), and by
(iii) updating the multipliers.
Denote \(L^t_\rho = \mathcal{L}_\rho(\mat{\Theta}^t, \mat{\Delta}^t, \mat{\Lambda}^t)\) where \(t\) indicates the iteration. The algorithm terminates when decrease in the augmented Lagrangian is below a threshold either in absolute terms \(\abs{L^t_\rho - L^{t-1}_\rho}\) or relative terms \(\abs{L^t_\rho - L^{t-1}_\rho} / \abs{L^{t-1}_\rho}\).
The necessary steps to perform the optimization are summarized in
Algorithm~\ref{alg:m-admm:admm-alg} and the solutions to the subproblems are
given in Table~\ref{tbl:m-admm:subproblems}.

\subsection{Convergence analysis}

The theory in \citet{Gao2020} was developed for optimization problems on \(\R^n\). In this work, some subproblems are constrained to an embedded manifold of \(\R^n\), namely the Stiefel and Oblique manifolds. We therefore extended the work of \citet{Gao2020} to proof that the convergence results hold even when embedded manifold set constraints are added to sequential subproblems. This holds as long as the subproblems can be solved to optimality over embedded manifolds that can be described by the level set of a smooth function. See the supplementary for a detailed description of our extension of \citet{Gao2020}.



\begin{theorem}\label{thm:solrcmf-convergence}
Choose
\begin{equation}
\rho > \max\left(2, \max\left(
2\frac{\mu \left(\max_l \bm{w}_2^{(l)}\right)^2}{\min_l \bm{w}_2^{(l)}},
\frac{1 + \mu \max_l \bm{w}_2^{(l)}}{2}
\left(1 + 2 \frac{\max_l \bm{w}_2^{(l)}}{\min_l \bm{w}_2^{(l)}}\right)^2
\right)\right).
\end{equation}
It then holds that Algorithm~\ref{alg:m-admm:admm-alg} applied to the augmented Lagrangian \(L_\rho(\mat{\Theta}, \mat{\Delta}, \mat{\Lambda})\) in Eq.~\eqref{eq:m-admm:admm-augmented-lagrangian} produces a sequence of iterates \((\mat{\Theta}^t, \mat{\Delta}^t, \mat{\Lambda}^t)\) that converges globally to a local minimum \((\mat{\Theta}^*, \mat{\Delta}^*, \mat{\Lambda}^*)\) of \eqref{eq:m-admm:admm-opt-problem}.
\end{theorem}

\begin{proof}
The proof uses an extended form of the framework by \citet{Gao2020} and can be found in the Supplementary Material.
\end{proof}

Particularly remarkable is that global convergence to a local minimum can be achieved for any combination of hyperparameters \(\lambda_1, \lambda_2\), and \(\mu\) as long as \(\rho\) is chosen as in Theorem~\ref{thm:solrcmf-convergence}.

\begin{algorithm}[htb]
\caption{Multi-block ADMM algorithm.
\newline The notation \(\mat{\Theta}_{<}\) denotes all variables in 
\(\mat{\Theta}\) updated before \(\vect{\theta}\) and \(\mat{\Theta}_{>}\) all variables updated after \(\vect{\theta}\). The superscript indicates if the variables have been already updated in step \(t\) or old values from step \(t-1\) are used.
\(L_\rho^t\) and \(L_\rho^{t-1}\) indicate the value of the augmented Lagrangian
at the updated and previous iterates, respectively.}%
\label{alg:m-admm:admm-alg}
\begin{algorithmic}
\Require{Data matrices \(\mat{X}_{ij,\gamma} \in \mathbb{R}^{p_i \times p_j}\), regularization parameters \(\mat{\lambda} = (\lambda_1, \lambda_2) > \mat{0}\), \(\mu > 0\), step size parameter \(\rho > 0\), tolerances \(\epsilon_{\text{abs}}\) and \(\epsilon_{\text{rel}}\), maximum iterations \(T_{\max}\)}
\Initialize{Initial values for \(\mat{\Theta}\) and \(\mat{\Delta}\).
Set \(\mat{\Lambda} = \mat{0}\) and \(L_\rho^0 = \infty\).}
\Statex
\For{\(t \in \{1, \dots, T_{\max}\}\)}
\Comment{Main loop}
    \For{\(\mat{\theta} \in \mat{\Theta}\)}
    \Comment{Sequential updates}
        \Let{\(\mat{\theta}^t\)}{\(\displaystyle\argmin_{\mat{\theta} \in \mathcal{M}_\theta}
        \mathcal{L}_{\rho}((\mat{\Theta}_{<}^t, \mat{\theta}, \mat{\Theta}_{>}^{t - 1}),
            \mat{\Delta}^{t - 1}, \mat{\Lambda}^{t - 1}) + \frac{\alpha}{2}\norm{\mat{\theta} - \mat{\theta}^{t-1}}^2\)}
        \State{or}
        \Let{\(\mat{\theta}^t\)}{\(\displaystyle\argmin_{\mat{\theta}}
        \mathcal{L}_{\rho}((\mat{\Theta}_{<}^t, \mat{\theta}, \mat{\Theta}_{>}^{t - 1}),
            \mat{\Delta}^{t - 1}, \mat{\Lambda}^{t - 1})\)}
    \EndFor
    \Let{\(\mat{\Delta}^t\)}{\(\displaystyle\argmin_{\mat{\Delta}}
        \mathcal{L}_{\rho}(\mat{\Theta}^t, \mat{\Delta}, \mat{\Lambda}^{t - 1})\)}
    \Comment{Simultaneous update}
    \Let{\(\mat{\Lambda}^t\)}{\(\mat{\Lambda}^{t - 1} +
        A(\mat{\Theta}^{t + 1}) + Q(\mat{\Delta}^{t + 1})\)}
    \Comment{Lagrange multiplier update}
    \Let{\(L_\rho^t\)}{\(\mathcal{L}_\rho(\mat{\Theta}^t, \mat{\Delta}^t, \mat{\Lambda}^t)\)}
    \Comment{Check convergence}
    \If{\(\abs{L_{\rho}^t - L_{\rho}^{t - 1}} < \epsilon_{\text{abs}}\) or
    \(\abs{L_{\rho}^t - L_{\rho}^{t - 1}} / \abs{L_{\rho}^{t - 1}} < \epsilon_{\text{rel}}\)}
        \Break
    \EndIf
\EndFor
\Statex
\Return{final estimates for \(\mat{\Theta}\), \(\mat{\Delta}\), \(\mat{\Lambda}\)}
\end{algorithmic}
\end{algorithm}


\begin{table}[ht]
\caption{Solutions of the ADMM Subproblems. The superscript \(+\) indicates an updated iterate and the superscript \(-\) indicates the previous iterate.}%
\label{tbl:m-admm:subproblems}
\begin{tabular}{p{0.25\textwidth}p{0.7\textwidth}}
\toprule
Variable & Solution to subproblem \\
\midrule
\multicolumn{2}{l}{\emph{Solved sequentially}} \\
\(\mat{V}_i\) for \(i \in \mathcal{V}\) & Compute
\(\mat{M} = \sum_{(i, l) \in \mathcal{I}} 
(\mat{Z}_{il}^- + \mat{\Lambda}^{(2), -}_{il}) \mat{V}_l \mat{D}_{il}^-
+ \sum_{(l, i) \in \mathcal{I}} 
(\mat{Z}_{li}^- + \mat{\Lambda}^{(2), -}_{li})^\top \mat{V}_l \mat{D}_{li}^-
+ (\mat{U}_i^- - {\mat{V}'_i}^- + \mat{\Lambda}^{(1), -}_i) 
+ \frac{\alpha}{\rho} \mat{V}_i^-
\). Project \(\mat{M}\) onto Stiefel manifold \(\Stiefel_{k, p_l}\)
by obtaining the SVD \(\mat{M} = \mat{B}_1 \mat{\Pi} \mat{B}_2^\top\) and
setting \(\mat{V}_i^+ = \mat{B}_1 \mat{B}_2^\top\). \\
\(\mat{D}_{ij}\) for \((i, j) \in \mathcal{I}\) & Compute
\(\mat{M} = {\mat{V}_i^+}^\top (\mat{Z}_{ij}^- + \mat{\Lambda}^{(2), -}_{ij})
{\mat{V}_j^+}\). To obtain \(\mat{D}_{ij}^+\) set
\({\mat{D}_{ij}^+}^{(r, c)} = 0\) for \(r \neq c\) and
\({\mat{D}_{ij}^+}^{(l, l)} =
\mathrm{ST}_{\lambda_1 \vect{w}_{ij}^{(l)} / \rho}(\mat{M}^{(l, l)})\)
for \(l \in \{1, \dots, k\}\). \\
\(\mat{U}_i\) for \(i \in \mathcal{V}\) & Compute
\(\mat{M}_i = \mat{V}_i^+ + {\mat{V}'_i}^- - \mat{\Lambda}_i^{(1), -} 
+ \frac{\alpha}{\rho} \mat{U}_i^-
\). Denote
\(w = \lambda_2 \vect{w}_1^{(i)} / \rho\) and for each \(l \in \{1, \dots, k\}\) set
\(w'_l = \max_j \abs{\mat{M}_i^{(j, l)}}\). If \(w'_l > w\), then set
\({\mat{U}^{(:, l)}}^+ = \mathrm{ST}_w(\mat{M}_i^{(:, l)}) /
\norm{\mathrm{ST}_w(\mat{M}_i^{(:, l)})}_2\). If \(w'_l \leq w\), then find 
\(j_0 = \argmin_j -\abs{\mat{M}_i^{(j, l)}} + w\), and set
\({\mat{U}^{(:, l)}}^+ = \Sign(\mat{M}_i^{(j_0, l)}) \vect{e}_{j_0}\).
\\
\multicolumn{2}{l}{\emph{Solved simultaneously}} \\
\(\mat{V}'_i\) for \(i \in \mathcal{V}\) & Set
\({\mat{V}'_i}^+ =  \frac{\rho}{\rho + \mu \vect{w}_2^{(i)}}
(\mat{U}_i^+ - \mat{V}_i^+ + \mat{\Lambda}_i^{(1),-})\). \\
\(\mat{Z}_{ij}\) for \((i, j) \in \mathcal{I}\) & Set
\(\mat{Z}_{ij}^{(r, c),+} = \frac{\rho}{\rho + 1}
(\mat{V}_i^+ \mat{D}_{ij}^+ {\mat{V}_j^+}^\top
- \mat{\Lambda}_{ij}^{(2),-})^{(r, c)} 
+ \frac{1}{\rho + 1} \mat{X}_{ij}^{(r, c)}\)
if \(\mat{X}_{ij}^{(r, c)}\) observed and
\(\mat{Z}_{ij}^{(r, c),+} =
(\mat{V}_i^+ \mat{D}_{ij}^+ {\mat{V}_j^+}^\top
- \mat{\Lambda}_{ij}^{(2),-})^{(r, c)}\) otherwise. \\
\bottomrule
\end{tabular}
\end{table}


\subsection{Algorithm initialization}

Due to the non-convexity of the constraints in \eqref{eq:m-admm:admm-opt-problem} convergence to a global minimum is not guaranteed and the computed solution depends on the initial values.
Finding a useful initial state is therefore important.
The algorithm can either be initialized randomly or with a structured initial guess.
For the purposes of this work, we found that testing multiple random
starting points, estimating our model without penalization (i.e.\@ \(\lambda_1 = 0 = \lambda_2\)), and keeping the solution leading to the minimal value of the augmented Lagrangian in Eq.~\eqref{eq:m-admm:admm-augmented-lagrangian} worked well.
A description of our initialization strategy and a structured alternative for multiview data is described in the supplementary.
This initial estimate is then used as a starting point during hyperparameter selection. For practical application we recommend normalizing input data matrices \(\mat{X}_{ij}\) before initialization to Frobenius norm \(n_{ij} / (p_i p_j)\), where \(n_{ij}\) is the number of observed entries in \(\mat{X}_{ij}\). This ensures that singular values are between -1 and 1, which in particular implies that scales of variables during optimization do not vary severely from each other.

\subsection{Hyperparameter selection}%
\label{ssec:hyperparameter-selection}

Optimal hyperparameters are data dependent and therefore need to be selected for each new application of solrCMF. We perform hyperparameter selection in two sequential steps: First, pairs of hyperparameters are chosen randomly from a uniform distribution on a logarithmic scale, model parameters are estimated using all data for each pair of hyperparameters, and the zero patterns in solutions for \(\mat{D}_{ij, \gamma}\) and \(\mat{U}_l\) are recorded. Second, we use Wold-type \(K\)-fold cross-validation (CV) \citep{WoldCV}, where the observed entries in each data matrix are randomly divided into \(K\) group, to determine the best structure and consequentially best pair of hyperparameters. During parameter estimation in this second step, penalties are not used but zero patterns in \(\mat{D}_{ij, \gamma}\) and \(\mat{U}_l\) are fixed. Following the Wold-type CV strategy, one fold per data matrix is considered missing in addition to actual missing values in the data. The trained model is then used to predict the values in the left out folds. The predicted values are then compared to observed test values and a score is computed. We use negative Mean Squared Error (MSE) to compare models. To choose the final model, first the zero pattern leading to the highest negative MSE is determined. Then the solution is chosen that has the most sparsity and whose score is at most one standard error away from the one with the maximum negative MSE. Once the best zero pattern is determined, the model is either retrained on all data using with this particular zero pattern, or model parameters are re-computed on all data with penalties corresponding to the best solution. In this work we only followed the first strategy, but both are offered in the package.

The reason for this two-step hyperparameter selection method is as follows.
When penalties are active (i.e.\@ one or both of \(\lambda_1 > 0\) and \(\lambda_2 > 0\)) and their magnitudes increase, then singular value estimates in \(\mat{D}_{ij, \gamma}\) are increasingly shrunk towards zero and columns of \(\mat{V}_l\) will progressively approximate standard unit vectors. These effects have been widely observed in penalized regression \citep{StatLearningSparsity}. This can lead to models that capture the correct sparsity structure, however, due to potentially high penalization singular values and factors can be shrunk too much to predict left out test data appropriately. A remedy to this problem is to estimate a sequence of structures for a sequence of pairs of hyperparameters, fix the sparsity structures and re-estimate model parameters without penalties. A variant of this has been used for LARS (a regression framework encompassing the Lasso) and ordinary least squares \citep{Efron2004} and found to improve \(R^2\) in regression.
We chose Wold-type CV over other types of matrix cross-validation, such as BCV \citep{Owen2009}, due to its ability to deal with missing values in the data.
A similar \emph{structure-estimation first, model selection second} strategy is used by SLIDE \citep{Gaynanova2019}.

\section{Simulation results}%
\label{sec:simulation}

To evaluate the performance of solrCMF we performed two simulation studies and compared the algorithm to other methods in the field.

\subsection{Data integration and structure estimation}

In a first study, we compare solrCMF to (1) CMF \citep{Klami2014}, (2) MMPCA \citep{Kallus2019}, (3) SLIDE \citep{Gaynanova2019}, and (4) MOFA \citep{Argelaguet2018} on a complex data setup, which includes both augmented and tensor-like data (Figures~\ref{fig:data-layouts:augmented} and \ref{fig:data-layouts:tensor-like}).
Data was generated according to the model in Eq.~\eqref{eq:data-model} with the additional assumption of normally distributed noise. We used a setup with \(\mathcal{V} = \{1, 2, 3, 4\}\) and \(\mathcal{I} = \{(1, 2, 1), (1, 3, 1), (1, 3, 2), (4, 3, 1), (1, 4, 1)\}\).
Three indices are used to indicate replication of view relationships in the last index.
Figure~\ref{fig:graph-res:full-layout} illustrates the layout of this simulation.
Dimensions of views were chosen as \(p_1 = 50\), \(p_2 = 25\), \(p_3 = 35\), and \(p_4 = 30\).
Five orthogonal factors with 75\% sparsity are simulated for each view and the following singular values were used:
\begin{equation*}
\begin{array}{rclllll}
\mat{D}_{12,1} &=& \Diag(3,   & 3.5,  & 0,   & 0,   & 4), \\
\mat{D}_{13,1} &=& \Diag(2.5, & 2.75, & 3,   & 0,   & 0), \\
\mat{D}_{13,2} &=& \Diag(2.5, & 0,    & 3.5, & 3,   & 0), \\
\mat{D}_{43,1} &=& \Diag(3,   & 0,    & 4.5, & 3.5, & 0), \text{ and} \\
\mat{D}_{14,1} &=& \Diag(0,   & 0,    & 3.5, & 0,   & 4).
\end{array}
\end{equation*}
Many shared subspace relationships are implied by this layout. For example, there is a two-dimensional shared subspace between tensor layers (B and C) and a one-dimensional individual subspace in each. There is one dimensional subspace that is shared between tensor layers which is also shared with the side-information matrices on the left (A) and above (D). There is an augmented matrix (E) between views 1 and 4 which shares separate one-dimensional subspaces of variation with A and D.
To simulate sparse orthogonal factors we sample non-zero entries in random locations from a standard normal distribution and employ a sparsity-preserving variant of the Gram-Schmidt orthogonalization algorithm to ensure orthogonality. Details can be found in the Supplementary Material.
The residual noise variance for each matrix was determined by assuming that the signal-to-noise ratio \(\norm{\mat{V}_i \mat{D}_{ij,\gamma} \mat{V}_j^\top}_F^2 / \mathbb{E}(\norm{\mat{E}_{ij,\gamma}}_F^2) = 0.5\), where \(\mathbb{E}(\norm{\mat{E}_{ij,\gamma}}_F^2) = \sigma^2 p_i p_j\) assuming that residuals are normally distributed with zero mean and variance \(\sigma^2\). Simulations were repeated 250 times.
Each data matrix was preprocessed by bicentering each matrix using the method described in \citet{Hastie2015}. After centering, each matrix was scaled to have Frobenius 1.

\begin{figure}[ht]
\centering
\input{directed-r2-figure-250runs}
\caption{
Directed \(R^2\) between matrices predicted by multiple methods and the layouts used to estimate all parts of the graph.
(a) shows the full layout used for simulation and estimation in solrCMF. (d) and (g) show unfolded layouts and dashed lines indicate different scenarios where it was not possible to estimate the full model at once. Descriptions of the unfolding strategies are given in the text.
Colors orange, blue, and black indicate results from different scenarios in all figures.
Matrices are represented by nodes in (b), (c), (e), (f), (h), and (i) and arrows indicate that a non-zero directed \(R^2\) was found between matrices. The size of the arrows is proportional to the magnitude of average directed \(R^2\) across 250 simulation runs, and arrows are pointing from the predictor matrix to the dependent matrix.}
\label{fig:graph-res}
\end{figure}
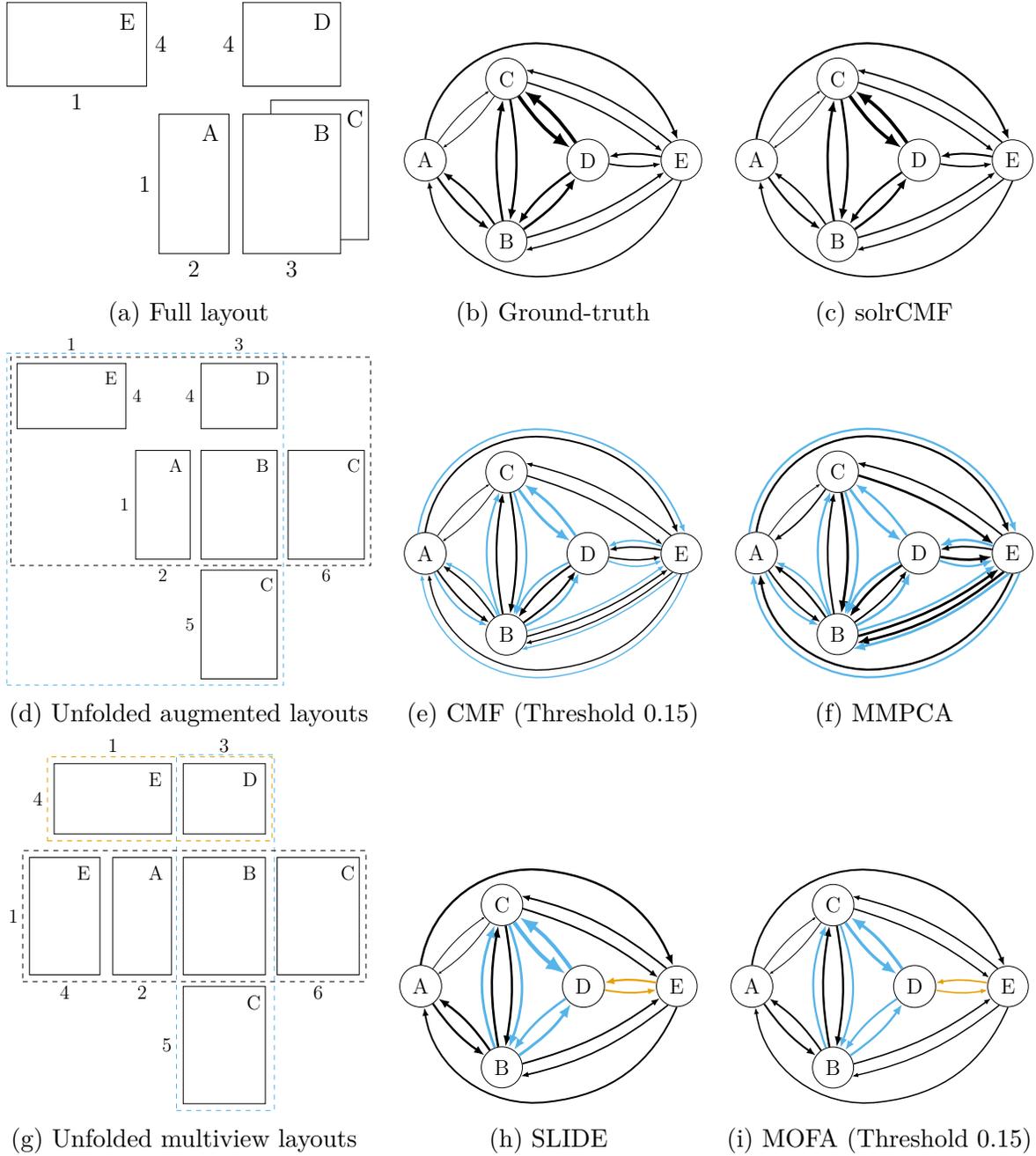

The four competing methods are divided into two groups: CMF and MMPCA, which can be applied to augmented layouts but not tensor-like data, and SLIDE and MOFA, which can only be applied to multi-view layouts.
To be able to use competitor methods on the complex data layout in Figure~\ref{fig:graph-res:full-layout}, we had to rearrange the data.

For CMF and MMPCA, the second layer of the tensor (C) was placed next to the first layer (B), once vertically, by columns, and once horizontally, by rows. The resulting augmented layouts are therefore called \textbf{horizontal} and \textbf{vertical} in the following. See Figure~\ref{fig:graph-res:unrolled-augmented-layout} for an illustration.
Changing the layout this way requires the introduction of new views making the unfolded layout not equivalent to the original one. Conceptually, view 1 and 5, as well as view 3 and 6 can respectively be seen as independent copies of each other.

To adapt the data for SLIDE and MOFA, we divided the original data into three partially overlapping layouts.
\begin{enumerate}
\item \textbf{Horizontal multi-view layout}: The tensor is unfolded by rows and the augmented matrix (E) is placed such that its column view matches the common row view 1. An additional view 6 needs to be introduced as an independent copy of view 3.
\item \textbf{Vertical multi-view layout}: The tensor is unfolded by columns and view 3 is regarded as the common view. An additional view 5 needs to be introduced as an independent copy of view 1.
\item \textbf{Upper multi-view layout}: Includes only the augmented matrix (E) along the common view 4 together with its match (D).
\end{enumerate}
See Figure~\ref{fig:graph-res:unrolled-multiview-layout} for an illustration of the layouts..

Whenever possible, all methods were run with default parameters.
For solrCMF a maximum rank of 10 was requested and 100 pairs of hyperparameters were chosen on a log-scale between 0.05 and 1. \(K = 10\) folds were used for CV.
CMF was run without bias terms and hyperparameters for the ARD prior were set to \(\alpha = 1\) and \(\beta = 0.001\) since factor selection performed poorly otherwise. 25000 iterations were performed and the maximum rank was set to 10.
MMPCA is setup to only use the data integration and sparsity penalties, maximum rank was set to 10, and 100 pairs of hyperparameters were chosen in the same way as for solrCMF.
SLIDE is run using default parameters only.
MOFA is run with ARD and spike-and-slab priors for both factors and weights, the maximum rank of the solution is set to 10, and MOFA is allowed to scale each data matrix by its standard deviation. This is the default preprocessing for MOFA and performance was poor without allowing this.

Directed \(R^2\) was used as a metric to compare methods due to its capability of quantifying shared variation between two matrices in a single number. Directed graphs can be constructed with data matrices as nodes by adding edges from the predictor to dependent matrix whenever directed \(R^2\) between them is non-zero. This is exemplified for the ground-truth of this simulation in Figure~\ref{fig:graph-res:truth}. Groundtruth values for directed \(R^2\) were computed using the regression formulation in Eq.~\eqref{eq:directed-r2} for each simulated signal. Visual summaries of the directed \(R^2\) graphs estimated by each method are shown in Figures~\ref{fig:graph-res}.

SLIDE and MMPCA return binary decisions on which factors are included and the regression formulation in Eq.~\eqref{eq:directed-r2} can readily be applied to returned estimates.
CMF and MOFA return expected values of parameter estimates which are continuous. This means in practice that exact zeros are observed seldom and usually all factors are active. However, due to the ARD prior, the contribution of non-active priors is shrunk towards zero. MOFA does remove factors that contribute very little across all data matrices, but does not explicitly predict which factors appear only in a subset of matrices.
To determine whether a factor is active in a signal matrix, we need to decide on a threshold applied to factor scale. Making manual inclusion/exclusion decisions on a per-factor basis is infeasible due to the large number of matrices and repetitions. Instead, for each unfolded layout we investigate a histogram of factor scales across all matrices, factors, and repeated simulation runs. We looked for a clear division between very small and larger values. For our specific example, a threshold of 0.15 was found to be suitable for both methods. See Figure 1 in the Supplementary Material for additional information. Note that in practice typically no repeated runs are available and the number of factors to use to decide on a possible threshold will be much smaller.


Numerical results are collected in Table~\ref{tbl:directed-r2} for each pair of connected data matrices. Expected directed \(R^2\) for the groundtruth is shown in rows marked with dependent and predictor matrices. Mean deviations from the groundtruth of each sample are computed in comparison to the directed \(R^2\) estimated for each method. We then use paired Wilcoxon signed rank tests \citep{Wilcoxon1945} to determine whether the differences between groundtruth and predicted directed \(R^2\) for solrCMF are significantly different than those for other methods and layouts. We show results of the test by indicating three different thresholds for statistical significance in the table.

The overall root mean squared deviation (RMSE) in Table~\ref{tbl:directed-r2} for each method shows that solrCMF performs overall best, with MOFA following not far after, however results vary depending on the layout under consideration. Overall, layouts including more data matrices resulted in better overall RMSE for MOFA. Investigating individual links solrCMF often either performs best or at least appears among the group of best-performing methods. In addition, in cases where solrCMF does not achieve the best mean deviation, its standard deviation is lower and therefore predicts results more consistently.


Table 1 in the Supplementary Material shows examples of estimated rank for matrices A, B, as well as shared signal between them. solrCMF and CMF behave similarly by occasionally over- and underestimating the rank but for a large proportion of runs the ranks were estimated correctly. SLIDE has a tendency to underestimate rank but picks the right rank for matrices A and B in most cases. Shared rank is underestimated more often. MOFA overestimates rank in almost all cases. MMPCA exhibits a wide range of estimated ranks, both largely over- and underestimated, but estimates the correct ranks in at least 50\% of all runs.

CMF and MOFA consistently underestimated directed \(R^2\) whereas SLIDE mostly overestimated the ground-truth value. We investigated the results for MOFA and CMF and found that the chosen threshold is not responsible for the underestimation, with results only changing minimally for lower thresholds. However, reducing the threshold too much and including all factors consistently leads to non-informative estimates of directed \(R^2\) since \(V_{ik \rightarrow ij} = V_{ij}\) for all predictors \(\mat{X}_{ik}\). Estimated ranks could be considered the culprit, with underestimated ranks leading to underestimated directed \(R^2\) and vice-versa. However, as is discussed above, CMF largely estimates the right ranks and MOFA actually overestimates ranks, whereas SLIDE tends to underestimate ranks. A possible explanation can be found in estimated factor scale. Table 2 in the Supplementary Material shows examples for matrices A and B. solrCMF estimates factor scale very close to the groundtruth, with the exception of Factor 3 in matrix A. CMF and MOFA both underestimate the groundtruth whereas SLIDE tends to overestimate it. This is likely the origin of the under-/overestimation of directed \(R^2\).
solrCMF as well as and MMPCA occasionally over and occasionally underestimate the directed \(R^2\) which is to be expected. However, MMPCA's large variation in estimated rank and therefore signal quality is likely the cause of the high standard deviation in directed \(R^2\) estimates which can be seen in Table~\ref{tbl:directed-r2}.


\subsection{Recovering factor sparsity}

\begin{figure}[hbt]
    \centering
    \includegraphics[width=\textwidth]{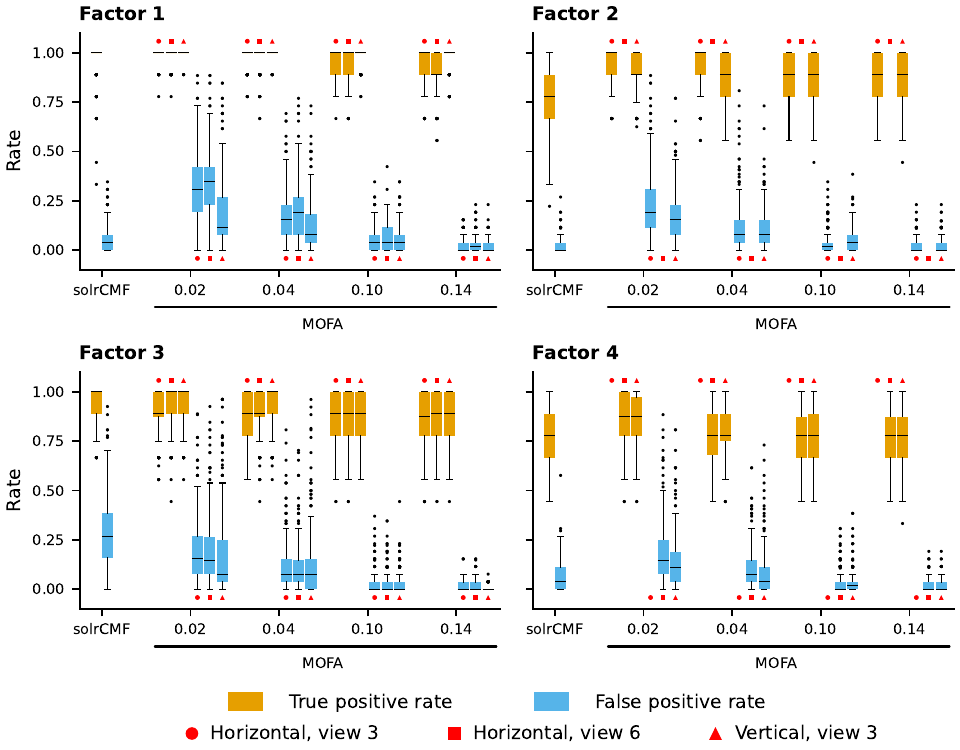}
    \caption{Sparsity patterns of factors in view 3 (and view 6 in the horizontal multi-view layout) in Figure~\ref{fig:graph-res:full-layout} are compared to groundtruth. True positive rate and false positive rate of estimating the correct non-zero entries are computed for each factor and shown as boxplots. For MOFA, view 3 appears also as view 6 in the horizontal layout and appears again as view 3 in the vertical layout. Results for all layouts are included. Thresholds used to determine sparsity patterns for factor vectors in MOFA are indicated on the horizontal axis. Note that Factor 2 and 4 are not estimable in view 6 and 3, respectively, due to the unfolding of tensor layers into a multi-view layout.}
    \label{fig:fpr-tpr-varied-signal}
\end{figure}

Estimation of factor sparsity is evaluated on results from Simulation 1. In particular, we focused on comparing results for solrCMF and MOFA. View 3 was selected since it appears three times in the unfolding for MOFA. Once as-is in the horizontal layout, once as view 6, a copy of view 3, again in the horizontal layout, and finally as the common column view in the vertical layout (see Figure~\ref{fig:graph-res:unrolled-multiview-layout}). Results for all three cases were considered.

As is the case during structure estimation, MOFA does not return estimates of factors containing exact zeros. Therefore, thresholding was required once again to evaluate factor sparsity. Since it is not clear a-priori which thresholds are suitable we test multiple. We decided on possible thresholds 0.02, 0.04, 0.1, and 0.14. All of these seem plausible without any obvious candidate when investigating results from representative simulation runs (see Figure 2 in the Supplementary Material for an example). Factors were then thresholded and the zero/non-zero pattern of the factors was extracted. solrCMF provides sparse estimates of factors directly, stored in matrix \(\mat{U}_3\), and no thresholding is necessary. We then compared sparsity patterns to the underlying groundtruth and evaluate results by computing true positive rates (TPR) and false positive rates (FPR).

For both solrCMF and MOFA, factors are returned in arbitrary order and need to be matched against groundtruth. To do so, the scalar product between estimated factors and groundtruth factors was computed and a match was considered when its value was at least 0.75 in absolute value.

Investigating results in Figure~\ref{fig:fpr-tpr-varied-signal}, we see that the performance of MOFA, as expected, depends on the choice of threshold, as well as on the choice of unfolding. Differences in unfoldings are most apparent in the FPR for small thresholds, suggesting that shrinkage has been applied in varying proportions dependent on the data. Note how not all factors are estimable in all views. However, this is a technical artifact due to the unfolding. 
Unsurprisingly, lower thresholds lead to a higher TPR as well as a higher FPR, since more entries are left untouched. Increasing the threshold removes spurious false positives, but eventually also decreases performance with respect to the TPR, setting true non-zero entries to zero.

solrCMF consistently performs comparably to MOFA. All factors are estimable and except for the case of Factor 3 the FPR is low. For Factor 2, solrCMF tends to discover fewer TPs than MOFA, however, FPR is almost always 0, which means that almost all of the discovered non-zero entries are correct. In case of Factor 4, MOFA can achieve a better TPR than solrCMF, however, this is dependent on the chosen threshold and leads to an increase in FPR as well.
If binary decisions on factor entries are required, solrCMF can provide a more comprehensive answer than MOFA by performing variable selection within factors during estimation. This, however, comes at the expense of flexibility by possibly discarding entries that could have been relevant.

\subsection{Identification of factors}

As mentioned during the discussion of model identifiability in Section~\ref{ssec:model-identification}, subspace rotations can occur when singular values are close. In the Supplementary Material of \citet{Argelaguet2020} it is mentioned that sparsity in factors can help identify factor rotation. In addition, we explore the impact of side information on identification of factors.
To do so, we use a second simulation setup with \(\mathcal{V} = \{1, 2, 3, 4\}\) and \(\mathcal{I} = \{(1, 2, 1), (1, 3, 1), (1, 3, 2), (4, 3, 1)\}\). Dimensions of views were chosen as \(p_1 = 100\), \(p_2 = 50\), \(p_3 = 100\), and \(p_4 = 50\)
Five factors are simulated overall with the following factor scales
\begin{equation*}
\begin{array}{rclllll}
\mat{D}_{12,1}&=&\Diag(0,   & 3.5,  & 0,    & 0, & 4), \\
\mat{D}_{13,1}&=&\Diag(3.1, & 3.15, & 3.05, & 0, & 0), \\
\mat{D}_{13,2}&=&\Diag(3.1, & 3.15, & 3.05, & 0, & 0), \text{ and} \\
\mat{D}_{43,1}&=&\Diag(0,   & 0,    & 3.5,  & 4, & 0).
\end{array}
\end{equation*}
In contrast to Simulation 1, orthogonal factors with 0\%, 25\%, 50\%, and 75\% sparsity were simulated.
Residual noise was simulated as in Simulation 1. The simulation was repeated 50 times for each sparsity level.
Note that singular values are very close in \(\mat{D}_{13,1}\) and \(\mat{D}_{13,2}\) such that, after adding noise, it is likely for factors to be unidentifiable.
The layout above is L-shaped with two-layers in the center. Since MOFA cannot be applied to tensor-like data, we unfold the tensor by columns and apply MOFA vertically with common view 3. This means that \(\mat{X}_{12, 1}\) is not part of the estimation.
solrCMF and MOFA are run with the same settings as for Simulation 1. For the simulation with 0\% sparsity, spike-and-slab priors are turned off for MOFA and solrCMF is setup to estimate structure only.


\begin{figure}[ht]
    \centering
    \includegraphics[trim={1.9cm 0 1cm 0},clip,width=\textwidth]{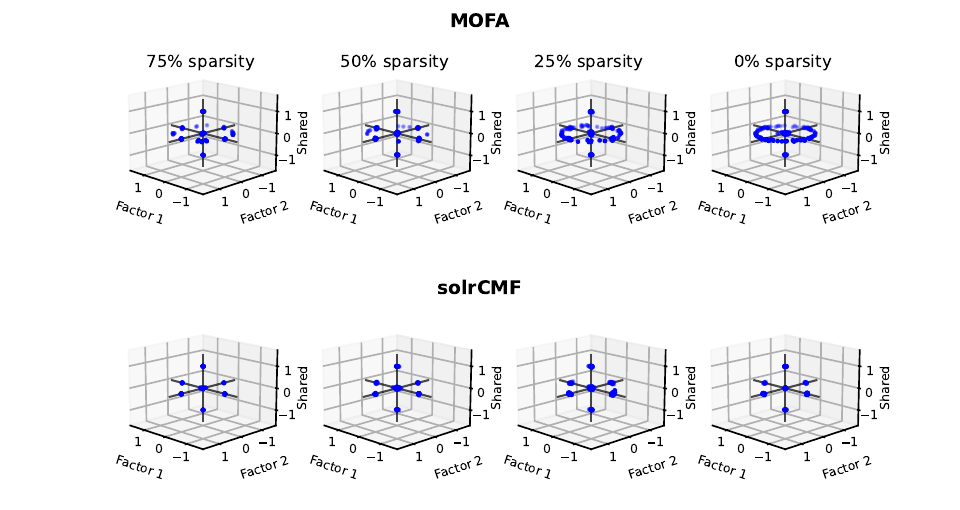}
    \caption{Point clouds of scalar products between estimated and groundtruth factors for view 3 in Simulation 2 are shown. Unrelated pairs of factors result in points around the origin. Axes are shown in black to indicate where factors are correctly identified.}
    \label{fig:factor-identification}
\end{figure}

The similarity between estimated factors and groundtruth for view 3 is then investigated. Here, we did not perform any matching as was done for the factor sparsity result above. Instead, scalar products between all estimated normalized factors and all groundtruth factors were computed and are shown in Figure~\ref{fig:factor-identification}. Unrelated factors appear as points around the origin. If a factor can be identified, then scalar products become close to -1 or 1 along the respective axis. If factor rotation occurs, then the estimated factor will be an approximate linear combination of two or more groundtruth factors. Due to normalization, scalar products are then expected to assume any value on a sphere within the unidentifiable subspace.

It can be seen in Figure~\ref{fig:factor-identification}, particularly for the 0\% sparsity case, that factor rotation occurs. Side information clearly identifies the vertical factor for both methods. The results for MOFA show that increasing sparsity in the groundtruth reduces factor rotation and helps to identify factors 1 and 2. However, even in case of 75\% sparsity there is no guarantee that rotations cannot occur. The results for solrCMF show that additional side information can help to fully identify the rotated factors, even in the case without sparsity. Side information is therefore a much stronger tool in identifying factors and emphasizes the importance of being able to integrate flexible layouts.

\begin{landscape}
\fontsize{10}{12}\selectfont
\begin{table}
\captionsetup{font=scriptsize}
\caption{Numerical comparison of directed \(R^2\) for all pairs of connected matrices in the full layout. The expected groundtruth is shown in absolute value and for each method the mean difference and standard deviation of the estimated directed \(R^2\) and the groundtruth is shown for each method. Positive numbers indicate an overestimation of the directed \(R^2\) and negative numbers indicate an underestimation. Results for competitor methods are shown separately for each unfolding.
Statistical significance levels for the paired Wilcoxon signed rank test between solrCMF and other methods are marked with each entry (* \(< 0.05\), ** \(< 0.01\), *** \(< 0.001\)). Table entries without any asterisk had a significance level \(\geq 0.05\).
The method achieving the smallest mean deviation in absolute value within each row is marked in bold font and color.
\textcolor{uogblue}{Blue} results indicate that either solrCMF achieved the best result and mean deviations are significantly different from all other methods, or, a method with a significantly different mean difference from solrCMF achieved the best result.
In rows containing \textcolor{uoggreen}{green} and \textcolor{uogred}{orange} entries, the best result was not significantly different from the result achieved by solrCMF, or despite solrCMF achieving the best result there were other methods achieving results of similar quality. The overall best entry is marked in green and whereas entries that are potential candidates for solutions are marked in orange.
}
\label{tbl:directed-r2}
\begin{adjustbox}{max width=\linewidth}
\sisetup{detect-weight, table-format=-2.2, mode=text}
\begin{tabular}{llS*{11}{Sr}}
\toprule
& & & \multicolumn{22}{c}{Difference between Expected and Estimate in \% (Mean, Standard Deviation, Significance level)} \\
\cmidrule(lr){4-25}
{Dependent} & {Predictor} & {Expected (\%)} & \multicolumn{2}{c}{Our method} & \multicolumn{4}{c}{CMF} & \multicolumn{4}{c}{MMPCA} & \multicolumn{6}{c}{SLIDE} & \multicolumn{6}{c}{MOFA} \\
\cmidrule(lr){6-9} \cmidrule(lr){10-13} \cmidrule(lr){14-19} \cmidrule(lr){20-25}
&&&&& \multicolumn{2}{c}{Horizontal} & \multicolumn{2}{c}{Vertical} & \multicolumn{2}{c}{Horizontal} & \multicolumn{2}{c}{Vertical} & \multicolumn{2}{c}{Upper} & \multicolumn{2}{c}{Middle} & \multicolumn{2}{c}{Vertical} & \multicolumn{2}{c}{Upper} & \multicolumn{2}{c}{Middle} & \multicolumn{2}{c}{Vertical} \\
\midrule
\input{directed-r2-table-250runs}
\end{tabular}
\end{adjustbox}
\end{table}
\end{landscape}

\clearpage

\section{Discussion}
\label{sec:discussion}

In this paper, we present a novel data integration method, solrCMF, capable of integrating and partitioning signal in a collection of data sources into subspaces of individual variation as well as variation shared, globally or restricted to a subset, with other data sources. Our method allows flexible integration of many data source layouts encompassing multi-view, multi-grid, augmented, and even tensor-like layouts. Interpretable factors are obtained using sparsity-inducing penalties during optimization. The resulting parameter estimation problem is challenging and we present an extension of multi-block ADMM with affine constraints to accommodate embedded manifold constraints. This extension can be be useful for the solution of related optimization problems in the future.

We demonstrate through simulation studies that solrCMF performs well in terms of rank and factor scale estimation in comparison to other methods. Being an optimization-based method and therefore committing to a point estimate, solrCMF tends to find fewer but more relevant entries in sparse factors in some cases and performs comparably with existing methods.

Factor sparsity and side information were both found to be helpful in identifying factors and avoiding rotational invariance. However, side information can improve identifiability even in situations where no sparsity is present. This emphasizes the importance of supporting flexible data layouts such that inclusion of side information is possible.

Requiring orthogonality between factors in every view restricts the range of matrices that can be represented by the solrCMF model. An argument for relaxing this assumption has been made in the case of an L-shaped layout with one central matrix \(\mat{X}_{12}\) and two matrices providing side information \(\mat{X}_{13}\) and \(\mat{X}_{24}\) with all matrices modeled according to \(\mat{X}_{ij} = \mat{V}_i \mat{D}_{ij} \mat{V}_j^\top + \mat{E}_{ij}\).
\citet{OConnell2019} show that in order to identifiably represent any three matrices in such a layout it is sufficient to assume that \(\mat{V}_1\) and \(\mat{V}_2\) have orthogonal columns and \(\mat{D}_{12}\). No assumptions about the other factors and scaling matrices are made. This can be alternatively formulated as choosing all \(\mat{V}_i\) to have orthogonal columns, letting \(\mat{D}_{12}\) be diagonal, and \(\mat{D}_{13}\) and \(\mat{D}_{24}\) to be arbitrary \(k \times k\) matrices.
A similar argument has been made with respect to sparse PCA. \citet{Chen2023} suggest to use the approach
\begin{equation*}
\argmin_{\mat{U}, \mat{V}, \mat{D}} \norm{\mat{X} - \mat{U} \mat{D} \mat{V}^\top}_F^2
\quad\text{such that}\quad
\mat{U} \in \Stiefel_{k, n}, \mat{V} \in \Stiefel_{k, p}, \norm{\mat{V}}_1 \leq \gamma,
\end{equation*}
to estimate sparse principal components. Note that there is no restriction on \(\mat{D}\) and it is therefore an arbitrary \(k \times k\) matrix. To extend the class of data collections that solrCMF can represent it would therefore be necessary to either relax the orthogonality constraint on some \(\mat{V}_i\) or to allow some \(\mat{D}_{ij}\) to be arbitrary \(k \times k\) matrices. It is not immediately clear how such a generalization could be performed, since the solrCMF problem in Eq.~\eqref{eq:sparse-cmf:penalized-model} ties automatic partitioning of variation in data sources intimately to the assumption that all \(\mat{D}_{ij}\) are diagonal. It remains to be explored if allowing correlation between columns of some \(\mat{V}_i\) could extend solrCMF's range of application.

By using a Frobenius loss in Eq.~\eqref{eq:sparse-cmf:penalized-model} and focusing on continuous data, solrCMF indirectly assumes an approximately normal distributed error model. It has been shown before in matrix factorization models as well as data integration models \citep[e.g.,][]{Collins2001,Klami2014,Li2018} that supporting other than continuous error models can be beneficial to improve prediction performance. However, the complexity of parameter estimation typically increases or approximations introducing additional bias need to be used.
To include other data distributions into solrCMF the loss function in Eq.~\eqref{eq:sparse-cmf:penalized-model} needs to be replaced.
This would affect the subproblem for \(\mat{Z}_{ij}\). Replacing the quadratic loss with a Bernoulli or Poisson loss instead, in order to model binary or count data, leads to a subproblem that cannot be solved analytically anymore. A possible solution is to extend the ADMM framework used in this article to support linearized subproblems \citep[e.g.,][]{Ouyang2015}.

\section*{Acknowledgments}

The computations were enabled by resources provided by the National Academic Infrastructure for Supercomputing in Sweden (NAISS) at National Supercomputer Centre (NSC), Linköping University, partially funded by the Swedish Research Council through grant agreement no. 2022-06725.


\bibliographystyle{jasa3}
\bibliography{main}

\end{document}


\def\spacingset#1{\renewcommand{\baselinestretch}%
{#1}\small\normalsize} \spacingset{1}


\if0\blind%
{
  \title{\bf Supplementary Material: Sparse and Orthogonal Low-rank Collective Matrix Factorization (solrCMF): Efficient data integration in flexible layouts}
  \author{Felix Held, Jacob Lindbäck, Rebecka Jörnsten}
  \date{}
  \maketitle
} \fi

\if1\blind%
{
  \bigskip
  \bigskip
  \bigskip
  \begin{center}
    {\LARGE\bf Supplementary Material: Sparse and Orthogonal Low-rank Collective Matrix Factorization (solrCMF): Efficient data integration in flexible layouts}
  \date{}
  \end{center}
  \medskip
} \fi




\tableofcontents


\section{Definitions and useful results}%
\label{apx:sec:defs-res}

\begin{definition}[Multi-affine functions and equations]
\hfill\par
\begin{enumerate}
\item A function \(f(x_1, x_2, \dots, x_n)\) is called multi-affine if
for each \(i = 1, \dots, n\) the function
\(x \mapsto f(x_1, \dots, x_{i - 1}, x, x_{i + 1}, \dots, x_n)\) with
\(x_j\) fixed for \(j \neq i\) is affine.
\item The equation \(f(x_1, \dots, x_n) = 0\) is called multi-affine if the
function \(f\) is multi-affine in \(x_1, \dots, x_n\).
\end{enumerate}
\end{definition}
\begin{definition}[Lipschitz differentiable functions]
A function \(f: \R^n \rightarrow \R\) is called \(M\)-Lipschitz differentiable
if it is differentiable and its gradient is Lipschitz with constant
\(M > 0\).
\end{definition}
\begin{definition}[Strongly convex functions]
A function \(f: \R^n \rightarrow \R\) is called \((m, M)\)-strongly convex if
it is convex, \(M\)-Lipschitz differentiable and if there is \(m > 0\) that
satisfies the inequality
\begin{equation}\label{eq:defs-res:strongly-convex}
f(\vect{y}) \geq f(\vect{x}) +
\left\langle \nabla_{\vect{x}} f(\vect{x}), \vect{y} - \vect{x}\right\rangle +
\frac{m}{2} \norm{\vect{y} - \vect{x}}_2^2
\end{equation}
for all \(\vect{x}\) and \(\vect{y}\).
\end{definition}
\begin{definition}[Vectorization]
Vectorization is a bijective operation that maps matrices in 
\(\R^{n \times m}\) to vectors in \(\R^{nm}\). Given a matrix 
\(\mat{A} \in \R^{n \times m}\), column-wise vectorization assigns the
corresponding vector 
\(\Vectorize(\mat{A}) =
[{\mat{A}^{(:, 1)}}^\top, {\mat{A}^{(:, 2)}}^\top, \dots,
{\mat{A}^{(:, m)}}^\top]^\top\). Row-wise vectorization can be defined
analogously but will not be used here.
\end{definition}
\begin{definition}[Kronecker product]
Given two matrices \(\mat{A} \in \R^{r \times c}\) and \(\mat{B} \in \R^{n \times p}\) the Kronecker product of \(\mat{A}\) and \(\mat{B}\) is defined as
\begin{equation*}
\mat{A} \times \mat{B} =
\begin{pmatrix}
\mat{A}^{(1, 1)} \mat{B} & \dots & \mat{A}^{(1, c)} \mat{B} \\
\vdots & & \vdots \\
\mat{A}^{(r, 1)} \mat{B} & \dots & \mat{A}^{(r, c)} \mat{B} \\
\end{pmatrix}
\in \R^{rn \times cp}.
\end{equation*}
\end{definition}
The following definitions can be found in \citep[Chapter~1]{Rockafellar2009}.
\begin{definition}[Extended value functions]
A function is called an \emph{extended value function} if its image is in the
set \(\overline{\R} := \R \cup \{-\infty, +\infty\}\).
\end{definition}
\begin{definition}[Domain]
The domain of an extended value function \(f: \R^n \rightarrow \overline{\R}\) 
is the set \(\Domain(f) := \{\vect{x} \in \R^n : f(\vect{x}) < \infty\}\).
\end{definition}
\begin{definition}[Proper functions]
An extended value function \(f\) is called proper if \(\Domain(f)\) is
non-empty and \(f(\vect{x}) > -\infty\) for all \(\vect{x}\).
\end{definition}
\begin{definition}[Lower-semicontinuous functions \emph{(lsc)}]
A function \(f: \R^n \rightarrow \overline{\R}\) is called
lower-semicontinuous (lsc) if
\begin{equation*}
\liminf_{\vect{x} \rightarrow \vect{x}_0} f(\vect{x}) = 
f(\vect{x}_0) \text{ for all } \vect{x}_0 \in \Domain(f).
\end{equation*}
\end{definition}
Some useful results for lower-semicontinuous functions.
\begin{lemma}\label{lem:defs-res:lsc-properties}
\begin{enumerate}
\item An indicator function \(\iota_{\mathrm{C}}\) is lsc
if \(\mathrm{C}\) is closed \citep[Remark on p.\ 11]{Rockafellar2009}.
\item The sum of finitely many lsc functions is lsc as long as the term
\(\infty - \infty\) never occurs \citep[Proposition~1.38]{Rockafellar2009}.
\end{enumerate}
\end{lemma}

\begin{definition}[Tangent cone, {\citep[Definition 6.1]{Rockafellar2009}}]
A vector \(\vect{w} \in \R^m\) is tangent to a set \(C \subset \R^m\) at a point \(\bar{\vect{x}} \in C\), written \(\vect{w} \in T_C(\bar{\vect{x}})\), if
\begin{equation*}
\frac{\vect{x}_n - \bar{\vect{x}}}{\tau_n} \rightarrow \vect{w}
\text{ for some } \vect{x}_n \in C, \tau_n > 0 \text{ such that } \vect{x}_n \rightarrow \bar{\vect{x}}, \tau_n \rightarrow 0.
\end{equation*}
The set \(T_C(\bar{\vect{x}})\) is called the \emph{tangent cone} of \(C\) at \(\bar{\vect{x}}\).
\end{definition}

\begin{definition}[Regular and general normal cone, {\citep[Definition 6.3]{Rockafellar2009}}]
Let \(C \subset \R^m\) and \(\bar{\vect{x}} \in C\). A vector \(\vect{v}\) is normal to \(C\) at \(\bar{\vect{x}}\) in the \emph{regular sense}, written \(\vect{v} \in \widehat{N}_C(\bar{\vect{x}})\), if
\begin{equation*}
\limsup_{\vect{x} \in C, \vect{x} \neq \bar{\vect{x}}, \vect{x} \rightarrow \bar{\vect{x}}}\frac{\langle \vect{v}, \vect{x} - \bar{\vect{x}}\rangle}{\norm{\vect{x} - \bar{\vect{x}}}} \leq 0.
\end{equation*}
The vector \(\vect{v}\) is normal to \(C\) at \(\bar{\vect{x}}\) in the \emph{general sense}, written \(\vect{v} \in N_C(\bar{\vect{x}})\), if there are sequences \(\vect{x}_n \in C, \vect{v}_n \in \widehat{N}_C(\vect{x}_n)\) such that \(\vect{x}_n \rightarrow \bar{\vect{x}}\) and \(\vect{v}_n \rightarrow \vect{v}\).
\(\widehat{N}_C(\bar{\vect{x}})\) and \(N_C(\bar{\vect{x}})\) are called the \emph{regular} and \emph{general normal cone} of \(C\) at \(\bar{\vect{x}}\).
\end{definition}
Note that if \(C\) is a convex set then the definition simplifies to
\begin{equation*}
N_C(\bar{\vect{x}}) = \{\vect{v} \,\vert\, \langle \vect{v}, \vect{x} - \bar{\vect{x}} \rangle \leq 0 \text{ for all } \vect{x} \in C\}.
\end{equation*}

\begin{definition}[Regular, general and horizon subgradients, {\citep[Definition 8.3]{Rockafellar2009}}]
Let \(f\) be a proper lsc function and \(\vect{x} \in \Domain(f)\).
A vector \(\vect{v}\) is a regular subgradient of \(f\) at \(\vect{x}\),
indicated by \(\vect{v} \in \widehat{\partial} f(\vect{x})\), if
\(f(\vect{y}) \geq f(\vect{x}) + \langle \vect{v}, \vect{y} - \vect{x}\rangle
+ o(\norm{\vect{y} - \vect{x}})\) for all \(\vect{y} \in \R^n\).
A vector \(\vect{v}\) is a general subgradient, indicated by
\(\vect{v} \in \partial f(\vect{x})\), if there exist sequences
\(\vect{x}_n \rightarrow \vect{x}\) and \(\vect{v}_n \rightarrow \vect{v}\)
with \(f(\vect{x}_n) \rightarrow f(\vect{x})\) and 
\(\vect{v}_n \in \widehat{\partial} f(\vect{x}_n)\). A vector \(\vect{v}\) is a horizon subgradient, indicated by \(v \in \partial^\infty f(\vect{x})\), if there exist sequences \(\vect{x}_n \rightarrow \vect{x}, \lambda_n \rightarrow 0\), and \(\vect{v}_n \in \widehat{\partial} f(x_n)\) with \(f(\vect{x}_n) \rightarrow f(\vect{x})\) and \(\lambda_n \vect{v}_n \rightarrow \vect{v}\).
\end{definition}
The following definitions can be found in \citep[Section 5.5]{Gao2020}.
\begin{definition}[Domain of a subgradient]
Let \(f\) be a proper lsc function. The domain of its subgradient mapping
\(\Domain(\partial f)\) is the set 
\(\{\vect{x} : \partial f(\vect{x}) \neq \emptyset\}\).
\end{definition}
\begin{definition}[K-\L\ property]
A function \(f: \R^n \rightarrow \overline{\R}\) is said to have the
K-\L\ property at \(\vect{x} \in \Domain(\partial f)\) if there exists
\(\eta \in (0, \infty]\), a neighbourhood \(U\) of \(\vect{x}\), and a
continuous concave function \(\phi: [0, \eta) \rightarrow \R\) such that:
\begin{enumerate}
\item \(\phi(0) = 0\)
\item For all \(s \in (0, \eta), \phi'(s) > 0\)
\item For all \(\vect{y} \in U \cap 
\{\vect{w}: f(\vect{x}) < f(\vect{w}) < f(\vect{x}) + \eta\}\), 
the K-\L\ inequality holds:
\begin{equation*}
\phi'\left(f(\vect{y}) - f(\vect{x})\right) \Distance(0, \partial f(\vect{x}))
\geq 1.
\end{equation*}
\end{enumerate}
A proper lsc function \(f\) that satisfies the K-\L\ property at every point
of \(\Domain(\partial f)\) is called a \emph{K-\L\ function}.
\end{definition}
\begin{definition}[Real semi-algebraic sets]\label{def:semi-algebraic-set}
A subset \(\mathcal{I}\) of \(\R^n\) is real semi-algebraic if there exists a
finite number of real polynomial functions 
\(P_{ij}, Q_{ij}: \R^n \rightarrow \R\) such that
\begin{equation*}
\mathcal{I} = \bigcup_{j = 1}^p \bigcap_{i = 1}^q
\left\{\vect{x} \in \R^n: P_{ij}(\vect{x}) = 0, Q_{ij}(\vect{x}) < 0\right\}.
\end{equation*}
\end{definition}
\begin{definition}[Semi-algebraic function]\label{def:semi-algebraic-function}
A function \(f: \R^n \rightarrow \R^m\) is semi-algebraic if its graph \(\{(\vect{x}, \vect{y}) \in \R^{n + m} : f(\vect{x}) = \vect{y}\}\) is a real semi-algebraic subset of \(\R^{n + m}\).
\end{definition}
It has been established \citep[see e.g.\@][]{Bolte2007a,Bolte2007b,Attouch2010,Attouch2013} that semi-algebraic functions are K-\L\ functions.
An projection property of semi-algebraic sets, following from the Tarski-Seidenberg principle \citep{RealAlgebraicGeometry1998}, is the theorem below.
\begin{theorem}[{\citep[Theorem 2.2.1]{RealAlgebraicGeometry1998}}]
Let \(\mathcal{I}\) be a semi-algebraic subset of \(\R^{n + m}\). Then the projection onto the first \(n\) coordinates is a semi-algebraic subset of \(\R^n\).    
\end{theorem}


\section{Initialization strategies}%
\label{sec:init-strategies}

To initialize our algorithm, we propose a general strategy based on random initialization. In addition, an example is given of how the algorithm can be initialized, using a more structured approach for the case of fully observed multiview setups.

A random initialization for our algorithm can be obtained as follows. Specification of the desired maximum rank \(k\) of the model is necessary. First, for each \(l \in \mathcal{V}\), QR factorization \citep[Chapter~5.2]{Golub2013} is performed on \(p_l \times k\) matrices with entries sampled from a standard normal distribution, i.e.\@ \(\tilde{\mat{V}}_l = \mat{Q}_l \mat{R}_l\). The matrix \(\mat{V}_l\) is then initialized with the first \(k\) columns of \(\mat{Q}_l\). For simplicity, \(\mat{U}_l\) is set equal to \(\mat{V}_l\) and \(\mat{V}'_l\) is set equal to all zeros.
The diagonal of \(\mat{D}_{ij, \gamma}\) is sampled from a uniform distribution between -1 and 1. This assumes that data matrices \(\mat{X}_{ij, \gamma}\) are preprocessed to have Frobenius norm 1.
The matrix \(\mat{Z}_{ij,\gamma}\) is set equal to \(\mat{V}_i \mat{D}_{ij,\gamma} \mat{V}_j^\top\).
All Lagrange multipliers are set equal to zero.

Due to the high dimensionality of the problem, it is unlikely that a single random initialization will perform well. We therefore generate \(n_{\text{init}}\) initial states and use them to estimate the model repeatedly without any penalization.
Finally, we choose the initial state which led to the minimal objective value after optimization.
It may be necessary to use a larger convergence threshold for these initial estimates to avoid that the algorithm is getting trapped in an uninformative local minimum.
The best initial state is then used during hyperparameter selection. 

Some data layouts allow the generation of structured initial states. One example is in case of multiview data. Assume \(n\) different data matrices are observed, sharing a common column view\footnote{The description below holds analogously for a common row view.}. Data matrices are concatenated row-wise and an SVD is performed on this concatenated matrix \(\mat{X}_C = \mat{U}_C \mat{D}_C \mat{V}_C^\top\). \(\mat{V}_1\), assuming that the common view is view 1, is easily found by taking the first \(k\) columns of \(\mat{V}_C\). To initialize \(\mat{V}_l\) for a row view \(l\), \(\mat{U}_C\) is split into sections in accordance to the dimensions of the input data matrices, and \(\mat{V}_l\) is initialized by taking the first \(k\) columns of the corresponding section. \(\mat{D}_{ij, \gamma}\) are initialized to the diagonal of \(\mat{V}_i^\top \mat{X}_{ij, \gamma} \mat{V}_j\). As for random initialization, \(\mat{U}_l\) can be set equal to \(\mat{V}_l\), and \(\mat{V}'_l\) can be set to zero. All Lagrange multipliers are set equal to zero.
\citet{Gaynanova2019} use a variant of this approach to initialize their method.
While easy to perform this approach can over-emphasize joint components, hiding individual or partially shared components. Large matrices and those with strong signals will dominate the concatenated matrix. Proper normalization of input data is therefore crucial.


\section{Simulating orthogonal factors}

To simulate dense orthogonal factors \(\mat{V} \in \R^{p \times k}\) without sparsity, first \(\mat{A} \in \R^{p \times k}\) is randomly generated by simulating \(pk\) independent values from a standard normal distribution and arranging them as a matrix. QR factorization \citep[Chapter 5.2]{Golub2013} is then applied to perform orthogonalization as and the first \(k\) columns of \(\mat{Q}\) are used to generate \(\mat{V}\). This procedure generates uniformly distributed random samples on the Stiefel manifold \citep{Stewart1980}.

To simulate sparse orthogonal factors with q\% sparsity, we sample, as above, \(\mat{A} \in \R^{p \times k}\) elementwise independently from a standard normal distribution.
Only the q\% largest entries in each column are kept, all other entries are set to zero. Keeping only the largest entries improves numerical stability and generates clearer sparsity patterns.
QR factorization for orthogonalization as above would introduce non-zero entries into zero locations again. We therefore present a variant of the classical Gram-Schmidt algorithm \citep[5.2.7]{Golub2013} which preserves zeros.
We proceed as follows. For convenience, denote columns of \(\mat{A}\) by \(\bm{a}_i\). We start by normalizing the first column and set \(\bm{v}_1 = \bm{a}_1 / \norm{\bm{a}_1}_2\). Then for columns \(\bm{a}_i\) for \(i = 2, \dots, k\), determine \(\mathrm{I}_i = \{l : \bm{a}_i^{(l)} \neq 0\}\). For any column \(j = 1, \dots, i - 1\), check if \(\bm{v}_j^{(\mathrm{I}_i)} \neq \bm{0}\). If so, update \(\bm{a}_i \leftarrow \bm{a}_i - \bm{a}_i^\top \bm{v}_j / \norm{\bm{v}_j^{(\mathrm{I}_i)}}_2 \bm{v}_j^{(\mathrm{I}_i)}\). This step orthogonalizes \(\bm{a}_i\) to \(\bm{v}_j\) without changing the zero pattern. Finally, normalize \(\bm{a}_i\) and set \(\bm{v}_i = \bm{a}_i / \norm{\bm{a}_i}_2\). The matrix of sparse orthogonal factors \(\mat{V}\) is obtained by concatenating vectors \(\bm{v}_i\) column-wise.


\section{Derivations of subproblem solutions}


\subsection{Subproblems for \texorpdfstring{\(\mat{\Theta}\)}{Theta}}%
\label{ssec:m-admm:subproblems-theta}


\subsubsection{Subproblem for \texorpdfstring{\(\mat{V}\)}{V}}%
\label{sssec:m-admm:subproblem-v}

The subproblems for \(\mat{V}_i\) have to be solved sequentially since they
depend on other \(\mat{V}_l\) for \(l \neq i\). Note that some \(\mat{V}_l\)
will be updated before \(\mat{V}_i\) and some after.
This is not indicated below to avoid excessive notation.
\(\mat{V}_i\) is restricted to be in the Stiefel manifold \(\Stiefel_{k, p_i}\)
and it therefore holds for arbitrary but fixed matrices
\(\mat{A} \in \R^{p_i \times p}\), \(\mat{B} \in \R^{k \times p}\), and
\(\mat{C} \in \R^{p_i \times k}\) that
\begin{align*}
\norm{\mat{A} - \mat{V}_i \mat{B}}_F^2 &=
-2 \Trace(\mat{B} \mat{A}^\top \mat{V}_i) + \text{const}, \\
\norm{\mat{A} - \mat{B} \mat{V}_i^\top}_F^2 &=
-2 \Trace(\mat{B}^\top \mat{A} \mat{V}_i) + \text{const}, \text{ and } \\
\norm{\mat{V}_i - \mat{C}}_F^2
&= -2 \Trace(\mat{C}^\top \mat{V}_i) + \text{const},
\end{align*}
where the constants are independent of \(\mat{V}_i\).
The optimization problem that needs to be solved to update \(\mat{V}_i\) is
\begin{equation*}
\begin{aligned}
\mat{V}_i^+ = \argmin_{\mat{V} \in \Stiefel_{k, p_i}}\quad 
&\frac{\rho}{2} \sum_{(i, l, \gamma) \in \mathcal{I}} 
\norm{\mat{Z}_{il,\gamma}^- - \mat{V} \mat{D}_{il,\gamma}^- \mat{V}_l^\top
+ \mat{\Lambda}^{(2), -}_{il,\gamma}}_F^2 \\
+ &\frac{\rho}{2} \sum_{(l, i, \gamma) \in \mathcal{I}} 
\norm{\mat{Z}_{li,\gamma}^- - \mat{V}_l \mat{D}_{li,\gamma}^- \mat{V}^\top 
+ \mat{\Lambda}^{(2), -}_{il,\gamma}}_F^2 \\
+ &\frac{\rho}{2} \norm{\mat{U}_i^- - \mat{V} - {\mat{V}'_i}^-
+ \mat{\Lambda}^{(1), -}_i}_F^2. 
\end{aligned}
\end{equation*}
This can be reformulated using the identities above to be
\begin{equation}\label{eq:m-admm:unbalanced-procrustes}
\mat{V}_i^+ = \argmin_{\mat{V} \in \Stiefel_{k, p_i}}\quad
\frac{1}{2} \norm{\mat{V} - \mat{M}}_F^2,
\end{equation}
where
\begin{equation*}
\begin{aligned}
\mat{M} = &\sum_{(i, l, \gamma) \in \mathcal{I}} 
(\mat{Z}_{il,\gamma}^- + \mat{\Lambda}^{(2), -}_{il,\gamma}) \mat{V}_l \mat{D}_{il,\gamma}^-
+ \sum_{(l, i, \gamma) \in \mathcal{I}} 
(\mat{Z}_{li,\gamma}^- + \mat{\Lambda}^{(2), -}_{li,\gamma})^\top \mat{V}_l \mat{D}_{li,\gamma}^-
\\
&+ (\mat{U}_i^- - {\mat{V}'_i}^- + \mat{\Lambda}^{(1), -}_i).
\end{aligned}
\end{equation*}
The optimization problem in Eq.~\eqref{eq:m-admm:unbalanced-procrustes} is
known as the unbalanced Procrustes problem \citep{Zhang2007}.
To find a solution, obtain the full SVD of
\(\mat{M} = \mat{B}_1 \mat{\Sigma} \mat{B}_2^\top\) where
\(\mat{B}_1 \in \R^{p_i \times p_i}\) and \(\mat{B}_2 \in \R^{K \times K}\)
are orthonormal matrices, and
\(\mat{\Sigma} = (\mat{\Sigma}_0, \mat{0}_{K \times p_i})^\top
\in \R^{p_i \times K}\) with  \(\mat{\Sigma}_0 \in \R^{K \times K}\).
Set then
\(\mat{Q} = \mat{B}_1 \mat{I}_{p_i, K} \mat{B}_2^\top\) where
\(\mat{I}_{p_i, K}\) are the first \(K\) columns of the unit matrix
\(\mat{I}_{p_i}\). \(\mat{Q}\) clearly has orthonormal columns and fulfills
\(\mat{Q} + \mat{Q} \mat{\Lambda} = \mat{M}\) with the symmetric matrix 
\(\mat{\Lambda} = \mat{B}_2 (\mat{\Sigma}_0 - \mat{I}_K) \mat{B}_2^\top\). 
Following \citet{Zhang2007},~Theorem~4.1, it is therefore sufficient that the
diagonal entries of \(\mat{\Sigma}_0\) are greater or equal to zero for 
\(\mat{Q}\) to be a globally optimal solution to
Eq.~\eqref{eq:m-admm:unbalanced-procrustes}.
Since the diagonal entries of \(\mat{\Sigma}_0\) contain the singular values
of \(\mat{M}\), which are always non-negative, this is clearly fulfilled.
\(\mat{Q}\) is the unique solution to
Eq.~\eqref{eq:m-admm:unbalanced-procrustes} if all diagonal entries of
\(\mat{\Sigma}_0\) are strictly greater than zero. Columns corresponding to singular values equal to zero are not necessarily unique. However, these columns are not relevant for further analysis and non-uniqueness is not problematic.

Set \(\mat{V}_i^+ = \mat{Q}\) to finish the update in this subproblem.


\subsubsection{Subproblem for \texorpdfstring{\(\mat{D}\)}{D}}%
\label{sssec:m-admm:subproblem-d}

The minimization problem for \(\mat{D}_{ij,\gamma}\) is separable across \((i, j, \gamma) \in \mathcal{I}\). Using the column orthogonality of \(\mat{V}_l\) for \(l \in [m]\) it follows that
\begin{equation*}
\begin{aligned}
\mat{D}_{ij,\gamma}^+ &= \argmin_{\mat{D} \in \mathcal{D}_k}
\frac{\rho}{2} \norm{\mat{Z}_{ij,\gamma}^- - \mat{V}_i^+ \mat{D} {\mat{V}_j^+}^\top
+ \mat{\Lambda}^{(2), -}_{ij,\gamma}}_F^2
+ \lambda_1 \sum_{l = 1}^k \vect{w}_{ij,\gamma}^{(l)} \abs{\mat{D}^{(l, l)}}_1 \\
&= \argmin_{\mat{D} \in \mathcal{D}_k}
\frac{1}{2} \norm{\mat{D} - \mat{M}}_F^2
+ \frac{\lambda_1}{\rho} \sum_{l = 1}^k
\vect{w}_{ij,\gamma}^{(l)} \abs{\mat{D}^{(l, l)}},
\end{aligned}
\end{equation*}
where
\begin{equation*}
\mat{M} = {\mat{V}_i^+}^\top 
(\mat{Z}_{ij,\gamma}^- + \mat{\Lambda}^{(2), -}_{ij,\gamma}) {\mat{V}_j^+}.
\end{equation*}
Replacing the set constraint with an indicator function, the solution of the above optimization problem is the proximal operator of the function \(\mat{D} \mapsto \iota_{\mathcal{D}_k}(\mat{D}) + \frac{\lambda_1}{\rho} \sum_{l = 1}^k \vect{w}_{ij,\gamma}^{(l)} \abs{\mat{D}^{(l, l)}}\) evaluated at \(\mat{M}\).
To determine the functional form of the proximal operator the following theorem can be used.
\begin{theorem}\label{thm:m-admm:prox-op}
Let \(f\) be a closed convex proper function then the proximal operator
of \(\iota_{\mathcal{D}_k} + f\) is
\(\Proximal_{\iota_{\mathcal{D}_k}} \circ \Proximal_f\).
\end{theorem}
\begin{proof}
Note that \(\mathcal{D}_k\) is a convex closed set and therefore, by
\citet[p.~310]{Rockafellar2009}, the subgradient of 
\(\iota_{\mathcal{D}_k}\) at any \(\mat{D} \in \mathcal{D}_k\) is equal to the normal cone
of \(\mathcal{D}_k\) at \(\mat{D}\), i.e.\@
\begin{equation*}
\begin{aligned}
\mathrm{N}_{\mathcal{D}_k}(\mat{D}) &= \left\{
    \mat{G} \in \R^{k \times k} :
    \langle\mat{G}, \mat{D}\rangle_F \geq
    \langle\mat{G}, \mat{A}\rangle_F
    \text{ for all }
    \mat{A} \in \mathcal{D}_k
\right\} \\
&= \left\{
    \mat{G} \in \R^{k \times k} :
    \mat{G}^{(i, i)} = 0 \text{ for } i \in [k]
\right\}.
\end{aligned}
\end{equation*}
Note that the subgradient is independent of \(\mat{D}\). In addition,
the function \(\mat{D} \mapsto \norm{\mat{D}}_1\) is convex and closed%
\footnote{%
\textbf{Definition} (\emph{closed function}): The epigraph of \(f\), i.e.\@
the set \(\{(\vect{x}, t) \in \R^{n + 1} :
\vect{x} \in \Domain f,\ f(\vect{x}) \leq t\}\), is closed.}, since it is
continuous on its closed domain (\(\R^{k \times k}\)).
Since \(f\) is assumed to be closed convex and proper, Theorem~1 in
\citet{Yu2013} then shows that the proximal operator of
\(\iota_{\mathcal{D}_k} + f\) is
\(\Proximal_{\iota_{\mathcal{D}_k}} \circ \Proximal_{f}\).
\end{proof}
The proximal operator of \(\iota_{\mathcal{D}_k}\) is simply the
projection of \(\mat{D}\) onto the diagonal matrices by setting all
non-diagonal entries to zero. The proximal operator of \(\beta \abs{x}\)
is the soft-thresholding operator
\(\mathrm{ST}_\beta(x) = \Sign(x) \max(\abs{x} - \beta, 0)\).

The updated iterate for \(\mat{D}_{ij,\gamma}\) is then obtained by setting
\(\mat{D}_{ij,\gamma}^{(r, c)} = 0\) for \(r \neq c\) and
\(\mat{D}_{ij,\gamma}^{(l, l)} =
\mathrm{ST}_{\lambda_1 \vect{w}_{ij,\gamma}^{(l)} / \rho}(\mat{M}^{(l, l)})\)
for \(l \in [k]\).


\subsubsection{Subproblem for \texorpdfstring{\(\mat{U}\)}{U}}%
\label{sssec:m-admm:subproblem-u}

The minimization sub-problem for \(\mat{U}_i\) is separable across \(i \in [m]\)
and can be solved in parallel. The optimization problem that needs to
solved for each \(i \in [m]\) is
\begin{equation*}
\mat{U}_i^+ = \argmin_{\mat{U} \in \Oblique_{k, p}}
\frac{1}{2} \norm{\mat{U} - \mat{M}_i}_F^2 
+ \frac{\lambda_2 \mat{w}_1^{(i)}}{\rho} \norm{\mat{U}}_1,
\end{equation*}
where \(\mat{M}_i = \mat{V}_i^+ + {\mat{V}'_i}^- - \mat{\Lambda}_i^{(1), -}\). 
Note that the minimization problem for \(\mat{U}_i\) is furthermore separable
across its columns. Denote an arbitrary column of \(\mat{M}_i\) by 
\(\vect{m} \in \R^{p_i}\) and let \(\vect{u} \in \R^{p_i}\) be the respective
column in \(\mat{U}_i\) which is to be updated.
Denote \(w = \lambda_2 \vect{w}_1^{(i)} / \rho\) and after
simplification, the optimization problem becomes
\begin{equation}\label{eq:m-admm:prox-on-oblique}
\argmin_{\vect{u}} -\vect{u}^\top \vect{m} + w \norm{\vect{u}}_1
\ \text{such that } \vect{u}^\top \vect{u} = 1.
\end{equation}

\begin{theorem}\label{thm:m-admm:prox-on-oblique}
The optimization problem in Eq.\@~\eqref{eq:m-admm:prox-on-oblique} has the
following optimal solutions \(\vect{u}^*\):
\begin{enumerate}
\item If \(\vect{m} = 0\), then any \(\vect{u}^* = \pm\vect{e}_l\)
for any \(l \in [p_i]\) is an optimal solution.
\item If \(\vect{m} \neq 0\), set \(w' = \max_i \abs{\vect{m}^{(i)}}\) and
define \(\mathrm{J}(\vect{m}) = \{j : \abs{\vect{m}^{(j)}} = w'\}\).
Then one of following three cases holds:
\begin{enumerate}
\item If \(w' > w\), then \(\vect{u}^* = \mathrm{ST}_w(\vect{m}) / 
\norm{\mathrm{ST}_w(\vect{m})}_2\) is the unique optimal solution.
\item If \(w' = w\), then all optimal solutions can be written as the
component-wise products \({\vect{u}^*}^{(j)}
= \vect{c}^{(j)} \Sign(\vect{m}^{(j)})\), where \(\vect{c} \in \R^{p_i}\)
is a non-negative vector such that \(\vect{c}^{(j)} = 0\) if 
\(j \not \in \mathrm{J}(\vect{m})\) and \(\vect{c}^{(j)} > 0\) for at least
one \(j \in \mathrm{J}(\vect{m})\).
\item If \(w' < w\), then any \(\vect{u}^* = \Sign(\vect{m}^{(j)}) \vect{e}_j\)
for any \(j \in \mathrm{J}(\vect{m})\) is an optimal solution.
\end{enumerate}
\end{enumerate}
\end{theorem}
\begin{proof}
If \(\vect{m} = \vect{0}\), then any vector \(\vect{u} = \pm \vect{e}_l\)
for \(l \in [p_i]\) is a solution to the minimization problem.
The reason for this is that \(1 = \norm{\vect{u}}_2 \leq \norm{\vect{u}}_1\) for
any unit vector \(\vect{u}\) with equality for standard unit vectors.
The optimization problem therefore has \(2 p_i\) solutions.

In the following, assume that \(\vect{m} \neq \vect{0}\) and define
\(w' = \max_i \abs{\vect{m}^{(i)}}\).
Introducing a Lagrange multiplier \(\delta \in \R\) and forming the
subgradients with respect to \(\vect{u}\) and \(\delta\) leads to the
following stationarity conditions
\begin{equation}\label{eq:stationary-conds}
\begin{aligned}
\vect{0} &\in -\vect{m} + 2 \delta \vect{u} 
+ w \partial_{\vect{u}} \norm{\vect{u}}_1, \\
0 &= \vect{u}^\top \vect{u} - 1.
\end{aligned}
\end{equation}
Since the gradient of the constraint \(2 \vect{u} \neq \vect{0}\) for all
feasible \(\vect{u}\), the Linear Independence Constraint
Qualification (LICQ) is fulfilled. This ensures that for every optimal point
of the original optimization point there exists a \(\delta\) such that
the stationarity conditions above are fulfilled.
The first stationarity condition in Eq.~\eqref{eq:stationary-conds}
is equivalent to
\begin{equation}\label{eq:equiv-cond1}
\begin{aligned}
-2\delta \vect{u}^{(i)} + \vect{m}^{(i)} = -w
& \ \text{if } \vect{u}^{(i)} < 0 \\
\vect{m}^{(i)} \in [-w, w] & \ \text{if } \vect{u}^{(i)} = 0 \\
-2\delta \vect{u}^{(i)} + \vect{m}^{(i)} = w
&\  \text{if } \vect{u}^{(i)} > 0 \\
\end{aligned}
\end{equation}
If \(w' > w\), then it follows from Eq.~\eqref{eq:equiv-cond1} that
every stationary point fulfills
\(\vect{u}^{(i)} = (\vect{m}^{(i)} + w) / (2\delta)\) if
\(\vect{m}^{(i)} < -w\) or
\(\vect{u}^{(i)} = (\vect{m}^{(i)} - w) / (2\delta) > 0\)
if \(\vect{m}^{(i)} > w\) for at least one component \(i\).
Therefore the subgradient in the first stationarity condition contains the vector
\(\vect{u} = \mathrm{ST}_w(\vect{m}) / (2\delta)\).
Since \(\vect{u} \neq \vect{0}\), the second stationarity condition in 
Eq.~\eqref{eq:stationary-conds} implies that
\begin{equation*}
(\vect{u}, \delta) =
(\mathrm{ST}_w(\vect{m}) / \norm{\mathrm{ST}_w(\vect{m})}_2,
\norm{\mathrm{ST}_w(\vect{m})}_2 / 2)
\end{equation*}
is the only stationary point. Note that there will always be at least one
global minimizer \(\vect{u}^*\), since the constraint set is compact and the
objective function is continuous. In addition, LICQ is fulfilled and there
exists a Lagrange multiplier \(\delta^*\) such that \((\vect{u}^*, \delta^*)\)
is a stationary point.
Since there is only one stationary point, it has to be the global minimum of the
original constrained optimization problem in Eq.~\eqref{eq:m-admm:prox-on-oblique}.

If \(0 < w' \leq w\), then \(\mathrm{ST}_w(\vect{m}) = \vect{0}\) and the arguments
above do not hold. Instead, the existence of global minimizers will be shown in
a direct way. Define \(\mathrm{J}(\vect{m}) = \{j : \abs{\vect{m}^{(j)}} = w'\}\).
Two cases will be considered.

\textbf{Case i.} (\(w' = w\)) In this case, any solution can be characterized
as \(\vect{u}^{(i)} = \vect{c}^{(i)} \Sign(\vect{m}^{(i)})\) where \(\vect{c}\)
is an element-wise non-negative unit vector such that \(\vect{c}^{(i)} = 0\) if
\(i \not \in \mathrm{J}(\vect{m})\) and \(\vect{c}^{(i)} > 0\) for at least one 
\(i \in \mathrm{J}(\vect{m})\). Each such vector has objective function value
\begin{equation*}
-\sum_i \vect{c}^{(i)} \abs{\vect{m}^{(i)}} + w \norm{\vect{u}}_1 =
w \sum_{i \in \mathrm{J}(\vect{m})} (\vect{c}^{(i)} - \vect{c}^{(i)}) = 0.
\end{equation*}
Note that for any vector \(\vect{u}\), using the triangle inequality and
Hölder's inequality, it holds that
\begin{equation*}
\abs{\vect{u}^\top \vect{m}}
\leq \sum_i \abs{\vect{u}^{(i)} \vect{m}^{(i)}}
\leq \norm{\vect{u}}_1 \norm{\vect{m}}_\infty =
\norm{\vect{u}}_1 w'.
\end{equation*}
Equality holds only for unit vectors \(\vect{u}\) that are non-zero on
components in \(\mathrm{J}(\vect{m})\) and zero otherwise. This can be
seen as follows. Equality in Hölder's inequality holds if
\begin{equation*}
0 = \sum_i \abs{\vect{u}^{(i)} \vect{m}^{(i)}}
- \sum_i \abs{\vect{u}^{(i)}} w'
= \sum_i \abs{\vect{u}^{(i)}} (\abs{\vect{m}^{(i)}} - w'),
\end{equation*}
which implies that \(\vect{u}^{(i)} = 0\) if \(\vect{m}^{(i)} \neq w'\) and
the entries where \(\vect{m}^{(i)} = w\) need to be chosen such that their
\(\ell_2\) norm is 1. Note that any other unit vector with non-zero components in 
\(\mathrm{J}(\vect{m})\) with sign pattern different from \(\vect{m}\) have
objective function value larger than 0, since
\begin{equation*}
\begin{aligned}
-\vect{u}^\top \vect{m} + w \norm{\vect{u}}_1 &=
\sum_{i \in \mathrm{J}(\vect{m})}
\left(-\vect{u}^{(i)} \vect{m}^{(i)} + w \abs{\vect{u}^{(i)}}\right) \\
&= \sum_{\substack{i \in \mathrm{J}(\vect{m})\\%
\Sign(\vect{u}^{(i)}) = \Sign(\vect{m}^{(i)})}}
\left(-\abs{\vect{u}^{(i)}} \abs{\vect{m}^{(i)}}
+ w \abs{\vect{u}^{(i)}}\right) \\
&+ \sum_{\substack{i \in \mathrm{J}(\vect{m})\\%
\Sign(\vect{u}^{(i)}) \neq \Sign(\vect{m}^{(i)})}}
\left(\abs{\vect{u}^{(i)}} \abs{\vect{m}^{(i)}}
+ w \abs{\vect{u}^{(i)}}\right) \\
&= \sum_{\substack{i \in \mathrm{J}(\vect{m})\\%
\Sign(\vect{u}^{(i)}) \neq \Sign(\vect{m}^{(i)})}} 2w \abs{\vect{u}^{(i)}}
\end{aligned}
\end{equation*}
which is larger than zero for any \(\vect{u}\) which is not of the optimal form
given above. This shows that there are infinitely many global minimizers
in this case.

\textbf{Case ii.} (\(w' < w\)) In this case, any solution is of the form
\(\vect{u} = \Sign(\vect{m}^{(j)}) \vect{e}_j\) for an arbitrary
\(j \in \mathrm{J}(\vect{m})\) with objective function value \(-w' + w > 0\).
Clearly, the objective function value is larger for any \(\vect{u} = \pm \vect{e}_l\) 
for \(l \in [p_i]\) not of the specific form above.
If \(\vect{u} \neq \pm \vect{e}_l\) for any \(l \in [p_i]\) the following
equivalence holds:
\begin{equation}\label{eq:case2-equiv}
-\vect{u}^\top \vect{m} + w \norm{\vect{u}}_1 > -w' + w
\quad\Leftrightarrow\quad
w > \frac{\vect{u}^\top \vect{m} -w'}{\norm{\vect{u}}_1 - 1}
\end{equation}
Arguing with Hölder's inequality again shows that
\begin{equation*}
\abs{\vect{u}^\top \vect{m}}
\leq \norm{\vect{u}}_1 \norm{\vect{m}}_\infty = \norm{\vect{u}}_1 w'
\end{equation*}
which ensures that the right hand side of Eq.~\eqref{eq:case2-equiv} holds which
gives the result. Therefore, there is a finite number of
solutions of the form \(\vect{u} = \Sign(\vect{m}^{(j)}) \vect{e}_j\) for an
arbitrary \(j \in \mathrm{J}(\vect{m})\).
\end{proof}

Theorem 5 shows that only in Case 2 (a) a unique optimal solution exists.
However, it is the only case that is of practical interest. The columns of 
\(\mat{M}_i\) are mostly controlled by the columns of \(\mat{V}_i^+\).
Since the columns of \(\mat{V}_i^+\) are of length one, there is at least one component that is larger than \(1 / \sqrt{p_i}\). As long as \(w \sqrt{p_i} < 1 \Leftrightarrow \lambda_2 < \rho/(\sqrt{p_i} \bm{w}_1^{(i)})\) it is therefore in practice unlikely that any of the other cases appears.
In our experience, \(w \sqrt{p_i} \geq 1\) indicates that \(\lambda_2\) has been chosen too large.

If Cases 1, 2 (b), or 2 (c) do occur, the standard unit vector minimizing the
objective function is chosen. It can easily be determined by determining \(l\)
such that \(-\abs{\vect{m}^{(l)}} + \lambda_2 \vect{w}_1^{(i)} / \rho\) is
minimal. The sign is given by the sign of \(\vect{m}^{(l)}\).
If there are ties, the first component determines the direction and sign.


\subsection{Subproblem for \texorpdfstring{\(\mat{\Delta}\)}{Delta}}%
\label{ssec:m-admm:subproblems-delta}

The subproblem for \(\mat{\Delta}\) is reduced to
\begin{equation}
\begin{aligned}
\mat{\Delta}^+ = \argmin_{\mat{\Delta}} \quad&
\frac{1}{2} \sum_{(i, j, \gamma) \in \mathcal{I}}
\norm*{\mathcal{P}_{\mathrm{\Omega}_{ij, \gamma}}\left(\mat{X}_{ij,\gamma} - \mat{Z}_{ij,\gamma}\right)}_2^2
+ \frac{\mu}{2} \sum_{l = 1}^m \vect{w}_2^{(l)} \norm{\mat{V}'_l}_F^2 \\
+\; &\frac{\rho}{2} \sum_l
\norm{\mat{U}_l^+ - \mat{V}_l^+ - \mat{V}'_l
+ \mat{\Lambda}^{(1),-}_{l}}_F^2 \\
+\; &\frac{\rho}{2} \sum_{(i, j, \gamma) \in \mathcal{I}}
\norm{\mat{Z}_{ij,\gamma} - \mat{V}_i^+ \mat{D}_{ij,\gamma}^+ {\mat{V}_j^+}^\top
+ \mat{\Lambda}^{(2),-}_{ij,\gamma}}_F^2. \\
\end{aligned}
\end{equation}
The minimization problem is separable in all involved variables and admits the following analytic solution
\begin{align}
{\mat{V}'_i}^+ &=  \frac{\rho}{\rho + \mu \vect{w}_2^{(i)}}
\left(\mat{U}_i^+ - \mat{V}_i^+ + \mat{\Lambda}_i^{(1),-}\right), \\
{\mat{Z}_{ij,\gamma}^{(r, c)}}^+ &=  \frac{\rho}{\rho + 1}
\left(\mat{V}_i^+ \mat{D}_{ij,\gamma}^+ {\mat{V}_j^+}^\top
- \mat{\Lambda}_{ij,\gamma}^{(2),-}\right)^{(r, c)}
+ \frac{1}{\rho + 1} \mat{X}_{ij,\gamma}^{(r, c)}
\end{align}
if entry \((r, c)\) is observed and \({\mat{Z}_{ij,\gamma}^{(r, c)}}^+ = \left(\mat{V}_i^+ \mat{D}_{ij,\gamma}^+ {\mat{V}_j^+}^\top
- \mat{\Lambda}_{ij,\gamma}^{(2),-}\right)^{(r, c)}\) otherwise.


\section{Embedded manifold set constraints}

The framework in \citet{Gao2020} was developed for multi-affine constraints.
However, some parts of the theory can be applied to more general constraints as long as they fulfill certain conditions.
By modifying the theory presented in \citep{Gao2020} accordingly, it is shown below that the sequential ADMM updates can be augmented with embedded manifold set constraints as long as each subproblem can be solved exactly.
We first present results from optimization theory that are necessary for our arguments below.

Assume that \(Y \in \R^n\) is constrained by \(m\) equations (with \(n > m\)) given by a smooth function \(\R^n \rightarrow \R^m, X \mapsto \Phi(X)\) such that \(Y \in \mathcal{M} = \{X \in \R^n: \Phi(X) = 0\}\). Let \(\mathrm{D}\Phi(Y) \in \R^{m \times n}\) denote the Jacobian of \(\Phi\) at \(Y\). The assumption is made that \(\mathrm{D}\Phi(Y)\) is full-rank (i.e.\@ rank \(m\)) for \(Y \in \mathcal{M}\), which implies that the gradients \(\{\nabla_Y \Phi({Y})^{(i, :)}\}_{i=1}^m\) are linearly independent. The set \(\mathcal{M}\) is then a \emph{smooth embedded manifold} \citep[Chapter 8]{Lee2012}.
Adding the set constraint \(Y \in \mathcal{M}\) for variable \(Y\) in a non-linear programming problem is formally equivalent to including the non-linear equality constraints \(\Phi(Y) = 0\). 

Below we will only be dealing with at most one smooth embedded manifold constraint within each subproblem and we will continue to assume, as in \citep{Gao2020}, that subproblems can be solved exactly. Therefore, we are avoiding to use tools from the manifold optimization literature \citep{Absil2008,Yang2014} and instead apply standard results from non-linear optimization \citep{Rockafellar2009} directly.

A reminder of necessary stationarity conditions for the minimization of a proper lsc function \(f: \R^n \rightarrow \overline{\R}\) under manifold set constraints is given first. Let \(\mathcal{M}\) be a smooth embedded manifold as described above. By \citep[Example 6.8]{Rockafellar2009} it then follows that
\begin{equation*}
N_{\mathcal{M}}(\bar{X}) = \{V \in \R^n \,\vert\, V = \mathrm{D}\Phi(\bar{X})^\top Y \text{ for } Y \in \R^m\}.
\end{equation*}
In addition, note that \(\mathcal{M}\) is a closed set by continuity of \(\mathrm{D}\Phi(X)\). Consider the minimization problem
\begin{equation}\label{eq:lsc-argmin}
\argmin_{X \in \mathcal{M}} f(X).
\end{equation}
Stationarity conditions for Eq.~\eqref{eq:lsc-argmin} are then described by the following theorem.
\begin{theorem}[{\citep[Theorem 8.15]{Rockafellar2009}}]%
\label{thm:general-stationarity}
Let \(f\) be a proper, lsc function. Consider minimizing \(f\) over a closed set \(C \subset \R^n\). Let \(\bar{x} \in C\) be such that: the set \(\partial^\infty f(\bar{x})\) contains no vector \(v \neq 0\) such that \(-v \in N_C(\bar{x})\). Then for \(\bar{x}\) to be locally optimal it is necessary that \(0 \in \partial f(\bar{x}) + N_C(\bar{x})\).
\end{theorem}
As described above, \(\mathcal{M}\) is a closed set and the normal cone can be explicitly determined. Below, we will be interested in two cases: (1) \(f\) is smooth or (2) \(f\) is a proper, lsc, convex function.
For (1) it holds that if \(f\) is smooth on a neighborhood of \(\bar{x}\), then \(\partial f(\bar{x}) = \{\nabla f(\bar{x})\}\) and \(\partial^\infty f(\bar{x}) = \{0\}\) \citep[p. 304]{Rockafellar2009}. For (2) it holds that if \(f\) is a proper, lsc, convex function \(\R^n \rightarrow \bar{\R}\) and \(\bar{x} \in \Domain f\) then \(\partial^\infty f(\bar{x}) = N_{\Domain f}(\bar{x})\) \citep[Proposition 8.12]{Rockafellar2009}. Typically, we will be dealing with the case \(\Domain f = \R^n\), which implies that \(N_{\Domain f}(\bar{x}) = N_{\R^n}(\bar{x}) = \{0\}\).

In both cases, only \(0 \in \partial^\infty f(\bar{x})\) which by Theorem~\ref{thm:general-stationarity} implies that the stationarity conditions for \(X \in \mathcal{M}\) are \(0 \in \partial f(X) + \mathrm{D} \Phi(X)^\top Y\) for \(Y \in \R^m\).

In the following, we consider the general optimization problem
\begin{equation}\tag{\(P_g\)}\label{eq:opt-prob-general}
\begin{cases}
    \inf_{\mathcal{X}, \mathcal{Z}} &\phi(\mathcal{X}, \mathcal{Z}) \\
    &A(\mathcal{X}, Z_0) + Q(\mathcal{Z}_>) = 0
\end{cases}
\end{equation}
where \(\mathcal{X} = (X_0, \dots, X_n)\), \(X_i \in \mathcal{M}_i\), \(\mathcal{Z} = (Z_0, Z_>)\), \(\mathcal{Z}_> = (Z_1, Z_2)\),
\begin{equation*}
\begin{aligned}
\phi(\mathcal{X}, \mathcal{Z}) &= f(\mathcal{X}) + \psi(\mathcal{Z}) \\
\text{and}\quad A(\mathcal{X}, Z_0) + Q(\mathcal{Z}_>) &=
\begin{pmatrix}
A_1(\mathcal{X}, Z_0) + Q_1(Z_1) \\
A_2(\mathcal{X}) + Q_2(Z_2)
\end{pmatrix}
\end{aligned}
\end{equation*}
where \(A_1\) and \(A_2\) are multiaffine maps and \(Q_1\) and \(Q_2\) are linear maps. The sets \(\mathcal{M}_i\) are either equal to \(\R^n\) for some \(n\) (this includes unconstrained matrix spaces \(\R^{n \times m}\) for some \(n, m\) as they are trivially identifiable with \(\R^{nm}\)) or a smooth embedded manifold.
The augmented Lagrangian for \eqref{eq:opt-prob-general} with penalty parameter \(\rho > 0\) is given by
\begin{equation}\label{eq:theoretical-lagrangian}
\mathcal{L}(\mathcal{X}, \mathcal{Z}, \mathcal{W}) =
\phi(\mathcal{X}, \mathcal{Z}) + \langle \mathcal{W}, A(\mathcal{X}, Z_0) + Q(\mathcal{Z}_>)\rangle + \frac{\rho}{2} \norm{A(\mathcal{X}, Z_0) + Q(\mathcal{Z}_>)}^2
\end{equation}
with \(\mathcal{W} = (W_1, W_2)\).
A set of assumptions is required to proof global convergence of the ADMM algorithm to a local optimum of \eqref{eq:opt-prob-general}. It will be shown further below that the optimization problem in the main article fulfills the assumptions below.
\begin{assumption}[{\citep[Assumption 4.1 and 4.2]{Gao2020}}]%
\label{ass:ass-4-1-and-4-2}
The following assumptions hold.\\
\textbf{A1.1:} For sufficiently large \(\rho\), every ADMM subproblem can be solved to optimality.\\
\textbf{A1.2:} \(\Image Q \supseteq \Image A\)\\
\textbf{A1.3:} The following statements regarding the objective function \(\phi\) and \(Q_2\) hold:
\begin{enumerate}
\item \(\phi\) is coercive on the feasible region \(\{(\mathcal{X}, \mathcal{Z}) \,\vert\, A(\mathcal{X}, Z_0) + Q(\mathcal{Z}_>) = 0\}\).
\item \(\psi(\mathcal{Z}) = h(Z_0) + g_1(Z_S) + g_2(Z_2)\) where
\begin{enumerate}
\item \(h\) is proper, convex, and lsc.
\item \(Z_S\) represents either \(Z_1\) or \((Z_0, Z_1)\) and \(g_1\) is \((m_1, M_1)\)-strongly convex.
\item \(g_2\) is \(M_2\)-Lipschitz differentiable.
\end{enumerate}
\item \(Q_2\) is injective.
\end{enumerate}
\textbf{A1.4:} The function \(f(\mathcal{X})\) splits into
\begin{equation*}
f(\mathcal{X}) = F(X_0, \dots, X_n) + \sum_{i = 0}^n f_i(X_i)
\end{equation*}
where \(F\) is \(M_F\)-Lipschitz differentiable, the functions \(f_0, f_1, \dots, f_n\) are proper, lsc, and each \(f_i\) is continuous on \(\Domain f_i\).\\
\textbf{A1.5:} For each \(0 \leq l \leq n\) one of the following assumptions holds:
\begin{enumerate}
\item \emph{Either}, the subproblem involving \(X_l\) is solved with a proximal term,
\item \emph{or} the following hold:
\begin{enumerate}
\item Viewing \(A(\mathcal{X}, Z_0) + Q(\mathcal{Z}_>) = 0\) as a system of constraints, there exists an index \(r(l)\) such that in the \(r(l)\)-th constraint,
\begin{equation*}
A_{r(l)}(\mathcal{X}, Z_0) = R_l(X_l) + \tilde{A}_l(\mathcal{X}_{\neq l}, Z_0)
\end{equation*}
for an injective linear map \(R_l\) and a multiaffine map \(\tilde{A}_l\).
\item \(f_l\) is either convex or \(M_l\)-Lipschitz differentiable.
\end{enumerate}
\end{enumerate}
\textbf{A1.6:} \(Z_0 \in Z_S\), so \(g_1(Z_S)\) is a strongly convex function of \(Z_0\) and \(Z_1\).
\end{assumption}

Let \(Y\) be any variable in \(\mathcal{X}\) or \(\mathcal{Z}\). Theorem~\ref{thm:general-stationarity} shows that the solution of each subproblem without manifold constraints fulfills the first-order condition \(0 \in \partial_Y \mathcal{L}\).
As argued above, as long as the Lagrangian is smooth or convex (or a sum of both convex and smooth parts) then there exist Lagrange multipliers \(M \in \R^m\) such that the first-order conditions \(\Phi(Y) = 0\) and \(0 \in \partial_Y \mathcal{L} + \mathrm{D}\Phi(Y)^\top M\) are fulfilled.

Since \(\Phi(X)\) is a smooth function, the continuity of the gradients guarantees their linear independence even in a neighborhood of \(Y \in \mathcal{M}\).
Together with Assumption~A1.2 and Lemma 5.4 in \citep{Gao2020}, this ensures that the \emph{constant rank constraint qualification} \citep{Janin1984} for the optimization problem \eqref{eq:opt-prob-general} is satisfied even when manifold constraints are added. This ensures that minimizers of \eqref{eq:opt-prob-general} are constrained stationary points the Lagrangian in Eq.~\eqref{eq:theoretical-lagrangian}. This means that the minimizer fulfills all manifold constraints, the multi-affine constraints and \(0 \in \partial \mathcal{L} + \sum_i \mathrm{D} \Phi_i(X_i)^\top M_i\) where \(i\) ranges over all manifold constraints and \(M_i\) are Lagrangian multipliers.

In the following, the theory in Sections 6 and 7 in \citet{Gao2020} is discussed and results are augmented to allow manifold-constrained subproblems during sequential updates.

A more general optimization problem than \eqref{eq:opt-prob-general} is considered first. Let
\begin{equation*}
\begin{cases}
\inf_{U_0, \dots, U_n} &f(U_0, \dots, U_n)\\
&C(U_0, \dots, U_n) = 0
\end{cases}
\end{equation*}
where  \(U_i \in \mathcal{M}_i\) and \(\mathcal{M}_i\) is an embedded submanifold of \(\R^n\). \(\mathcal{M}_i = \R^n\) is allowed. The augmented Lagrangian is then given by
\begin{equation*}
\mathcal{L}(U_0, \dots, U_n, W) = f(U_0, \dots, U_n) + \langle W, C(U_0, \dots, U_n) \rangle + \frac{\rho}{2} \norm{C(U_0, \dots, U_n)}^2.
\end{equation*}
The ADMM algorithm in Algorithm~1 in the main manuscript is applied to this problem. Note that manifold constrained subproblems are solved with an additional proximal term.
A set of assumptions about this generalized optimization problem are made below.
\begin{assumption}[{\citep[Assumption 6.1]{Gao2020}}]%
\label{ass:ass-6-1}
The following assumptions hold.\\
\noindent \textbf{A2.1}: For sufficiently large \(\rho\), every (proximal\footnote{If variable is manifold-constrained}) ADMM subproblem attains its optimal value.\\
\noindent \textbf{A2.2}: \(C(U_0, \dots, U_n)\) is smooth.\\
\noindent \textbf{A2.3}: \(f(U_0, \dots, U_n)\) is proper and lower semicontinuous.\\
\end{assumption}
Throughout \citep[Section 6]{Gao2020} it is assumed that Assumption~\ref{ass:ass-6-1} holds and we will assume the same below.

\begin{assumption}[{\citep[Assumption 6.2]{Gao2020}}]%
\label{ass:ass-6-2}
The function \(f\) has the form \(f(U_0, \dots, U_n) = F(U_0, \dots, U_n) + \sum_{i = 0}^n g_i(U_i)\), where \(F\) is smooth and \(g_i\) is continuous on \(\Domain g_i\).
\end{assumption}

In the following, denote the variable currently in focus by \(Y\), the remaining variables by \(U\), and the Lagrange multipliers by \(W\). In addition, we annotate variables already updated in one round of ADMM with a \(+\) and those still having the value from the previous round with a \(-\). Denote with \(U_<\) the variables updated before variable \(Y\) and with \(U_>\) those that will be updated after \(Y\). Write \(f_U(Y) := f(U_<, Y, U_>)\) and \(C(U, Y) := C(U_<, Y, C_>)\).
A manifold-constrained subproblem for a variable \(Y \in \mathcal{M}\) solves the optimization problem
\begin{equation*}
Y^+ = \argmin_{Y \in \mathcal{M}} \mathcal{L}(U, Y, W) + \frac{\alpha}{2}\norm{Y - Y^-}^2.
\end{equation*}
Associated with this subproblem is the Lagrangian
\begin{equation*}
\mathcal{L}_Y^P(U, Y, W, Y^-) = \mathcal{L}(U, Y, W) + \frac{\alpha}{2}\norm{Y - Y^-}^2.
\end{equation*}
Lemma~6.1 in \citep{Gao2020} continues to be valid unchanged.
\begin{lemma}[Modified {\citep[Lemma 6.2]{Gao2020}}]%
\label{lmm:lmm-6-2-mod}
The general subgradient of \(\mathcal{L}(U, Y, W)\) with respect to \(Y\) is given by
\begin{equation*}
\partial_Y \mathcal{L}(U, Y, W) = \partial f_U(Y)+(\nabla_Y C(U, Y))^\top W +\rho (\nabla_Y C(U, Y))^\top C(U, Y)
\end{equation*}
where \(\nabla_Y C(U, Y)\) is the Jacobian of \(Y \mapsto C(U, Y)\) and \((\nabla_Y C(U, Y))^\top\) is its adjoint. Defining \(V(U, Y, W) = (\nabla_Y C(U, Y))^\top W + \rho(\nabla_Y C(U, Y))^\top C(U, Y)\), the function \(V(U, Y, W)\) is continuous, and \(\partial_Y \mathcal{L}(U, Y, W) = \partial f_U(Y) + V(U, Y, W)\). The first-order condition satisfied by \(Y^+\) is therefore \(\Phi_Y(Y^+) = 0\) and
\begin{equation*}
\begin{aligned}
 0 &\in \partial f_{U_<^+,U_>}(Y^+) + (\nabla_Y C(U_<^+, Y^+, U_>))^\top W \\
 &+ \rho (\nabla_Y C(U_<^+, Y^+, U_>))^\top C(U_<^+, Y^+, U_>) \\
 &+ \alpha (Y^+ - Y^-) + \mathrm{D} \Phi_Y(Y^+)^\top M^+\\
&= \partial f_{U_<^+,U_>}(Y^+) + V(U_<^+, Y^+, U_>, W) \\
&+ \alpha (Y^+ - Y^-) + \mathrm{D} \Phi_Y(Y^+)^\top M^+
\end{aligned}
\end{equation*}
for some \(M^+\).
\end{lemma}
\begin{proof}
This follows immediately from Lemma~6.2 in \citep{Gao2020} and the definition of a stationary point on embedded submanifolds applied to \(\mathcal{L}^P_Y(U, Y, W, Y^-)\).
\end{proof}

\begin{lemma}[Modified {\citep[Lemma 6.3]{Gao2020}}]%
\label{lmm:lmm-6-3-mod}
Suppose that Assumptions~\ref{ass:ass-6-1} and \ref{ass:ass-6-2} hold. It then follows that
\begin{equation*}
\begin{aligned}
&V(U_<^{k + 1}, Y^{k + 1}, U_>^{k + 1}, W^{k + 1})
- V(U_<^{k + 1}, Y^{k + 1}, U_>^k, W^k) \\
+\, &\nabla_Y F(U_<^{k + 1}, Y^{k + 1}, U_>^{k + 1})
- \nabla_Y F(U_<^{k + 1}, Y^{k + 1}, U_>^k) \\
-\, &\alpha (Y^{k + 1} - Y^k)
\end{aligned}
\end{equation*}
is contained in  \(\partial_Y \mathcal{L}_Y(U^{k + 1}, Y^{k + 1}, W^{k + 1}) + \mathrm{D}\Phi(Y^{k + 1})^\top M^{k + 1}\) for some \(M^{k + 1}\).

Consider any limit point \((U^*, Y^*, W^*)\) of the ADMM algorithm. If \(\norm{W^+ - W} \rightarrow 0\) and \(\norm{U_>^+ - U_>} \rightarrow 0\), then for any subsequence \(\{(U^{k(s)}, Y^{k(s)}, W^{k(s)})\}_{s = 0}^\infty\) converging to \((U^*, Y^*, W^*)\), there exists a sequence \(\nu^s \in \partial_Y \mathcal{L}(U^{k(s)}, Y^{k(s)}, W^{k(s)}) + \mathrm{D}\Phi(Y^{k(s)})^\top M^{k(s)}\) with \(\nu^s \rightarrow 0\).
\end{lemma}

\begin{proof}
By Lemma~\ref{lmm:lmm-6-2-mod} there exists \(M^{k + 1}\) such that
\begin{equation*}
\begin{aligned}
-\Big(
    &\nabla_Y F(U_<^{k + 1}, Y^{k + 1}, U_>^k)
    + V(U_<^{k + 1}, Y^{k + 1}, U_>^k, W^k) \\
    &+ \alpha (Y^{k + 1} - Y^k)
\Big) \in \partial g_Y(Y^{k + 1}) + \mathrm{D}\Phi(Y^{k + 1})^\top M^{k + 1},
\end{aligned}
\end{equation*}
and also
\begin{equation*}
\begin{aligned}
\partial_Y \mathcal{L}_Y(U_<^{k + 1}, Y^{k + 1}, U_>^{k + 1}, W^{k + 1}) + \mathrm{D}\Phi(Y^{k + 1})^\top M^{k + 1} &=
\partial g_Y(Y^{k + 1}) + \mathrm{D}\Phi(Y^{k + 1})^\top M^{k + 1} \\
&+ \nabla_Y F(U_<^{k + 1}, Y^{k + 1}, U_>^{k + 1}) \\
&+ V(U_<^{k + 1}, Y^{k + 1}, U_>^{k+1}, W^{k+1}).
\end{aligned}
\end{equation*}
Combining these two results shows the statement about the subgradient.

Note that since \(Y^{k(s)} \rightarrow Y^*\) the sequence \(\{Y^{k(s)}\}\) is a Cauchy sequence and therefore \(\alpha (Y^{k + 1} - Y^k) \rightarrow 0\). Applying this and the result about the subgradient above to \(\partial_Y \mathcal{L}(U^{k(s)}, Y^{k(s)}, W^{k(s)}) + \mathrm{D}\Phi(Y^{k(s)})^\top M^{k(s)}\) shows that the second result above follows in the same way as in the proof by \citet[Lemma 6.3]{Gao2020}.
\end{proof}



\begin{lemma}[Modified {\citep[Lemma A.12]{Gao2020}}]%
\label{lmm:lmm-A-12-mod}
Let \(f: \mathcal{M} \rightarrow \overline{\R}\) be a proper lsc function on a smooth embedded manifold \(\mathcal{M}\). Suppose that we have sequences \(x_k \rightarrow x\) in \(\Domain f\) and \(\nu_k \in \partial f(x_k) + \mathrm{D}g(x_k)^\top \lambda_k\) for some \(\lambda_k\) such that \(f(x_k) \rightarrow f(x)\) and \(\nu_k \rightarrow \nu\). Then \(\nu \in \partial f(x) + \mathrm{D}g(x)^\top \lambda\) for some \(\lambda\).
\end{lemma}

\begin{proof}
Note that \(\R^n = N_\mathcal{M}(x_k) \oplus T_\mathcal{M}(x_k)\) where \(T_\mathcal{M}(x_k) = \{y \in \R^n \,\vert\, \mathrm{D}\Phi(x_k) y = 0\}\) since \(T_\mathcal{M}(x_k) = (N_\mathcal{M}(x_k))^\perp\) \citep[Example 6.8]{Rockafellar2009}. Without loss of generality we can therefore assume that \(\nu_k = \mu_k + \mathrm{D}g(x_k)^\top \lambda_k\) for \(\mu_k \in \partial f(x_k) \cap T_\mathcal{M}(x_k)\). Otherwise we could write \(\mu_k = \tilde{\mu}_k + \mathrm{D}g(x_k)^\top \tilde{\lambda}_k\) with \(\tilde{\mu}_k \in \partial f(x_k) \cap T_\mathcal{M}(x_k)\) and therefore \(\nu_k = \tilde{\mu}_k + \mathrm{D}g(x_k)^\top (\lambda_k + \tilde{\lambda})k)\).

Since \(\mathrm{D}\Phi(x_k) \mu_k = 0\) and using that \(\mathrm{D}\Phi(x_k)\) has full row-rank for all \(k\), it then holds that \(\lambda_k = (\mathrm{D}\Phi(x) \mathrm{D}\Phi(x)^\top)^{-1} \mathrm{D}\Phi(x_k) \nu_k\) and \(\lambda_k \rightarrow \lambda = (\mathrm{D}\Phi(x) \mathrm{D}\Phi(x)^\top)^{-1} \mathrm{D}\Phi(x) \nu\) since \(\nu_k \rightarrow \nu\) and the continuity of \(\mathrm{D}\Phi\). This implies that \(\mu_k = \nu_k - \mathrm{D}g(x_k)^\top \lambda_k \rightarrow \mu = \nu - \mathrm{D}g(x)^\top \lambda\). Lemma A.12 in \citep{Gao2020} then implies that \(\mu \in \partial f(x)\) which is equivalent to \(\nu \in \partial f(x) + \mathrm{D}\Phi(x)^\top \lambda\) for some \(\lambda\).
\end{proof}

\begin{lemma}[Modified {\citep[Lemma 6.4]{Gao2020}}]%
\label{lmm:lmm-6-4-mod}
Suppose that Assumptions~\ref{ass:ass-6-1} and \ref{ass:ass-6-2} hold. Let \((U^*, Y^*,W^*)\) be a feasible limit point. By passing to a subsequence converging to the limit point, let \(\{(U^s, Y^s, W^s)\}\) be a subsequence of the ADMM iterates with \((U^s, Y^s, W^s) \rightarrow (U^*, Y^*, W^*)\). Suppose that there exists a sequence \(\{\nu^s\}\) such that \(\nu^s \in \partial_Y \mathcal{L}(U^s, Y^s, W^s) + \mathrm{D}\Phi(Y^s)^\top M^s\) for some \(M^s\) all \(s\) and \(\nu^s \rightarrow 0\). Then \(0 \in \partial g_Y(Y^*) + \nabla_Y F(U^*, Y^*) + (\nabla_Y C(U^*, Y^*))^\top W^* + \mathrm{D}\Phi(Y^*)^\top M^*\) for some \(M^*\), so \((U^*, Y^*, W^*)\) is a constrained stationary point.
\end{lemma}
\begin{proof}
The proof is analogous to the proof of Lemma~6.4 in \citep{Gao2020} with the modification of using Lemma~\ref{lmm:lmm-A-12-mod} instead of Lemma~A.12 in \citep{Gao2020}. As in the proof of Lemma~\ref{lmm:lmm-6-3-mod}, the term in the subgradient introduced by the proximal term \(\alpha (Y^{s + 1} - Y^s)\) vanishes due to \(\{Y^s\}\) being a Cauchy sequence.
\end{proof}
Corollary 6.5 in \citep{Gao2020} holds unchanged.
\begin{corollary}[Modified {\citep[Corollary 6.7]{Gao2020}}]%
\label{corr:corr-6-7-mod}
Assume now that \(C(U, Y)\) is multiaffine. Taking \(\nabla_Y C(U, Y) = C_U\) in Lemma~\ref{lmm:lmm-6-2-mod} and writing \(C(U, Y) = C_U(Y) - b_U\), the general subgradient of \(Y \mapsto \mathcal{L}(U, Y, W)\) is given by
\begin{equation*}
\partial_Y \mathcal{L}(U, Y, W) =
\partial f_U(Y) + C_U^\top W + \rho C_U^\top (C_U(Y) - b_U).
\end{equation*}
Thus the first-order condition for \(Y \mapsto \mathcal{L}(U, Y, W)\) at \(Y^+\) is given by
\begin{equation*}
0 \in
\partial f_U(Y^+) + C_U^\top W + \rho C_U^\top (C_U(Y^+) - b_U)
+ \alpha(Y^+ - Y^-) + \mathrm{D}\Phi(Y^+)^\top M^+
\end{equation*}
for some \(M^+\).
\end{corollary}
\begin{proof}
Follows directly from Lemma~\ref{lmm:lmm-6-2-mod}.
\end{proof}

\begin{lemma}[Modified {\citep[Lemma 6.8]{Gao2020}}]%
\label{lmm:lmm-6-8-mod}
The change in the augmented Lagrangian when the primal variable \(Y\) is updated to \(Y^+\) is given by
\begin{equation*}
\mathcal{L}(U, Y, W) - \mathcal{L}(U, Y^+, W) = f_U(Y) - f_U(Y^+) - \langle \nu, Y - Y^+ \rangle + \frac{\rho}{2} \norm{C_U(Y) - C_U(Y^+)}^2
\end{equation*}
for some \(\nu \in \partial f_U(Y^+) + \alpha(Y^+ - Y^-) + \mathrm{D} \Phi(Y^+)^\top M^+\) and \(M^+\).
\end{lemma}

\begin{proof}
The proof is analogous to the proof of Lemma~6.8 in \citep{Gao2020}. The only difference is that Corollary~\ref{corr:corr-6-7-mod} implies that the subgradient \(\nu \in \partial f_U(Y^+) + \alpha(Y^+ - Y^-) + \mathrm{D} \Phi(Y^+)^\top M^+\) for some \(M^+\).
\end{proof}

From Lemma~\ref{lmm:lmm-6-8-mod} it follows that if \(f_U(Y)\) is convex or \(L\)-Lipschitz differentiable, then
\begin{equation*}
\begin{aligned}
f_U(Y) - f_U(Y^+) - \langle \nu, Y - Y^+ \rangle &=
f_U(Y) - f_U(Y^+) \\
&- \langle \nu - \mathrm{D} \Phi(Y^+)^\top M^+ - \alpha (Y - Y^+), Y - Y^+ \rangle \\
&- \langle \alpha (Y - Y^+) + \mathrm{D} \Phi(Y^+)^\top M^+, Y - Y^+ \rangle \\
&\geq - \frac{\mu + \alpha}{2} \norm{Y - Y^+}^2 - \langle \mathrm{D} \Phi(Y^+)^\top M^+, Y - Y^+ \rangle \\
\end{aligned}
\end{equation*}
where \(\mu = L\) if \(f_U(Y)\) is L-Lipschitz differentiable or \(\mu = 0\) if \(f_U(Y)\) is convex.
Despite \(\mathrm{D}\Phi(Y^+)^\top M^+\) being in the normal cone to \(\mathcal{M}\) at \(Y^+\), unless \(\mathcal{M}\) is convex (which it typically is not if it is an embedded manifold), it will not generally hold that \(\langle \mathrm{D} \Phi(Y^+)^\top M^+, Y - Y^+ \rangle \leq 0\). Therefore, despite the similarity of results, Lemma~\ref{lmm:lmm-6-8-mod} cannot fulfill the same role as \citep[Lemma 6.8]{Gao2020} in ensuring sufficient descent of the ADMM algorithm when updating manifold constrained variables. However, this is where the importance of the proximal update comes in. Using the proximal update ensures that
\begin{equation*}
\mathcal{L}(U, Y^+, W) + \frac{\alpha}{2}\norm{Y^+ - Y^-}^2 = \mathcal{L}_Y^P(U, Y^+, W, Y^-) \leq \mathcal{L}_Y^P(U, Y^-, W, Y^-) = \mathcal{L}(U, Y^-, W),
\end{equation*}
or put differently \(\mathcal{L}(U, Y^-, W) - \mathcal{L}(U, Y^+, W) \geq \frac{\alpha}{2}\norm{Y^+ - Y^-}^2\).

To finish up the results in \citep[Section 6]{Gao2020}, Lemmas 6.11 and 6.12 as well as Corollary 6.14 in \citep{Gao2020} address only the variables in the simultaneous update and therefore continue to hold unchanged.
\begin{lemma}[Modified {\citep[Corollary 6.15 and Lemma 6.16]{Gao2020}}]%
\label{lmm:corr-6-15-lmm-6-16}
Let \(C_<(Y) = b_<\) denote the constraint \(C(U_<^+, Y, U_>) = b_<\) as a linear function of \(Y\), after updating the variables \(U_<\), and let \(C_>(Y) = b_>\) denote the constraint \(C(U_<^+, Y, U_>^+) = b_>\). Then the general subgradient \(\partial_Y \mathcal{L}(U_<^+, Y^+, U_>^+, W^+) + \mathrm{D}\Phi(Y^+)^\top M^+\) for some \(M^+\) contains
\begin{equation*}
\begin{aligned}
&(C_>^\top - C_<^\top)W^+ + C_<^\top(W^+ - W) \\
+\, &\rho (C_>^\top - C_<^\top)(C_>(Y^+) - b_>) \\
+\, &\rho C_<^\top (C_>(Y^+) - b_> - (C_<(Y^+) - b_<)) \\
+\, &\nabla_Y F(U_<^+, Y^+, U_>^+) - \nabla_Y F(U_<^+, Y^+, U_>) \\
-\, &\alpha (Y^+ - Y^-).
\end{aligned}
\end{equation*}
In particular, if \(Y\) is the final block, then \(C_<^\top (W^+ - W) \in \partial_Y \mathcal{L}(U_<^+, Y^+, W^+) + \mathrm{D}\Phi(Y^+)^\top M^+\). 
Assume that the following holds
\begin{enumerate}
\item \(\norm{W^{k - 1} - W^k} \rightarrow 0\),
\item \(\norm{C_>^k - C_<^k} \rightarrow 0\),
\item \(\norm{b_> - b_<} \rightarrow 0\),
\item \(\norm{Y^{k - 1} - Y^k} \rightarrow 0\),
\item \(\{W^k\}, \{Y^k\}, \{C_<^k\}, \{C_>^k(Y^k) - b_>\}\) are bounded,
\item and \(\norm{U_>^k - U_>^{k - 1}} \rightarrow 0\).
\end{enumerate}
Then there exists a sequence \(\nu_k \in \partial_Y \mathcal{L}^k\) with \(\nu^k \rightarrow 0\). In particular, if \(Y\) is the final block, then only Condition 1 and the boundedness of \(\{C_<^k\}\) are needed.
\end{lemma}
\begin{proof}
The proof is analogous to the proofs of Corollary 6.15 and Lemma 6.16 in \citep{Gao2020}.
\end{proof}

Note that the usage of the proximal update required the addition of Condition 4 in Lemma~\ref{lmm:corr-6-15-lmm-6-16} above. However, this is not an issue since the only two places where this result is used, the proofs of \citep[Theorem 4.1 and 4.3]{Gao2020}, are either concerned only with final blocks, meaning that no proximal updates are present, or ensure that Condition 4 holds.

Lemma 6.17 holds unchanged. This concludes changes to the theory in \citep[Section 6]{Gao2020}.

\citep[Section 7]{Gao2020} uses Assumption~\ref{ass:ass-4-1-and-4-2} and the theory developed in \citep[Section 6]{Gao2020} to proof convergence of the ADMM iterates and the augmented Lagrangian objective function. If manifold constraints are present on subproblems, the modified results above can be used as drop-in replacements all throughout \citep[Section 7]{Gao2020}. This is primarily necessary in \citep[Lemma 7.9]{Gao2020} where the convergence of the variables involved in sequential updates is shown. The addition of a manifold constraint adds an additional inner product as shown after Lemma~\ref{lmm:lmm-6-8-mod} which can be bounded from below as shown there. The second adjustment that needs to be performed is in \citep[Proof of Theorem 4.5]{Gao2020}. By \citep[p. 310]{Rockafellar2009} it holds for the indicator function \(\iota_C\) of a set \(C \subseteq \R^n\) that \(\partial \iota_C(\bar{\vect{x}}) = \partial^\infty \iota_C(\bar{\vect{x}}) = N_C(\bar{\vect{x}})\). Instead of adding the manifold constraints as set constraints for the variables, it is therefore possible to add them as indicator functions to the augmented Lagrangian in Eq.~\eqref{eq:theoretical-lagrangian}. Using the result about subgradients of indicator functions, the extended theory above holds unchanged.
The same arguments about subgradients and their bounds by differences of iterates as in \citep[Proof of Theorem 4.5]{Gao2020} can then be made for the augmented Lagrangian with added indicator functions.

\begin{theorem}[{\citep[Theorem~4.1, Theorem~4.3, Remark~4.2, Remark~4.4]{Gao2020}}]%
\label{thm:convergence:existence-limit-points}
Assume Assumption~\ref{ass:ass-4-1-and-4-2} holds. Write \(\sigma_1\) for
the smallest positive eigenvalue of \(Q_1^\top Q_1\), \(\sigma_2\) for
the smallest eigenvalue (positive or negative) of \(Q_2^\top Q_2\),
\(r_l\)  for the smallest eigenvalue of \(R_l^\top R_l\) and set
\(\kappa = M_1 / m_1\). Then suppose that \(\rho\) fulfills
\begin{align}
\frac{\sigma_2 \rho}{2} - \frac{M_2^2}{\sigma_2 \rho} &> \frac{M_2}{2}
\label{eq:convergence:rho-ineq-1}
\text{ and} \\
\rho &> \max\left(\frac{2 M_1 \kappa}{\sigma_1},\,
\frac{M_1 + M_2}{2}
\max\left(\sigma_2^{-1}, \frac{(1 + 2\kappa)^2}{\sigma_1}\right)\right).
\label{eq:convergence:rho-ineq-2}
\end{align}
Then the sequence of iterates \((\bm{\Theta}^t, \bm{\Delta}^t, \bm{\Lambda}^t)\)
produced by Algorithm~1 in the main manuscript is bounded, which implies the
existence of limit points. It holds that every limit point
\((\bm{\Theta}^*, \bm{\Delta}^*, \bm{\Lambda}^*)\) satisfies
\(A(\bm{\Theta}^*) + Q(\bm{\Delta}^*) = \bm{0}\)
and is a stationary point of \eqref{eq:m-admm:admm-opt-problem}.
\end{theorem}


If, in addition, it is assumed that the augmented Lagrangian is a K-\L\ function it is possible to prove the convergence
of the sequence of iterates produced by Algorithm~1 in the main manuscript.

\begin{theorem}[Modified {\citep[Theorem 4.5]{Gao2020}}]%
\label{thm:convergence:uniqueness-limit-point-v1}
Assume the augmented Lagrangian in
Eq.\@~\eqref{eq:theoretical-lagrangian} is a K-\L\ function and
suppose that all the assumptions and conditions on \(\rho\) are fulfilled as
in Theorem~\ref{thm:convergence:existence-limit-points}. The sequence of iterates 
\((\bm{\Theta}^t, \bm{\Delta}^t, \bm{\Lambda}^t)\) produced by 
Algorithm~1 in the main manuscript converges to a single limit point 
\((\bm{\Theta}^*, \bm{\Delta}^*, \bm{\Lambda}^*)\) which is a constrained stationary point for \eqref{eq:m-admm:admm-opt-problem}.
\end{theorem}


\section{Applying the convergence theory to (P)}%
\label{sec:check-assump}

The optimization problem in the main manuscript is
\begin{equation}\tag{P}\label{eq:m-admm:admm-opt-problem}
\left\{\begin{array}{rl}
\displaystyle\min_{\mat{\Theta}, \mat{\Delta}}
&\displaystyle\frac{1}{2} \sum_{(i, j, \gamma) \in \mathcal{I}}
\norm*{\mathcal{P}_{\mathrm{\Omega}_{ij, \gamma}}\left(\mat{X}_{ij,\gamma} - \mat{Z}_{ij,\gamma}\right)}_2^2 \\
&\displaystyle+ \lambda_1 \sum_{(i, j, \gamma) \in \mathcal{I}} \sum_{l = 1}^k
\vect{w}_{ij,\gamma}^{(l)} \abs{\mat{D}_{ij,\gamma}^{(l, l)}}
+ \lambda_2 \sum_{l = 1}^m  \vect{w}_1^{(l)} \norm{\mat{U}_l}_1
+ \frac{\mu}{2} \sum_{l = 1}^m \vect{w}_2^{(l)} \norm{\mat{V}'_l}_F^2 \\
\text{subject to }& \mat{V}_l \in \Stiefel_{k, p_l}, \mat{U}_l \in \Oblique_{k, p_l}, \mat{D}_{ij,\gamma} \in \mathcal{D}_k, \text{ and} \\
A(\mat{\Theta}) + Q(\mat{\Delta}) &=
\begin{pmatrix}
A_1(\mat{\Theta}) & + & Q_1(\mat{\Delta}_1) \\
A_2(\mat{\Theta}) & + & Q_2(\mat{\Delta}_2) \\
\end{pmatrix}
=
\begin{pmatrix}
\mat{U}_l - \mat{V}_l & + & (-\mat{V}'_l) \\
-\mat{V}_i \mat{D}_{ij,\gamma} \mat{V}_j^\top & + & \mat{Z}_{ij,\gamma} \\
\end{pmatrix}
= \mat{0}
\end{array}
\right.
\end{equation}
where it is understood that each constraint is included for all \(l \in [m]\) and \((i, j, \gamma) \in \mathcal{I}\). The following augmented Lagrangian is associated with \eqref{eq:m-admm:admm-opt-problem}
\begin{equation}\label{eq:m-admm:admm-augmented-lagrangian}
\begin{aligned}
\mathcal{L}_{\rho}(\mat{\Theta}, \mat{\Delta}, \mat{\Lambda}) &=
\frac{1}{2} \sum_{(i, j, \gamma) \in \mathcal{I}}
\norm*{\mathcal{P}_{\mathrm{\Omega}_{ij, \gamma}}\left(\mat{X}_{ij,\gamma} - \mat{Z}_{ij,\gamma}\right)}_2^2 \\
&+ \lambda_1 \sum_{(i, j, \gamma) \in \mathcal{I}} \left[\iota_{\mathcal{D}_k}(\mat{D}_{ij,\gamma}) +
\sum_{l = 1}^k \vect{w}_{ij,\gamma}^{(l)} \abs{\mat{D}_{ij,\gamma}^{(l, l)}}\right]
+ \lambda_2 \sum_{l = 1}^m \vect{w}_1^{(l)} \norm{\mat{U}_l}_1
+ \frac{\mu}{2} \sum_{l = 1}^m 
\vect{w}_2^{(l)} \norm{\mat{V}'_l}_F^2 \\
&+ \frac{\rho}{2} \sum_{l = 1}^m \left[
\norm{\mat{U}_l - \mat{V}_l - \mat{V}'_l + \mat{\Lambda}^{(1)}_{l}}_F^2
- \norm{\mat{\Lambda}^{(1)}_{l}}_F^2\right] \\
&+ \frac{\rho}{2} \sum_{(i, j, \gamma) \in \mathcal{I}} \left[
\norm{\mat{Z}_{ij,\gamma} - \mat{V}_i \mat{D}_{ij,\gamma} \mat{V}_j^\top
+ \mat{\Lambda}^{(2)}_{ij,\gamma}}_F^2
- \norm{\mat{\Lambda}^{(2)}_{ij,\gamma}}_F^2\right].
\end{aligned}
\end{equation}
In the following, all assumptions in Assumption~\ref{ass:ass-4-1-and-4-2} will be checked to ensure that that the extended convergence theory presented above holds.

It is easy to see that the functions \(A_1\) and \(A_2\) are multiaffine and \(Q_1\) and \(Q_2\) are linear.

The sets \(\Stiefel_{k, p}\), \(\Oblique_{k, p}\), and \(\mathcal{D}_k\) are smooth embedded manifolds. However, \(\mathcal{D}_k\) is included in the augmented Lagrangian as an indicator function. This is possible since \(\mathcal{D}_k\) is a convex set and therefore \(\iota_{\mathcal{D}_k}\) is a convex function.

The Stiefel manifold is the set of all \(p \times k\) matrices with orthogonal columns, formally defined as \(\Stiefel_{k, p} = \{\mat{V} \in \R^{p \times k} \,\vert\, \mat{V}^\top \mat{V} = \mat{I}\}\). The set is described by the zeros of the function \(\R^{p \times k} \rightarrow \R^{k \times k}, \mat{V} \mapsto \Phi(\mat{V}) = \mat{V}^\top \mat{V} - \mat{I}\). \(\Phi(\mat{V})\) is smooth with \(k^2 \times pk\) Jacobian \(\mathrm{D}\Phi(\mat{V}) = (\mat{K}_k + \mat{I}_{k^2}) (\mat{I}_k \otimes \mat{V}^\top)\) where \(\mat{K}_k\) is the unique orthogonal permutation matrix such that \(\mat{K}_k \Vectorize(\mat{\xi}) = \Vectorize(\mat{\xi}^\top)\). The matrix \(\mat{K}_k + \mat{I}_{k^2}\) is invertible and if \(\mat{V} \in \Stiefel_{k, p}\) then the rows of \(\mat{V}^\top\) are orthogonal, hence linearly independent. By definition of the Kronecker product, it follows that \(\mat{I}_k \otimes \mat{V}^\top\) has \(k^2\) linearly independent rows. This shows that \(\mathrm{D}\Phi(\mat{V})\) has full row-rank \(k^2\) on \(\Stiefel_{k, p}\).

The oblique manifold is the set of all \(p \times k\) matrices with column norm equal to 1, formally defined as \(\Oblique_{k, p} = \{\mat{V} \in \R^{p \times k} \,\vert\, (\mat{V}^\top\mat{V})^{(i, i)} - 1 = 0 \text{ for } i = 1, \dots, k\}\). The set is described by the zeros of the function \(\R^{p \times k} \rightarrow \R^k, \mat{V} \mapsto \sum_{i = 1}^k (\norm{\mat{V} \vect{e}_i}_2^2 - 1) \vect{e}_i\), where \(\vect{e}_i\) is the \(i\)-th standard unit vector. \(\Phi(\mat{V})\) has the \(k \times pk\) Jacobian \(\mathrm{D}\Phi(\mat{V}) = 2\sum_{i = 1}^k \left[(\vect{e}_i^\top \otimes \vect{e}_i) \otimes \vect{e}_i^\top \mat{V}^\top\right]\). If \(\mat{V} \in \Oblique_{k, p}\) then \(\norm{\mat{V}^{(:, i)}}_2 = 1\) which in particular implies that each column is different from the zero vector. This ensures that \(\mathrm{D}\Phi(\mat{V})\) has full row-rank \(k\) on \(\Oblique_{k, p}\).\\


\emph{\textbf{A1.1:} For sufficiently large \(\rho\), every ADMM subproblem can be solved to optimality.}

As detailed in
Sections~\ref{ssec:m-admm:subproblems-theta}~and~%
\ref{ssec:m-admm:subproblems-delta} every subproblem is solvable analytically
for any \(\rho > 0\) and therefore each attains its optimal value.\\


\emph{\textbf{A1.2:} \(\Image Q \supseteq \Image A\)}

The linear mapping \(Q\) in \eqref{eq:m-admm:admm-opt-problem} is a concatenation of identity matrices or negative identity matrices. Therefore, \(\Image(\Vectorize(Q)) = \R^N\) where \(N = k \sum_l p_l + \sum_{(i, j, \gamma) \in \mathcal{I}} p_i p_j\). This always contains \(\Image(\Vectorize(A)) \subset \R^N\).\\


\emph{\textbf{A1.3.1:} \(\phi\) is coercive on the feasible region \(\{(\mathcal{X}, \mathcal{Z}) \,\vert\, A(\mathcal{X}, Z_0) + Q(\mathcal{Z}_>) = 0\}\).}

Note that the Stiefel and Oblique manifolds are both compact sets and therefore in particular bounded. Coercivity in these variables is therefore not required.

The functions \(\mat{D}_{ij,\gamma} \mapsto \iota_{\mathrm{D}_k}(\mat{D}_{ij,\gamma}) + \sum_l \vect{w}_{ij,\gamma}^{(l)} \abs{\mat{D}_{ij,\gamma}^{(l, l)}}\) are coercive and since \(\mat{Z}_{ij,\gamma} = \mat{V}_i \mat{D}_{ij,\gamma} \mat{V}_j^\top\) on the feasible region it follows that the objective is coercive in \(\mat{Z}_{ij,\gamma}\) as well.

On the feasible region it holds that \(\mat{V}'_l = \mat{U}_l - \mat{V}_l\) and therefore \(\mat{V}'_l\) is bounded in that region as well. In addition, the functions \(\mat{V}'_l \mapsto \frac{\mu}{2} \vect{w}_2^{(l)} \norm{\mat{V}'_l}_F^2\) are coercive.

The objective function \(\phi\) is therefore coercive on the feasible region.\\


\emph{\textbf{A1.3.2:} \(\psi(\mathcal{Z}) = h(Z_0) + g_1(Z_S) + g_2(Z_2)\) where (i) \(h\) is proper, convex, and lsc., (ii) \(Z_S\) represents either \(Z_1\) or \((Z_0, Z_1)\) and \(g_1\) is \((m_1, M_1)\)-strongly convex, and (iii) \(g_2\) is \(M_2\)-Lipschitz differentiable.}

Note that careful analysis of the proofs in \citep{Gao2020}, in particular \citep[Corollary 7.3]{Gao2020}, reveals that \(h\) is not necessary for establishing convergence and that only one of \(g_1\) or \(g_2\) needs to be present. However, all blocks in the linear parts of each constraint, i.e.\@ those in \(Q\) in \eqref{eq:m-admm:admm-opt-problem}, need to be part of either \(g_1\) or \(g_2\).

For \eqref{eq:m-admm:admm-opt-problem} it holds that \(Z_0 = \varnothing\) and therefore \(h\) does not exist.
\(Z_1\) is equal to \(\mat{V}' = (\mat{V}'_1, \dots, \mat{V}'_m)\) and \(g_1(Z_1) = \frac{\mu}{2} \sum_{l = 1}^m \vect{w}_2^{(l)} \norm{\mat{V}'_l}_F^2\).
To show that \(g_1\) is \((m_1, M_1)\)-strongly convex, it has to be shown that \(g_1\) is convex, \(M_1\)-Lipschitz differentiable, and fulfills Eq.\@~\eqref{eq:defs-res:strongly-convex} with constant \(m_1\).
\(g_1\) is smooth, its gradient is 
\begin{equation}
\nabla_{\Vectorize(\mat{V}')} g_1(\mat{V}') = \mu
\left[
\vect{w}_2^{(1)} \Vectorize(\mat{V}'_1)^\top,
\dots,
\vect{w}_2^{(m)} \Vectorize(\mat{V}'_m)^\top
\right]^\top,
\end{equation}
and the Hessian is
\begin{equation}
\nabla^2_{\Vectorize(\mat{V}')\Vectorize(\mat{V}')} g_1(\mat{V}') = \mu
\begin{pmatrix}
\vect{w}_2^{(1)} \mat{I}_{p_1k} & & \\
& \ddots & \\
& & \vect{w}_2^{(m)} \mat{I}_{p_mk} \\
\end{pmatrix}
\end{equation}
for all \(\mu > 0\) and \(\vect{w}_2^{(i)} > 0\). The Hessian is positive definite with maximal singular value \(M_1 := \mu \max_l \vect{w}_2^{(l)}\), which shows that \(g_1\) is convex and \(M_1\)-Lipschitz differentiable.
In addition, it holds for all \(\mat{V}'\) and \(\tilde{\mat{V}}'\) that
\begin{equation}
g_1(\tilde{\mat{V}}') - g_1(\mat{V}') -
\left\langle
\nabla_{\Vectorize(\mat{V}')} g_1(\mat{V}'),
\tilde{\mat{V}}' - \mat{V}
\right\rangle_F
\geq \frac{\mu \min_l \vect{w}_2^{(l)}}{2}
\norm{\tilde{\mat{V}}' - \mat{V}}_F^2.
\end{equation}
Setting \(m_1 := \mu \min_l \vect{w}_2^{(l)}\) shows that \(g_1\) is 
\((m_1, M_1)\)-strongly convex.

\(Z_2\) is equal to \(\mat{Z} = (\mat{Z}_{ij,\gamma})_{(i,j,\gamma) \in \mathcal{I}}\) and
\begin{equation*}
g_2(Z_2) = \frac{1}{2} \sum_{(i, j, \gamma) \in \mathcal{I}}
\norm*{\mathcal{P}_{\mathrm{\Omega}_{ij, \gamma}}\left(\mat{X}_{ij,\gamma} - \mat{Z}_{ij,\gamma}\right)}_2^2.
\end{equation*}
Its partial derivatives are
\begin{equation*}
\frac{\partial}{\partial \mat{Z}_{ij,\gamma}^{(r, c)}} g_2(\mat{Z}) =
\begin{cases}
\mat{Z}_{ij,\gamma}^{(r, c)} - \mat{X}_{ij,\gamma}^{(r, c)} & \text{if } \mat{X}_{ij,\gamma}^{(r, c)} \text{ is observed} \\
0 & \text{otherwise}
\end{cases}
\end{equation*}
which clearly defines a gradient which is 1-Lipschitz continuous. The function \(g_2\) is therefore 1-Lipschitz differentiable.\\


\emph{\textbf{A1.3.3:} \(Q_2\) is injective.}

Since \(Q_2(\mat{\Delta}_2) = \mat{\Delta}_2\) the function is clearly injective.\\


\emph{\textbf{A1.4:} The function \(f(\mathcal{X})\) splits into \(f(\mathcal{X}) = F(X_0, \dots, X_n) + \sum_{i = 0}^n f_i(X_i)\) where \(F\) is \(M_F\)-Lipschitz differentiable, the functions \(f_0, f_1, \dots, f_n\) are proper, lsc, and each \(f_i\) is continuous on \(\Domain f_i\).}

The function \(F\) is not used in our formulation of \eqref{eq:m-admm:admm-opt-problem} and the theory in \citep{Gao2020} allows for omission of this term.
Each variable block which is not part of \(Z_1\) or \(Z_2\) is associated with a function \(f_l\) and they are described below.
\begin{itemize}
\item Stiefel manifold constrained variables \(\mat{V}_l\) do not appear in the objective function and we can therefore conceptually add \(\mat{V}_l \mapsto 0\) which trivially fulfills the assumptions.
\item The variables \(\mat{D}_{ij,\gamma}\) appear in with \(\mat{D}_{ij,\gamma} \mapsto \iota_{\mathrm{D}_k}(\mat{D}_{ij,\gamma}) + \sum_{l = 1}^k \vect{w}_{ij,\gamma}^{(l)} \abs{\mat{D}_{ij,\gamma}^{(l, l)}}\). Since the set of diagonal matrices \(\mathrm{D}_k\) is a convex set it holds that this function is a continuous convex function and therefore fulfills the assumptions.
\item Oblique-manifold constrained variables \(\mat{U}_l\) appear in \(\mat{U}_l \mapsto \lambda_2 \vect{w}_1^{(l)} \norm{\mat{U}_l}_1\) which is a continuous convex function and therefore fulfills the assumptions.
\end{itemize}


\emph{\textbf{A1.5:} \emph{Either} the subproblem involving \(X_l\) is solved with a proximal term \emph{or} the following hold: For each \(0 \leq l \leq n\) it holds that: (i) Viewing \(A(\mathcal{X}, Z_0) + Q(\mathcal{Z}_>) = 0\) as a system of constraints, there exists an index \(r(l)\) such that in the \(r(l)\)-th constraint, \(A_{r(l)}(\mathcal{X}, Z_0) = R_l(X_l) + \tilde{A}_l(\mathcal{X}_{\neq l}, Z_0)\) for an injective linear map \(R_l\) and a multiaffine map \(\tilde{A}_l\).
(ii) \(f_l\) is either convex or \(M_l\)-Lipschitz differentiable.}

It is only necessary to show this property for \(\mat{D}_{ij,\gamma}\) since the subproblems for \(\mat{V}_l\) and \(\mat{U}_l\) are solved with a proximal term.

Technically, this assumption is not fulfilled for \(\mat{D}_{ij,\gamma}\) for \eqref{eq:m-admm:admm-opt-problem}. However, due to how the assumption is used in \citep[Lemma~7.9]{Gao2020} (and only there) and the fact that \(\mat{V}_i, \mat{V}_j\) have orthogonal columns, it can be shown that it is not necessary in case of \eqref{eq:m-admm:admm-opt-problem}. When \citep[Lemma~7.9]{Gao2020} is applied to \eqref{eq:m-admm:admm-opt-problem} the term \(\norm{\mat{V}_i \mat{D}_{ij,\gamma} \mat{V}_j^\top - \mat{V}_i \mat{D}_{ij,\gamma}^+ \mat{V}_j^\top}_F^2\) appears. The purpose of the assumption above is to be able to bound this norm by \(\norm{R (\mat{D}_{ij,\gamma} - \mat{D}_{ij,\gamma}^+)}_F^2\), where \(R\) is a linear map. However, here this can be shown directly making the assumption unnecessary. It holds
\begin{equation*}
\begin{aligned}
\norm{\mat{V}_i \mat{D}_{ij,\gamma} \mat{V}_j^\top - \mat{V}_i \mat{D}_{ij,\gamma}^+ \mat{V}_j^\top}_F^2 &= \norm{\mat{V}_i (\mat{D}_{ij,\gamma} - \mat{D}_{ij,\gamma}^+) \mat{V}_j^\top}_F^2 \\
&= \Trace(\mat{V}_j (\mat{D}_{ij,\gamma} - \mat{D}_{ij,\gamma}^+) \mat{V}_i^\top \mat{V}_i (\mat{D}_{ij,\gamma} - \mat{D}_{ij,\gamma}^+) \mat{V}_j^\top) \\
&= \Trace((\mat{D}_{ij,\gamma} - \mat{D}_{ij,\gamma}^+) \mat{V}_i^\top \mat{V}_i (\mat{D}_{ij,\gamma} - \mat{D}_{ij,\gamma}^+) \mat{V}_j^\top \mat{V}_j) \\
&= \Trace((\mat{D}_{ij,\gamma} - \mat{D}_{ij,\gamma}^+)^2) = \norm{\mat{D}_{ij,\gamma} - \mat{D}_{ij,\gamma}^+}_F^2.
\end{aligned}
\end{equation*}


\emph{\textbf{A1.6:} \(Z_0 \in Z_S\), so \(g_1(Z_S)\) is a strongly convex function of \(Z_0\) and \(Z_1\).}

For \eqref{eq:m-admm:admm-opt-problem} it holds that \(Z_0 = \varnothing\), as noted in the justification of \textbf{A1.3.2} above, and therefore this assumption is fulfilled trivially.


\begin{remark}
For the optimization problem \eqref{eq:m-admm:admm-opt-problem} the inequality in
Eq.\@~\eqref{eq:convergence:rho-ineq-1} for \(\rho\) turns out to be
fulfilled if
%
\begin{equation}
\frac{\rho}{2} - \frac{1}{\rho} > \frac{1}{2}
\Leftrightarrow \rho > 2.
\end{equation}
In addition, Eq.\@~\eqref{eq:convergence:rho-ineq-2} requires that
\begin{equation}
\rho > \max\left(
2\frac{\mu \left(\max_l \bm{w}_2^{(l)}\right)^2}{\min_l \bm{w}_2^{(l)}},
\frac{1 + \mu \max_l \bm{w}_2^{(l)}}{2}
\left(1 + 2 \frac{\max_l \bm{w}_2^{(l)}}{\min_l \bm{w}_2^{(l)}}\right)^2
\right).
\end{equation}
\end{remark}


Lastly, it needs to be shown that the augmented Lagrangian in Eq.\@~\eqref{eq:m-admm:admm-augmented-lagrangian} is a Kurdyka--{\L}ojasiewicz function. We will do so by showing that the augmented Lagrangian with manifold set constraints added as indicator functions is a semi-algebraic function (cfr.\@ Definition~\ref{def:semi-algebraic-function}).

Clearly, indicator functions of subsets that can be described by polynomials are semi-algebraic functions. This includes the Stiefel and Oblique manifolds as well as the set of \(k \times k\) diagonal matrices. Elementary functions such as absolute values, squares, multiplication by fixed scalars, as well as products of variables can easily seen to be semi-algebraic by Definition~\ref{def:semi-algebraic-set}.

As an example, it holds for the graph of the absolute value that
\begin{equation}
\{(x, |x|) : x \in \R\} = 
\{(t, t) : t > 0\} \cup \{(t, -t) : t < 0\} \cup \{(0, 0)\}.
\end{equation}
These sets can be described as follows
\begin{align}
\{(t, t) : t > 0\} &= 
\{\vect{x} \in \R^2 :
\vect{x}^{(1)} - \vect{x}^{(2)} = 0, -\vect{x}^{(1)} < 0\} \\
\{(t, -t) : t < 0\} &= 
\{\vect{x} \in \R^2 :
\vect{x}^{(1)} + \vect{x}^{(2)} = 0, \vect{x}^{(1)} < 0\} \\
\{(0, 0)\} &= 
\{\vect{x} \in \R^2 : \vect{x}^{(1)} = 0\} \cap 
\{\vect{x} \in \R^2 : \vect{x}^{(2)} = 0\}
\end{align}
which shows that the absolute value is a semi-algebraic function. 

To show that sums and compositions of semi-algebraic functions are semi-algebraic, the projection stability of semi-algebraic sets has to be used. For compositions, let \(f: \R^l \rightarrow \R^m\) and \(g: \R^n \rightarrow \R^l\) be two semi-algebraic functions, then
\begin{equation}
G_{f \circ g} =
\{(\vect{x}, \vect{y}) \in \R^{n + m}: \vect{y} = f(g(\vect{x}))\}
\end{equation}
can be written as the projection on the first \(n + m\) components of
\begin{equation}
\begin{aligned}
&\{(\vect{x}, \vect{y}, \vect{z}) \in \R^{n + m + l}: 
\vect{z} = g(\vect{x}), \vect{y} = f(\vect{z})\} \\
=\ &\{(\vect{x}, \vect{y}, \vect{z}) \in \R^{n + m + l}: 
\vect{z} = g(\vect{x})\} \cap
\{(\vect{x}, \vect{y}, \vect{z}) \in \R^{n + m + l}: 
\vect{y} = f(\vect{z})\}
\end{aligned}
\end{equation}
where the latter two sets are trivial embeddings of the graphs of \(f\) and \(g\) into higher-dimensional spaces. Compositions of semi-algebraic functions are therefore semi-algebraic.

Similarly, for sums of semi-algebraic functions, let \(f: \R^n \rightarrow \R^m\) and \(g: \R^n \rightarrow \R^m\) be two semi-algebraic functions, then
\begin{equation}
G_{f + g} =
\{(\vect{x}, \vect{y}) \in \R^{n + m}: \vect{y} = f(\vect{x}) + g(\vect{x})\}
\end{equation}
can be written as the projection on the first \(n + m\) components of
\begin{equation}
\begin{aligned}
&\{(\vect{x}, \vect{y}, \vect{w}, \vect{z}) \in \R^{n + 3m}: 
\vect{y} = \vect{w} + \vect{z},
\vect{w} = g(\vect{x}),
\vect{z} = f(\vect{x})\} \\
=\ &\{(\vect{x}, \vect{y}, \vect{w}, \vect{z}) \in \R^{n + 3m}:
\vect{y} - \vect{w} - \vect{z} = 0\} \\
\cap\ &\{(\vect{x}, \vect{y}, \vect{w}, \vect{z}) \in \R^{n + 3m}:
\vect{w} = g(\vect{x})\} \\
\cap\ &\{(\vect{x}, \vect{y}, \vect{w}, \vect{z}) \in \R^{n + 3m}:
\vect{z} = f(\vect{x})\}
\end{aligned}
\end{equation}
where the last two sets are trivial embeddings of the graphs of \(f\) and \(g\) into higher-dimensional spaces, and the first set is a polynomially constrained subset of \(\R^{n + 3m}\). Sums of semi-algebraic functions are therefore semi-algebraic. Subtractions can now be seen to be semi-algebraic as a composition of multiplication by -1, a fixed scalar, of one function and addition.

In summary, all elementary functions involved in the augmented Lagrangian in Eq.\@~\eqref{eq:m-admm:admm-augmented-lagrangian} are semi-algebraic and so is their composition and sums of terms. Therefore, the augmented Lagrangian is a semi-algebraic function and a K-\L\ function.

\newpage
\section{Additional results}

\begin{figure}[ht]
\begin{subfigure}[t]{1.0\textwidth}
\centering
\includegraphics{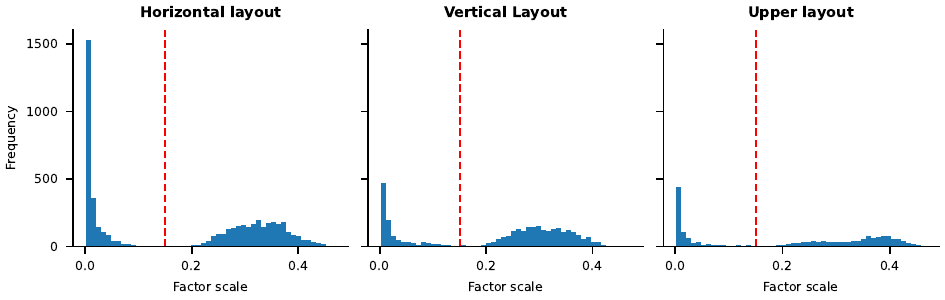}
\caption{MOFA}
\end{subfigure}
\begin{subfigure}[t]{1.0\textwidth}
\centering
\includegraphics{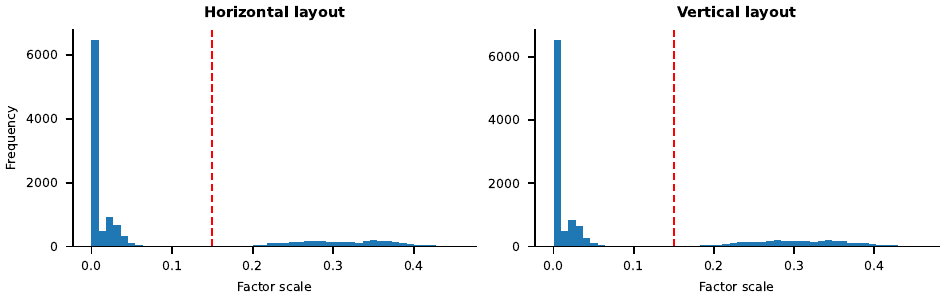}
\caption{CMF}
\end{subfigure}
\caption{Evaluation of thresholds for MOFA and CMF to determine which factors are active within a signal matrix. Factor scales across all factors, matrices, and simulation runs were plotted as a histogram. Red dashed lines show the location of the chosen threshold. For both methods and in all layout scenarios it was chosen to be equal to 0.15.}
\label{fig:mofa-cmf-thresholds}
\end{figure}

\begin{table}[h]
\centering
\begin{tabular}{clcccccc}
\toprule
Matrix & Method & \multicolumn{6}{c}{Estimated Rank} \\
\cmidrule{3-8}
& & Mean (Std. Dev.) & \multicolumn{5}{c}{Quantiles} \\
\cmidrule{4-8}
& & & Min & 25\% & 50\% & 75\% & Max \\
\midrule
A & Groundtruth & 3 & \multicolumn{5}{c}{---} \\
 & solrCMF & 2.84 (0.38) & 2 & 3 & 3 & 3 & 4 \\
 & CMF & 3.01 (0.41) & 2 & 3 & 3 & 3 & 4 \\
 & MMPCA & 2.94 (0.98) & 0 & 2 & 3 & 4 & 6 \\
 & SLIDE & 2.91 (0.52) & 2 & 3 & 3 & 3 & 4 \\
 & MOFA & 4.48 (0.79) & 3 & 4 & 5 & 5 & 6 \\
\\
B & Groundtruth & 3 & \multicolumn{5}{c}{---} \\
 & solrCMF & 2.99 (0.09) & 2 & 3 & 3 & 3 & 3 \\
 & CMF & 3.00 (0.06) & 2 & 3 & 3 & 3 & 3 \\
 & MMPCA & 3.57 (1.23) & 1 & 3 & 3 & 4 & 7 \\
 & SLIDE & 2.88 (0.34) & 1 & 3 & 3 & 3 & 3 \\
 & MOFA & 4.68 (0.69) & 3 & 4 & 5 & 5 & 6 \\
\\
A and B & Groundtruth & 2 & \multicolumn{5}{c}{---} \\
 & solrCMF & 1.83 (0.38) & 1 & 2 & 2 & 2 & 3 \\
 & CMF & 1.98 (0.39) & 1 & 2 & 2 & 2 & 3 \\
 & MMPCA & 2.37 (1.09) & 0 & 2 & 2 & 3 & 5 \\
 & SLIDE & 1.68 (0.50) & 0 & 1 & 2 & 2 & 2 \\
 & MOFA & 4.09 (0.90) & 2 & 4 & 4 & 5 & 6 \\
\bottomrule
\end{tabular}
\caption{Estimated rank in Simulation 1 restricted to matrices A, B, and factors found to be shared between A and B. Results are aggregated over all 250 simulation runs. For all methods but solrCMF the horizontal layout was considered.}
\label{tab:estimated-rank-A-B}
\end{table}

\begin{table}[h]
\centering
\begin{tabular}{clccc}
\toprule
Matrix & Method & \multicolumn{3}{c}{Factor scale \emph{Mean (Std. Dev., Occurrences)}} \\
\cmidrule{3-5}
& & Factor 1 & Factor 2 & Factor 3 \\
\midrule
A & Groundtruth & 0.39 (0.01, 250) & 0.33 (0.01, 250) & 0.28 (0.01, 250) \\
 & solrCMF & 0.39 (0.03, 250) & 0.33 (0.04, 250) & 0.24 (0.06, 208) \\
 & CMF & 0.34 (0.02, 250) & 0.29 (0.02, 250) & 0.25 (0.03, 249) \\
 & MMPCA & 0.35 (0.03, 238) & 0.29 (0.04, 226) & 0.24 (0.05, 171) \\
 & SLIDE & 0.41 (0.02, 250) & 0.37 (0.02, 250) & 0.34 (0.02, 203) \\
 & MOFA & 0.36 (0.02, 250) & 0.31 (0.02, 250) & 0.26 (0.03, 250) \\
\\
B & Groundtruth & 0.37 (0.01, 250) & 0.33 (0.01, 250) & 0.30 (0.01, 250) \\
 & solrCMF & 0.37 (0.02, 250) & 0.34 (0.02, 250) & 0.29 (0.03, 248) \\
 & CMF & 0.34 (0.02, 250) & 0.30 (0.02, 250) & 0.27 (0.02, 250) \\
 & MMPCA & 0.37 (0.02, 241) & 0.30 (0.05, 232) & 0.24 (0.08, 200) \\
 & SLIDE & 0.39 (0.01, 250) & 0.36 (0.01, 249) & 0.33 (0.02, 221) \\
 & MOFA & 0.35 (0.02, 250) & 0.32 (0.02, 250) & 0.28 (0.02, 250) \\
\bottomrule
\end{tabular}
\caption{Factor scale of the top 3 factors estimated in Simulation 1. For all methods but solrCMF the horizontal layout was considered. Factor 1 to 3 are with regard to factor scale (largest to smallest) and do not correspond to the order given in the simulation description. Mean values are taken over runs where the factor was present, therefore the number of occurrences is specified.}
\label{tab:top3-factor-scale-A-B}
\end{table}


\begin{figure}
\centering
\includegraphics{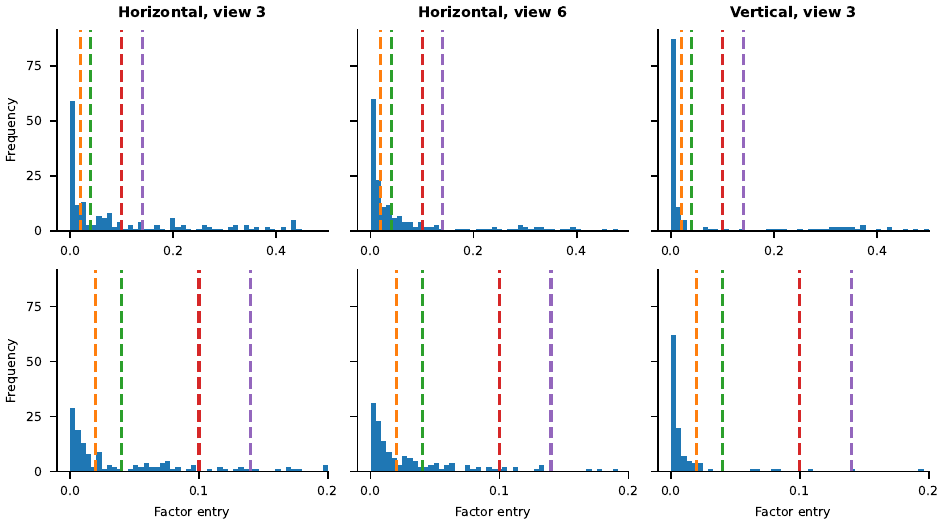}
\caption{Examples of factor entries estimated by MOFA from one representative estimation run during Simulation 1. The top and bottom row show the same information, however, the bottom row is zoomed in on the first half of the horizontal axis to exemplify small thresholds more clearly. Possible thresholds that are considered in the main manuscript are shown as dashed lines. Colors are added to clarify which thresholds in the upper and lower row are the same.}
\label{fig:factor-threshold-mofa}
\end{figure}

\clearpage
\phantomsection
\addcontentsline{toc}{section}{References}
\bibliographystyle{jasa3}
\bibliography{refs}

%% file: directed-r2-figure-250runs.tex
\begin{subfigure}{6cm}
\centering
\begin{adjustbox}{max width=5.5cm}
\begin{tikzpicture}
\node[
    rectangle,
    draw,
    minimum height=2.5cm,
    minimum width=1.25cm,
    label={[shift={(-3.5ex, -3.5ex)}]north east:A},
    label=left:1,
    label=below:2]
    (A) {};
\node[
    rectangle,
    draw,
    minimum height=2.5cm,
    minimum width=1.75cm,
    label={[shift={(-3ex, -3.5ex)}]north east:C}]
    (C) at ($(A) + (2.25cm, 0.25cm)$) {};
\node[
    rectangle,
    draw,
    minimum height=2.5cm,
    minimum width=1.75cm,
    label={[shift={(-3.5ex, -3.5ex)}]north east:B},
    label=below:3,
    fill=white]
    (B) at ($(A) + (1.75cm, 0cm)$) {};
\node[
    rectangle,
    draw,
    minimum height=1.5cm,
    minimum width=1.75cm,
    label={[shift={(-3.5ex, -3.5ex)}]north east:D},
    label=left:4]
    (D) at ($(B) + (0cm, 2.5cm)$) {};
\node[
    rectangle,
    draw,
    minimum height=1.5cm,
    minimum width=2.5cm,
    label={[shift={(-3.5ex, -3.5ex)}]north east:E},
    label=right:4,
    label=below:1]
    (E) at ($(D) + (-3.85cm, 0cm)$) {};
\end{tikzpicture}
\end{adjustbox}
\caption{Full layout}
\label{fig:graph-res:full-layout}
\end{subfigure}%
\begin{subfigure}{5cm}
\centering
\begin{adjustbox}{max width=4.5cm}
\begin{tikzpicture}[>=Latex]
\node[circle,draw] (A) {A};
\node[circle,draw,below right=of A] (B) {B};
\node[circle,draw,above right=of A] (C) {C};
\node[circle,draw,below right=of C] (D) {D};
\node[circle,draw,right=of D] (E) {E};

\draw[-{>[length=1.926mm]}, line width=0.385mm] (B) to[out=125, in=-35] (A);
\draw[-{>[length=0.821mm]}, line width=0.164mm] (C) to[out=235, in=35] (A);
\draw[-{>[length=1.439mm]}, line width=0.288mm] (E.270) to[bend left=30] ($(A)!0.5!(E) - (0cm, 2.25cm)$) to[bend left=45] (A.280);
\draw[-{>[length=2.029mm]}, line width=0.406mm] (A) to[out=-55, in=145] (B);
\draw[-{>[length=2.234mm]}, line width=0.447mm] (C) to[out=280, in=80] (B);
\draw[-{>[length=2.237mm]}, line width=0.447mm] (D) to[out=215, in=55] (B);
\draw[-{>[length=1.325mm]}, line width=0.265mm] (E.245) to[bend left=12.5] (B.-25);
\draw[-{>[length=0.765mm]}, line width=0.153mm] (A) to[out=55, in=215] (C);
\draw[-{>[length=2.244mm]}, line width=0.449mm] (B) to[out=100, in=260] (C);
\draw[-{>[length=3.332mm]}, line width=0.666mm] (D) to[out=125, in=-35] (C);
\draw[-{>[length=1.488mm]}, line width=0.298mm] (E.120) to[bend right=12.5] (C.25);
\draw[-{>[length=2.368mm]}, line width=0.474mm] (B) to[out=35, in=235] (D);
\draw[-{>[length=3.348mm]}, line width=0.670mm] (C) to[out=-55, in=145] (D);
\draw[-{>[length=1.646mm]}, line width=0.329mm] (E) to[out=170, in=10] (D);
\draw[-{>[length=1.884mm]}, line width=0.377mm] (A.90) to[bend left=45] ($(A)!0.5!(E) + (0cm, 2.25cm)$) to[bend left=30] (E.100);
\draw[-{>[length=1.462mm]}, line width=0.292mm] (B.10) to[bend right=10] (E.215);
\draw[-{>[length=1.458mm]}, line width=0.292mm] (C.-10) to[bend left=10] (E.155);
\draw[-{>[length=1.471mm]}, line width=0.294mm] (D) to[out=-10, in=190] (E);

\end{tikzpicture}
\end{adjustbox}
\caption{Ground-truth}
\label{fig:graph-res:truth}
\end{subfigure}%
\begin{subfigure}{5cm}
\centering
\begin{adjustbox}{max width=4.5cm}
\begin{tikzpicture}[>=Latex]
\node[circle,draw] (A) {A};
\node[circle,draw,below right=of A] (B) {B};
\node[circle,draw,above right=of A] (C) {C};
\node[circle,draw,below right=of C] (D) {D};
\node[circle,draw,right=of D] (E) {E};

\draw[-{>[length=1.604mm]}, line width=0.321mm] (B) to[out=125, in=-35] (A);
\draw[-{>[length=0.545mm]}, line width=0.109mm] (C) to[out=235, in=35] (A);
\draw[-{>[length=1.499mm]}, line width=0.300mm] (E.270) to[bend left=30] ($(A)!0.5!(E) - (0cm, 2.25cm)$) to[bend left=45] (A.280);
\draw[-{>[length=1.891mm]}, line width=0.378mm] (A) to[out=-55, in=145] (B);
\draw[-{>[length=2.293mm]}, line width=0.459mm] (C) to[out=280, in=80] (B);
\draw[-{>[length=2.139mm]}, line width=0.428mm] (D) to[out=215, in=55] (B);
\draw[-{>[length=1.369mm]}, line width=0.274mm] (E.245) to[bend left=12.5] (B.-25);
\draw[-{>[length=0.678mm]}, line width=0.136mm] (A) to[out=55, in=215] (C);
\draw[-{>[length=2.339mm]}, line width=0.468mm] (B) to[out=100, in=260] (C);
\draw[-{>[length=3.278mm]}, line width=0.656mm] (D) to[out=125, in=-35] (C);
\draw[-{>[length=1.561mm]}, line width=0.312mm] (E.120) to[bend right=12.5] (C.25);
\draw[-{>[length=2.247mm]}, line width=0.449mm] (B) to[out=35, in=235] (D);
\draw[-{>[length=3.213mm]}, line width=0.643mm] (C) to[out=-55, in=145] (D);
\draw[-{>[length=1.792mm]}, line width=0.358mm] (E) to[out=170, in=10] (D);
\draw[-{>[length=2.005mm]}, line width=0.401mm] (A.90) to[bend left=45] ($(A)!0.5!(E) + (0cm, 2.25cm)$) to[bend left=30] (E.100);
\draw[-{>[length=1.532mm]}, line width=0.306mm] (B.10) to[bend right=10] (E.215);
\draw[-{>[length=1.532mm]}, line width=0.306mm] (C.-10) to[bend left=10] (E.155);
\draw[-{>[length=1.532mm]}, line width=0.306mm] (D) to[out=-10, in=190] (E);

\end{tikzpicture}
\end{adjustbox}
\caption{solrCMF}
\label{fig:graph-res:solrcmf}
\end{subfigure}

\begin{subfigure}{6cm}
\centering
\begin{adjustbox}{max width=5.5cm}
\begin{tikzpicture}
\node[
    rectangle,
    draw,
    minimum height=2.5cm,
    minimum width=1.25cm,
    label={[shift={(-3.5ex, -3.5ex)}]north east:A},
    label=left:1,
    label={[shift={(0ex, -0.5ex)}]below:2}]
    (A) {};
\node[
    rectangle,
    draw,
    minimum height=2.5cm,
    minimum width=1.75cm,
    label={[shift={(-3.5ex, -3.5ex)}]north east:B}]
    (B) at ($(A) + (1.75cm, 0cm)$) {};
\node[
    rectangle,
    draw,
    minimum height=2.5cm,
    minimum width=1.75cm,
    label={[shift={(-3ex, -3.5ex)}]north east:C},
    label={[shift={(0ex, -0.5ex)}]below:6}]
    (C1) at ($(A) + (3.75cm, 0cm)$) {};
\node[
    rectangle,
    draw,
    minimum height=2.5cm,
    minimum width=1.75cm,
    label={[shift={(-3ex, -3.5ex)}]north east:C},
    label=left:5]
    (C2) at ($(B) + (0cm, -2.75cm)$) {};
\node[
    rectangle,
    draw,
    minimum height=1.5cm,
    minimum width=1.75cm,
    label={[shift={(-3.5ex, -3.5ex)}]north east:D},
    label=left:4,
    label={[shift={(0ex, 0.85ex)}]above:3}]
    (D) at ($(B) + (0cm, 2.5cm)$) {};
\node[
    rectangle,
    draw,
    minimum height=1.5cm,
    minimum width=2.5cm,
    label={[shift={(-3.5ex, -3.5ex)}]north east:E},
    label=right:4,
    label={[shift={(0ex, 0.85ex)}]above:1}]
    (E) at ($(D) + (-3.85cm, 0cm)$) {};

\node[fit=(E)(A)(B)(C1), draw, dashed] {};
\coordinate (Eext) at ($(E.north west) + (-0.5ex, 0.5ex)$);
\node[fit=(Eext)(A)(B)(C2), draw, cblue, dashed] {};
\end{tikzpicture}
\end{adjustbox}
\caption{Unfolded augmented layouts}
\label{fig:graph-res:unrolled-augmented-layout}
\end{subfigure}%
\begin{subfigure}{5cm}
\centering
\begin{adjustbox}{max width=4.5cm}
\begin{tikzpicture}[>=Latex]
\node[circle,draw] (A) {A};
\node[circle,draw,below right=of A] (B) {B};
\node[circle,draw,above right=of A] (C) {C};
\node[circle,draw,below right=of C] (D) {D};
\node[circle,draw,right=of D] (E) {E};

\draw[-{>[length=1.515mm]}, line width=0.303mm] (B) to[out=125, in=-35] (A);
\draw[-{>[length=0.660mm]}, line width=0.132mm] (C) to[out=235, in=35] (A);
\draw[-{>[length=1.205mm]}, line width=0.241mm] (E.270) to[bend left=30] ($(A)!0.5!(E) - (0cm, 2.25cm)$) to[bend left=45] (A.280);
\draw[-{>[length=1.652mm]}, line width=0.330mm] (A) to[out=-55, in=145] (B);
\draw[-{>[length=1.932mm]}, line width=0.386mm] (C) to[out=280, in=80] (B);
\draw[-{>[length=1.903mm]}, line width=0.381mm] (D) to[out=215, in=55] (B);
\draw[-{>[length=1.191mm]}, line width=0.238mm] (E.245) to[bend left=12.5] (B.-25);
\draw[-{>[length=0.610mm]}, line width=0.122mm] (A) to[out=55, in=215] (C);
\draw[-{>[length=1.872mm]}, line width=0.374mm] (B) to[out=100, in=260] (C);
\draw[-{>[length=1.305mm]}, line width=0.261mm] (E.120) to[bend right=12.5] (C.25);
\draw[-{>[length=1.897mm]}, line width=0.379mm] (B) to[out=35, in=235] (D);
\draw[-{>[length=1.468mm]}, line width=0.294mm] (E) to[out=170, in=10] (D);
\draw[-{>[length=1.638mm]}, line width=0.328mm] (A.90) to[bend left=45] ($(A)!0.5!(E) + (0cm, 2.25cm)$) to[bend left=30] (E.100);
\draw[-{>[length=1.345mm]}, line width=0.269mm] (B.-5) to[bend right=10] (E.230);
\draw[-{>[length=1.339mm]}, line width=0.268mm] (C.-10) to[bend left=10] (E.155);
\draw[-{>[length=1.362mm]}, line width=0.272mm] (D) to[out=-10, in=190] (E);

\draw[-{>[length=1.440mm]}, line width=0.288mm, cblue] (B) to[out=115, in=-25] (A);
\draw[-{>[length=1.204mm]}, line width=0.241mm, cblue] (E.290) to[bend left=30] ($(A)!0.5!(E) - (0cm, 2.4cm)$) to[bend left=45] (A.260);
\draw[-{>[length=1.593mm]}, line width=0.319mm, cblue] (A) to[out=-65, in=155] (B);
\draw[-{>[length=1.910mm]}, line width=0.382mm, cblue] (C) to[out=290, in=70] (B);
\draw[-{>[length=1.889mm]}, line width=0.378mm, cblue] (D) to[out=205, in=65] (B);
\draw[-{>[length=1.194mm]}, line width=0.239mm, cblue] (E.260) to[bend left=12.5] (B.-45);
\draw[-{>[length=1.829mm]}, line width=0.366mm, cblue] (B) to[out=110, in=250] (C);
\draw[-{>[length=2.651mm]}, line width=0.530mm, cblue] (D) to[out=125, in=-35] (C);
\draw[-{>[length=2.027mm]}, line width=0.405mm, cblue] (B) to[out=25, in=245] (D);
\draw[-{>[length=2.765mm]}, line width=0.553mm, cblue] (C) to[out=-55, in=145] (D);
\draw[-{>[length=1.550mm]}, line width=0.310mm, cblue] (E) to[out=160, in=20] (D);
\draw[-{>[length=1.659mm]}, line width=0.332mm, cblue] (A.110) to[bend left=45] ($(A)!0.5!(E) + (0cm, 2.4cm)$) to[bend left=30] (E.80);
\draw[-{>[length=1.304mm]}, line width=0.261mm, cblue] (B.10) to[bend right=10] (E.215);
\draw[-{>[length=1.322mm]}, line width=0.264mm, cblue] (D) to[out=-20, in=200] (E);

\end{tikzpicture}
\end{adjustbox}
\caption{CMF (Threshold 0.15)}
\label{fig:graph-res:cmf}
\end{subfigure}%
\begin{subfigure}{5cm}
\centering
\begin{adjustbox}{max width=4.5cm}
\begin{tikzpicture}[>=Latex]
\node[circle,draw] (A) {A};
\node[circle,draw,below right=of A] (B) {B};
\node[circle,draw,above right=of A] (C) {C};
\node[circle,draw,below right=of C] (D) {D};
\node[circle,draw,right=of D] (E) {E};

\draw[-{>[length=1.852mm]}, line width=0.370mm] (B) to[out=125, in=-35] (A);
\draw[-{>[length=1.495mm]}, line width=0.299mm] (C) to[out=235, in=35] (A);
\draw[-{>[length=2.044mm]}, line width=0.409mm] (E.270) to[bend left=30] ($(A)!0.5!(E) - (0cm, 2.25cm)$) to[bend left=45] (A.280);
\draw[-{>[length=1.682mm]}, line width=0.336mm] (A) to[out=-55, in=145] (B);
\draw[-{>[length=2.516mm]}, line width=0.503mm] (C) to[out=280, in=80] (B);
\draw[-{>[length=2.418mm]}, line width=0.484mm] (D) to[out=215, in=55] (B);
\draw[-{>[length=2.418mm]}, line width=0.484mm] (E.245) to[bend left=12.5] (B.-25);
\draw[-{>[length=0.857mm]}, line width=0.171mm] (A) to[out=55, in=215] (C);
\draw[-{>[length=1.740mm]}, line width=0.348mm] (B) to[out=100, in=260] (C);
\draw[-{>[length=1.655mm]}, line width=0.331mm] (E.120) to[bend right=12.5] (C.25);
\draw[-{>[length=1.682mm]}, line width=0.336mm] (B) to[out=35, in=235] (D);
\draw[-{>[length=1.682mm]}, line width=0.336mm] (E) to[out=170, in=10] (D);
\draw[-{>[length=1.971mm]}, line width=0.394mm] (A.90) to[bend left=45] ($(A)!0.5!(E) + (0cm, 2.25cm)$) to[bend left=30] (E.100);
\draw[-{>[length=2.423mm]}, line width=0.485mm] (B.-5) to[bend right=10] (E.230);
\draw[-{>[length=2.286mm]}, line width=0.457mm] (C.-10) to[bend left=10] (E.155);
\draw[-{>[length=2.423mm]}, line width=0.485mm] (D) to[out=-10, in=190] (E);

\draw[-{>[length=1.723mm]}, line width=0.345mm, cblue] (B) to[out=115, in=-25] (A);
\draw[-{>[length=1.970mm]}, line width=0.394mm, cblue] (E.290) to[bend left=30] ($(A)!0.5!(E) - (0cm, 2.4cm)$) to[bend left=45] (A.260);
\draw[-{>[length=1.840mm]}, line width=0.368mm, cblue] (A) to[out=-65, in=155] (B);
\draw[-{>[length=2.315mm]}, line width=0.463mm, cblue] (C) to[out=290, in=70] (B);
\draw[-{>[length=2.456mm]}, line width=0.491mm, cblue] (D) to[out=205, in=65] (B);
\draw[-{>[length=2.456mm]}, line width=0.491mm, cblue] (E.260) to[bend left=12.5] (B.-45);
\draw[-{>[length=2.057mm]}, line width=0.411mm, cblue] (B) to[out=110, in=250] (C);
\draw[-{>[length=2.313mm]}, line width=0.463mm, cblue] (D) to[out=125, in=-35] (C);
\draw[-{>[length=2.379mm]}, line width=0.476mm, cblue] (B) to[out=25, in=245] (D);
\draw[-{>[length=2.628mm]}, line width=0.526mm, cblue] (C) to[out=-55, in=145] (D);
\draw[-{>[length=2.379mm]}, line width=0.476mm, cblue] (E) to[out=160, in=20] (D);
\draw[-{>[length=2.037mm]}, line width=0.407mm, cblue] (A.110) to[bend left=45] ($(A)!0.5!(E) + (0cm, 2.4cm)$) to[bend left=30] (E.80);
\draw[-{>[length=2.076mm]}, line width=0.415mm, cblue] (B.10) to[bend right=10] (E.215);
\draw[-{>[length=2.076mm]}, line width=0.415mm, cblue] (D) to[out=-20, in=200] (E);

\end{tikzpicture}
\end{adjustbox}
\caption{MMPCA}
\label{fig:graph-res:mmpca}
\end{subfigure}

\begin{subfigure}{6cm}
\centering
\begin{adjustbox}{max width=5.5cm}
\begin{tikzpicture}
\node[
    rectangle,
    draw,
    minimum height=2.5cm,
    minimum width=1.25cm,
    label={[shift={(-3.5ex, -3.5ex)}]north east:A},
    label={[shift={(0ex, -0.5ex)}]below:2}]
    (A) {};
\node[
    rectangle,
    draw,
    minimum height=2.5cm,
    minimum width=1.75cm,
    label={[shift={(-3.5ex, -3.5ex)}]north east:B}]
    (B) at ($(A) + (1.75cm, 0cm)$) {};
\node[
    rectangle,
    draw,
    minimum height=2.5cm,
    minimum width=1.75cm,
    label={[shift={(-3ex, -3.5ex)}]north east:C},
    label={[shift={(0ex, -0.5ex)}]below:6}]
    (C1) at ($(A) + (3.75cm, 0cm)$) {};
\node[
    rectangle,
    draw,
    minimum height=2.5cm,
    minimum width=1.75cm,
    label={[shift={(-3ex, -3.5ex)}]north east:C},
    label={[shift={(-0.5ex, 0ex)}]left:5}]
    (C2) at ($(B) + (0cm, -2.75cm)$) {};
\node[
    rectangle,
    draw,
    minimum height=1.5cm,
    minimum width=1.75cm,
    label={[shift={(-3.5ex, -3.5ex)}]north east:D},
    label={[name=Dl3,shift={(0ex, 0.5ex)}]above:3}]
    (D) at ($(B) + (0cm, 2.5cm)$) {};
\node[
    rectangle,
    draw,
    minimum height=1.5cm,
    minimum width=2.5cm,
    label={[shift={(-3.5ex, -3.5ex)}]north east:E},
    label={[shift={(-0.5ex, 0)}]left:4},
    label={[name=E1l1,shift={(0ex, 0.5ex)}]above:1}]
    (E1) at ($(D) + (-2.375cm, 0cm)$) {};
\node[
    rectangle,
    draw,
    minimum height=2.5cm,
    minimum width=1.5cm,
    label={[shift={(-3.5ex, -3.5ex)}]north east:E},
    label={[name=E2l1,shift={(-0.5ex, 0ex)}]left:1},
    label={[name=E2l4, shift={(0ex, -0.5ex)}]below:4}]
    (E2) at ($(A) + (-1.65cm, 0cm)$) {};

\node[fit=(E1)(D), draw, cred, dashed] {};
\coordinate (Dext) at ($(D.north east) + (0.25ex, 0.25ex)$);
\node[fit=(Dext)(D)(B)(C2), draw, cblue, dashed] {};
\node[fit=(E2)(A)(B)(C1), draw, dashed] {};
\end{tikzpicture}
\end{adjustbox}
\caption{Unfolded multiview layouts}
\label{fig:graph-res:unrolled-multiview-layout}
\end{subfigure}%
\begin{subfigure}{5cm}
\centering
\begin{adjustbox}{max width=4.5cm}
\begin{tikzpicture}[>=Latex]
\node[circle, draw] (A) {A};
\node[circle, draw, below right=of A] (B) {B};
\node[circle, draw, above right=of A] (C) {C};
\node[circle, draw, below right=of C] (D) {D};
\node[circle, draw, right=of D] (E) {E};

\draw[-{>[length=2.159mm]}, line width=0.432mm] (B) to[out=125, in=-35] (A);
\draw[-{>[length=0.611mm]}, line width=0.122mm] (C) to[out=235, in=35] (A);
\draw[-{>[length=1.682mm]}, line width=0.336mm] (E.270) to[bend left=30] ($(A)!0.5!(E) - (0cm, 2.25cm)$) to[bend left=45] (A.280);
\draw[-{>[length=2.082mm]}, line width=0.416mm] (A) to[out=-55, in=145] (B);
\draw[-{>[length=2.146mm]}, line width=0.429mm] (C) to[out=280, in=80] (B);
\draw[-{>[length=1.498mm]}, line width=0.300mm] (E.245) to[bend left=12.5] (B.-25);
\draw[-{>[length=0.554mm]}, line width=0.111mm] (A) to[out=55, in=215] (C);
\draw[-{>[length=2.273mm]}, line width=0.455mm] (B) to[out=100, in=260] (C);
\draw[-{>[length=1.686mm]}, line width=0.337mm] (E.120) to[bend right=12.5] (C.25);
\draw[-{>[length=2.129mm]}, line width=0.426mm] (A.90) to[bend left=45] ($(A)!0.5!(E) + (0cm, 2.25cm)$) to[bend left=30] (E.100);
\draw[-{>[length=1.666mm]}, line width=0.333mm] (B.10) to[bend right=10] (E.215);
\draw[-{>[length=1.666mm]}, line width=0.333mm] (C.-10) to[bend left=10] (E.155);

\draw[-{>[length=2.463mm]}, line width=0.493mm, cblue] (C) to[out=290, in=70] (B);
\draw[-{>[length=2.318mm]}, line width=0.464mm, cblue] (D) to[out=215, in=55] (B);
\draw[-{>[length=2.543mm]}, line width=0.509mm, cblue] (B) to[out=110, in=250] (C);
\draw[-{>[length=3.639mm]}, line width=0.728mm, cblue] (D) to[out=125, in=-35] (C);
\draw[-{>[length=2.691mm]}, line width=0.538mm, cblue] (B) to[out=35, in=235] (D);
\draw[-{>[length=3.925mm]}, line width=0.785mm, cblue] (C) to[out=-55, in=145] (D);

\draw[-{>[length=1.643mm]}, line width=0.329mm, cred] (D) to[out=-10, in=190] (E);
\draw[-{>[length=1.903mm]}, line width=0.381mm, cred] (E) to[out=170, in=10] (D);

\end{tikzpicture}
\end{adjustbox}
\caption{SLIDE}
\label{fig:graph-res:slide}
\end{subfigure}%
\begin{subfigure}{5cm}
\centering
\begin{adjustbox}{max width=4.5cm}
\begin{tikzpicture}[>=Latex]
\node[circle,draw] (A) {A};
\node[circle,draw,below right=of A] (B) {B};
\node[circle,draw,above right=of A] (C) {C};
\node[circle,draw,below right=of C] (D) {D};
\node[circle,draw,right=of D] (E) {E};

\draw[-{>[length=1.710mm]}, line width=0.342mm] (B) to[out=125, in=-35] (A);
\draw[-{>[length=0.747mm]}, line width=0.149mm] (C) to[out=235, in=35] (A);
\draw[-{>[length=1.299mm]}, line width=0.260mm] (E.270) to[bend left=30] ($(A)!0.5!(E) - (0cm, 2.25cm)$) to[bend left=45] (A.280);
\draw[-{>[length=1.808mm]}, line width=0.362mm] (A) to[out=-55, in=145] (B);
\draw[-{>[length=2.005mm]}, line width=0.401mm] (C) to[out=280, in=80] (B);
\draw[-{>[length=1.215mm]}, line width=0.243mm] (E.245) to[bend left=12.5] (B.-25);
\draw[-{>[length=0.684mm]}, line width=0.137mm] (A) to[out=55, in=215] (C);
\draw[-{>[length=2.045mm]}, line width=0.409mm] (B) to[out=100, in=260] (C);
\draw[-{>[length=1.386mm]}, line width=0.277mm] (E.120) to[bend right=12.5] (C.25);
\draw[-{>[length=1.748mm]}, line width=0.350mm] (A.90) to[bend left=45] ($(A)!0.5!(E) + (0cm, 2.25cm)$) to[bend left=30] (E.100);
\draw[-{>[length=1.359mm]}, line width=0.272mm] (B.10) to[bend right=10] (E.215);
\draw[-{>[length=1.356mm]}, line width=0.271mm] (C.-10) to[bend left=10] (E.155);

\draw[-{>[length=1.951mm]}, line width=0.390mm, cblue] (C) to[out=290, in=70] (B);
\draw[-{>[length=1.953mm]}, line width=0.391mm, cblue] (D) to[out=215, in=55] (B);
\draw[-{>[length=1.997mm]}, line width=0.399mm, cblue] (B) to[out=110, in=250] (C);
\draw[-{>[length=2.939mm]}, line width=0.588mm, cblue] (D) to[out=125, in=-35] (C);
\draw[-{>[length=2.111mm]}, line width=0.422mm, cblue] (B) to[out=35, in=235] (D);
\draw[-{>[length=2.956mm]}, line width=0.591mm, cblue] (C) to[out=-55, in=145] (D);

\draw[-{>[length=1.331mm]}, line width=0.266mm, cred] (D) to[out=-10, in=190] (E);
\draw[-{>[length=1.483mm]}, line width=0.297mm, cred] (E) to[out=170, in=10] (D);

\end{tikzpicture}
\end{adjustbox}
\caption{MOFA (Threshold 0.15)}
\label{fig:graph-res:mofa}
\end{subfigure}

%% file: directed-r2-table-250runs.tex
A & A & 33.46 & \B -2.42 & (4.8) & -6.65 & (3.0, ***) & -7.10 & (3.0, ***) & -8.52 & (6.6, ***) & -11.22 & (7.8, ***) & \multicolumn{2}{c}{---} & 7.36 & (5.4, ***) & \multicolumn{2}{c}{---} & \multicolumn{2}{c}{---} & -3.67 & (2.9, ***) & \multicolumn{2}{c}{---} \\
& B & 19.26 & -3.22 & (4.5) & -4.11 & (2.2, ***) & -4.86 & (2.4, ***) & \B -0.73 & (8.8, ***) & -2.03 & (9.8, *) & \multicolumn{2}{c}{---} & 2.33 & (6.0) & \multicolumn{2}{c}{---} & \multicolumn{2}{c}{---} & -2.16 & (2.1, ***) & \multicolumn{2}{c}{---} \\
& C &  8.21 & -2.75 & (3.5) & -1.60 & (1.6, ***) & \multicolumn{2}{c}{---} & 6.75 & (9.5, ***) & \multicolumn{2}{c}{---} & \multicolumn{2}{c}{---} & -2.09 & (5.6) & \multicolumn{2}{c}{---} & \multicolumn{2}{c}{---} & \B -0.74 & (1.5, ***) & \multicolumn{2}{c}{---} \\
& E & 14.39 & \B 0.60 & (2.2) & -2.35 & (1.9, ***) & -2.35 & (1.9, ***) & 6.06 & (7.3, ***) & 5.31 & (7.7, ***) & \multicolumn{2}{c}{---} & 2.43 & (2.0, ***) & \multicolumn{2}{c}{---} & \multicolumn{2}{c}{---} & -1.40 & (1.8, ***) & \multicolumn{2}{c}{---} \\
\midrule
B & A & 20.29 & \E -1.37 & (4.4) & -3.77 & (2.0, ***) & -4.35 & (2.4, ***) & -3.49 & (8.7, ***) & \E -1.88 & (10.3) & \multicolumn{2}{c}{---} & \BE 0.53 & (5.9) & \multicolumn{2}{c}{---} & \multicolumn{2}{c}{---} & -2.21 & (1.9, ***) & \multicolumn{2}{c}{---} \\
& B & 33.42 & \B 0.61 & (3.0) & -5.24 & (2.3, ***) & -5.49 & (2.4, ***) & -5.13 & (6.6, ***) & -6.44 & (8.6, ***) & \multicolumn{2}{c}{---} & 4.96 & (3.9, ***) & 7.31 & (3.4, ***) & \multicolumn{2}{c}{---} & -3.44 & (2.4, ***) & -4.33 & (2.5, ***) \\
& C & 22.34 & \BE 0.59 & (2.1) & -3.02 & (1.9, ***) & -3.24 & (2.1, ***) & 2.80 & (7.2, ***) & \E 0.81 & (8.0) & \multicolumn{2}{c}{---} & \E -0.88 & (5.8) & 2.28 & (5.0, ***) & \multicolumn{2}{c}{---} & -2.30 & (1.9, ***) & -2.84 & (1.9, ***) \\
& D & 22.37 & \E -0.98 & (4.2) & -3.34 & (2.3, ***) & -3.48 & (2.3, ***) & \E 1.80 & (7.6) & 2.19 & (8.6, **) & \multicolumn{2}{c}{---} & \multicolumn{2}{c}{---} & \BE 0.81 & (5.7) & \multicolumn{2}{c}{---} & \multicolumn{2}{c}{---} & -2.84 & (2.0, ***) \\
& E & 13.25 & \B 0.44 & (1.4) & -1.34 & (1.3, ***) & -1.31 & (1.3, ***) & 10.92 & (7.6, ***) & 11.31 & (8.6, ***) & \multicolumn{2}{c}{---} & 1.73 & (1.4, ***) & \multicolumn{2}{c}{---} & \multicolumn{2}{c}{---} & -1.10 & (1.4, ***) & \multicolumn{2}{c}{---} \\
\midrule
C & A &  7.65 & -0.88 & (3.3) & -1.55 & (1.3, ***) & \multicolumn{2}{c}{---} & 0.92 & (6.5) & \multicolumn{2}{c}{---} & \multicolumn{2}{c}{---} & -2.11 & (5.1, *) & \multicolumn{2}{c}{---} & \multicolumn{2}{c}{---} & \B -0.82 & (1.2, ***) & \multicolumn{2}{c}{---} \\
& B & 22.44 & 0.95 & (2.0) & -3.72 & (1.9, ***) & -4.15 & (2.0, ***) & -5.03 & (5.3, ***) & -1.87 & (7.4) & \multicolumn{2}{c}{---} & \B 0.29 & (4.8, **) & 2.99 & (4.6, ***) & \multicolumn{2}{c}{---} & -1.99 & (1.8, ***) & -2.47 & (1.9, ***) \\
& C & 33.32 & \B 0.84 & (3.4) & -6.49 & (2.5, ***) & -6.57 & (2.5, ***) & -14.65 & (6.0, ***) & -10.02 & (7.2, ***) & \multicolumn{2}{c}{---} & 4.09 & (5.0, ***) & 6.16 & (4.2, ***) & \multicolumn{2}{c}{---} & -3.47 & (2.5, ***) & -3.92 & (2.5, ***) \\
& D & 33.32 & \B -0.54 & (4.7) & \multicolumn{2}{c}{---} & -6.81 & (2.7, ***) & \multicolumn{2}{c}{---} & -10.19 & (7.3, ***) & \multicolumn{2}{c}{---} & \multicolumn{2}{c}{---} & 3.07 & (6.0, ***) & \multicolumn{2}{c}{---} & \multicolumn{2}{c}{---} & -3.93 & (2.5, ***) \\
& E & 14.88 & \BE 0.73 & (1.5) & -1.83 & (1.5, ***) & \multicolumn{2}{c}{---} & \E 1.67 & (5.6) & \multicolumn{2}{c}{---} & \multicolumn{2}{c}{---} & 1.98 & (1.5, ***) & \multicolumn{2}{c}{---} & \multicolumn{2}{c}{---} & \E -1.02 & (1.5) & \multicolumn{2}{c}{---} \\
\midrule
D & B & 23.68 & \E -1.21 & (3.9) & -4.71 & (2.7, ***) & -3.41 & (2.6, ***) & -6.85 & (5.3, ***) & \BE 0.12 & (7.2) & \multicolumn{2}{c}{---} & \multicolumn{2}{c}{---} & 3.23 & (5.6, ***) & \multicolumn{2}{c}{---} & \multicolumn{2}{c}{---} & -2.57 & (2.8, ***) \\
& C & 33.48 & \B -1.36 & (4.7) & \multicolumn{2}{c}{---} & -5.84 & (3.3, ***) & \multicolumn{2}{c}{---} & -7.20 & (7.2, ***) & \multicolumn{2}{c}{---} & \multicolumn{2}{c}{---} & 5.76 & (6.5, ***) & \multicolumn{2}{c}{---} & \multicolumn{2}{c}{---} & -3.93 & (3.3, ***) \\
& D & 33.48 & \B -1.29 & (4.6) & -10.85 & (5.0, ***) & -5.81 & (3.3, ***) & -14.91 & (5.8, ***) & -6.97 & (7.2, ***) & 3.03 & (12.0, **) & \multicolumn{2}{c}{---} & 8.37 & (6.5, ***) & -6.14 & (4.0, ***) & \multicolumn{2}{c}{---} & -3.85 & (3.3, ***) \\
& E & 16.46 & 1.46 & (2.3) & -1.78 & (2.1) & -0.96 & (2.1) & \B 0.37 & (5.2, ***) & 7.33 & (7.2, ***) & 2.57 & (4.8, ***) & \multicolumn{2}{c}{---} & \multicolumn{2}{c}{---} & -1.63 & (2.3) & \multicolumn{2}{c}{---} & \multicolumn{2}{c}{---} \\
\midrule
E & A & 18.84 & 1.21 & (2.3) & -2.46 & (2.0, ***) & -2.25 & (1.9, ***) & \B 0.88 & (7.2, **) & 1.53 & (7.7) & \multicolumn{2}{c}{---} & 2.45 & (2.0, ***) & \multicolumn{2}{c}{---} & \multicolumn{2}{c}{---} & -1.37 & (1.9) & \multicolumn{2}{c}{---} \\
& B & 14.62 & \BE 0.71 & (1.7) & -1.17 & (1.5, *) & -1.57 & (1.6, ***) & 9.60 & (8.9, ***) & 6.14 & (9.5, ***) & \multicolumn{2}{c}{---} & 2.04 & (1.6, ***) & \multicolumn{2}{c}{---} & \multicolumn{2}{c}{---} & \E -1.03 & (1.6) & \multicolumn{2}{c}{---} \\
& C & 14.58 & \BE 0.75 & (1.7) & -1.19 & (1.5, *) & \multicolumn{2}{c}{---} & 8.27 & (8.6, ***) & \multicolumn{2}{c}{---} & \multicolumn{2}{c}{---} & 2.08 & (1.6, ***) & \multicolumn{2}{c}{---} & \multicolumn{2}{c}{---} & \E -1.02 & (1.6) & \multicolumn{2}{c}{---} \\
& D & 14.71 & \B 0.62 & (1.7) & -1.09 & (1.6, *) & -1.49 & (1.6, ***) & 9.51 & (8.9, ***) & 6.05 & (9.5, ***) & 1.72 & (4.0, ***) & \multicolumn{2}{c}{---} & \multicolumn{2}{c}{---} & -1.39 & (1.7, ***) & \multicolumn{2}{c}{---} & \multicolumn{2}{c}{---} \\
& E & 33.27 & 2.04 & (2.7) & -4.13 & (2.6, ***) & -4.54 & (2.6, ***) & \B -1.01 & (4.2, **) & -6.43 & (6.2, ***) & 6.38 & (2.8, ***) & 5.16 & (2.9, ***) & \multicolumn{2}{c}{---} & -3.49 & (2.6, ***) & -2.64 & (2.6) & \multicolumn{2}{c}{---} \\
\bottomrule
\multicolumn{3}{l}{Root mean-squared deviation} & 1.40 & & 4.17 & & 4.42 & & 7.20 & & 6.56 & & 3.85 & & 3.21 & & 5.05 & & 3.69 & & 2.13 & & 3.47 \\